\documentclass{article}
\usepackage{amsmath}
\usepackage{amsthm}
\usepackage{tikz}
\usepackage[ruled,vlined,noend]{algorithm2e}
\usepackage{biblatex}
\usepackage{enumitem}
\usepackage[page]{appendix}
\usepackage[a4paper, top=25mm, bottom=25mm]{geometry}
\usepackage{pgfplots}
\usepackage{multirow}
\usepackage{amsfonts}
\pgfplotsset{compat=1.17, width=7cm, height=6cm}
\newtheorem{theorem}{Theorem}[section]
\newtheorem{lemma}{Lemma}[section]
\theoremstyle{definition}
\newtheorem{defi}{Definition}[section]

\SetCommentSty{mycommfont}

\title{Finding a Lower Bound for k-Unbounded Hamiltonian Cycles}
\author{Albert Jiang}

\begin{document}

\maketitle

\begin{abstract}
Methods to determine the existence of Hamiltonian Cycles in graphs have been extensively studied. However, little research has been done following cases when no Hamiltonian Cycle exists. Let a vertex be ``unbounded" if it is visited more than once in a path. Furthermore, let a k-Unbounded Hamiltonian Cycle be a path with finite length that visits every vertex, has adjacent start and end vertices, and contains k unbounded vertices. We consider a novel variant of the Hamiltonian Cycle Problem in which the objective is to find an m-Unbounded Hamiltonian Cycle where m is the minimum value of k such that a k-Unbounded Hamiltonian Cycle exists. We first consider the task on well-known non-Hamiltonian graphs. We then provide an exponential-time brute-force algorithm for the determination of an m-Unbounded Hamiltonian Cycle and discuss approaches to solve the variant through transformations to the Hamiltonian Cycle Problem and the Asymmetric Traveling Salesman Problem. Finally, we present a polynomial-time heuristic for the determination of an m-Unbounded Hamiltonian Cycle that is also shown to be an effective heuristic for the original Hamiltonian Cycle Problem.
\end{abstract}

\BlankLine

\section{Introduction}
One of the most explored NP-complete problems is the Hamiltonian Cycle Problem (HCP), in which the objective is to identify a simple cycle that passes through every vertex of an undirected graph. The HCP is a subproblem of the more famous Traveling Salesman Problem (TSP), where the objective is to find a Hamiltonian Cycle with minimum length. However, both of these problems are closely related and have attracted large bodies of research.

Due to its NP-complete nature, the HCP can generally be divided into two different approaches: exponential-time algorithms, in which a solution is guaranteed to be found if one exists, and polynomial-time heuristics, in which a solution can be found in the vast majority of cases with greatly reduced runtimes. The Held-Karp Algorithm \cite{Held-Karp} utilizes a dynamic programming approach to solve both the HCP and TSP and has served as a benchmark for the time complexity of subsequent algorithms in general graphs. Björklund \cite{Held-KarpOptimization} proposed a Monte Carlo algorithm that is able to determine Hamiltonicity in a probability exponentially small in the number of vertices with a smaller exponential base than the Held-Karp Algorithm. Martello \cite{Backtrack} proposed an algorithm for the HCP in digraphs utilizing backtracking and low degree vertex selection that greatly optimized the exhaustive search method. Pósa's idea of a \textit{rotation} \cite{Posa-Rotation} and Komlós and Szemerédi's notion of a \textit{cycle extension} \cite{Cycle-Extension} revolutionized HCP heuristics due to their ability to easily modify the end vertices of a path. Examples of heuristics utilizing these techniques include HAM \cite{HAM}, SparseHAM \cite{Sparse-HAM}, and the most recent HybridHAM \cite{HybridHAM}. The Lin-Kernighan Heuristic \cite{Lin-Kernighan} utilizes \textit{k-opt transformations} and is one of the most effective approaches for the TSP. Its complex implementations have also been shown to be efficient HCP solvers \cite{LKH-Implementation}. The idea of k-opt transformations inspired the deterministic Snakes and Ladders Heuristic \cite{SnakesLadders}, which has been shown to reliably find Hamiltonian Cycles in difficult instances. Heuristics have also been developed to determine Hamiltonian Cycles in special graphs, including k-connected graphs \cite{SpecialGraph1}, solid grid graphs \cite{SpecialGraph2}, graphs of bounded tree-width \cite{SpecialGraph3}, etc. 

A less studied area of research explores variations of the Hamiltonian Cycle Problem (excluding the TSP and its variants). In the Minimum Labeling Hamiltonian Cycle Problem (MLHCP), a color is assigned to every edge of a graph and the objective is to find a Hamiltonian Cycle using the minimum number of distinct colors. Jozefowiez et al. proposed a branch-and-cut algorithm \cite{MLHCP2} and Cerulli et al. described several heuristic approaches for this problem \cite{MLHCP}. Another interesting variant is the Parity Hamiltonian Cycle Problem (PHCP), in which the objective is to find a closed walk that visits every vertex an odd number of times allowing for the revisiting of vertices and edges. Nishiyama et al. introduce and provide a comprehensive analysis of the PHCP, including a linear-time algorithm to solve the problem while visiting each edge no greater than 4 times \cite{PHCP}. Other variations include the Rainbow, Minimum Flow Cost, and Multi-Objective Hamiltonian Cycle Problems \cite{rainbowHCP, minimumFlowHCP, multiObjectiveHCP}.

In this paper, we introduce the notion of a \textit{k-Unbounded} Hamiltonian Cycle. For the remainder of the paper, a \textit{path} will allow for edges and vertices to appear multiple times. If an unbounded vertex is a vertex that is visited more than once, then a k-Unbounded Hamiltonian Cycle is a closed path containing exactly k unbounded vertices. In this context we propose the \textit{k-Unbounded Hamiltonian Cycle Problem}, which is to find an \textit{m-Unbounded} Hamiltonian Cycle in an undirected graph where m is the minimum value of k such that a k-Unbounded Hamiltonian Cycle exists. The k-Unbounded HCP is NP-Complete, as it is necessary to determine the existence of a Hamiltonian Cycle (or a ``0-Unbounded" Hamiltonian Cycle)  before finding an m-Unbounded Hamiltonian Cycle. Thus, the k-Unbounded HCP may be considered a more difficult extension of the original HCP. 

We first consider the k-Unbounded HCP in two well-known non-Hamiltonian classes of graphs: subcases of the generalized Petersen graphs and trees. For the former, we analyze the bounds of $m$ and propose two constructions, and for the latter, we establish linear-time algorithms to identify k-Unbounded Hamiltonian Cycles and k-Unbounded Hamiltonian Paths. We then provide an exponential-time brute force algorithm inspired by the Held-Karp Algorithm that utilizes a multi-source BFS to account for the existence of unbounded vertices. We also propose a transformation in which the HCP may be applied to solve the decision variant of the problem and a transformation in which the Asymmetric Traveling Salesman Problem (ATSP), or the TSP with directed edges, may be applied to find an m-Unbounded Hamiltonian Cycle. 

Finally, we propose the deterministic Unbounded Heuristic with an effective cubic time complexity that incorporates rotations and the novel idea of a preemptive cycle check into a shortest path algorithm. Experimental results on randomly generated graphs show that while the heuristic is unable to identify m-Unbounded Hamiltonian Cycles in every instance in sparse graphs, it is able to find m-Unbounded Hamiltonian Cycles in the \textit{vast majority} of cases and with a very low margin of error. We also propose a variation called the Fast Unbounded Heuristic that runs several times faster in random graphs in exchange for suboptimal accuracy. Both heuristics are also shown to be competitive solvers of the original HCP, and we compare both accuracy and runtime to state-of-the-art HCP heuristic(s) through established benchmarks.

Besides addressing an interesting and intuitive variant of the HCP, the proposed variant may also have real-world applications. A potential implementation of the k-Unbounded HCP may be found in minimizing the traffic caused by delivery services. Oftentimes, networks connecting points of interest do not contain Hamiltonian Cycles, so the Unbounded Heuristic may be implemented as a routing algorithm to direct delivery trucks to a minimum number of ``hubs" of traffic so that resources can be efficiently invested into a small number of regions. We suspect that the k-Unbounded HCP may be able to address problems that were never considered simply because a Hamiltonian Cycle does not exist as well as serve as an extension to existing applications of the HCP.    

\section{Preliminaries}
\subsection{Definitions}
\begin{defi}[Unbounded Vertex]
An \textit{unbounded vertex} is a vertex that can appear an arbitrarily large number of times in a path
\end{defi}
\noindent A \textit{bounded} vertex is a vertex that appears exactly once in a path.
\begin{defi}[k-Unbounded Hamiltonian Path]
A \textit{k-Unbounded Hamiltonian Path} is a path with finite length that visits every vertex and contains k unbounded vertices. 
\end{defi}

\begin{defi}[k-Unbounded Hamiltonian Cycle]
A \textit{k-Unbounded Hamiltonian Cycle} is a k-Unbounded Hamiltonian Path with adjacent starting and ending vertices.
\end{defi}

\noindent Throughout the paper, $m$ will refer to the minimum value of k such that a k-Unbounded Hamiltonian Path or Cycle exists.

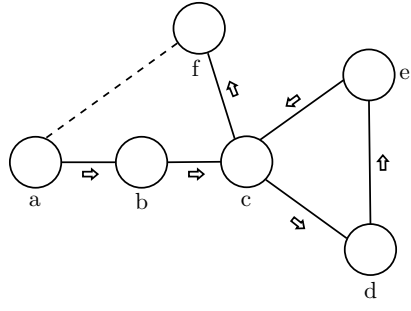
\begin{figure}[htbp] \label{figHcycle}
    \centering
    \tikzset{every picture/.style={line width=0.75pt}} 
\resizebox{5.5cm}{!}{
    \begin{tikzpicture}[x=0.75pt,y=0.75pt,yscale=-1,xscale=1]
    
\draw   (270.29,155) .. controls (270.29,146.94) and (276.83,140.4) .. (284.89,140.4) .. controls (292.95,140.4) and (299.49,146.94) .. (299.49,155) .. controls (299.49,163.07) and (292.95,169.6) .. (284.89,169.6) .. controls (276.83,169.6) and (270.29,163.07) .. (270.29,155) -- cycle ;
\draw   (330.5,155) .. controls (330.5,146.94) and (337.04,140.4) .. (345.1,140.4) .. controls (353.17,140.4) and (359.7,146.94) .. (359.7,155) .. controls (359.7,163.07) and (353.17,169.6) .. (345.1,169.6) .. controls (337.04,169.6) and (330.5,163.07) .. (330.5,155) -- cycle ;
\draw   (390.22,154.72) .. controls (390.22,146.66) and (396.76,140.12) .. (404.82,140.12) .. controls (412.88,140.12) and (419.42,146.66) .. (419.42,154.72) .. controls (419.42,162.78) and (412.88,169.32) .. (404.82,169.32) .. controls (396.76,169.32) and (390.22,162.78) .. (390.22,154.72) -- cycle ;
\draw   (363.23,78.29) .. controls (363.23,70.23) and (369.77,63.69) .. (377.83,63.69) .. controls (385.9,63.69) and (392.43,70.23) .. (392.43,78.29) .. controls (392.43,86.35) and (385.9,92.89) .. (377.83,92.89) .. controls (369.77,92.89) and (363.23,86.35) .. (363.23,78.29) -- cycle ;
\draw   (459.65,105) .. controls (459.65,96.94) and (466.18,90.4) .. (474.25,90.4) .. controls (482.31,90.4) and (488.85,96.94) .. (488.85,105) .. controls (488.85,113.07) and (482.31,119.6) .. (474.25,119.6) .. controls (466.18,119.6) and (459.65,113.07) .. (459.65,105) -- cycle ;
\draw   (460.5,205) .. controls (460.5,196.94) and (467.04,190.4) .. (475.1,190.4) .. controls (483.17,190.4) and (489.7,196.94) .. (489.7,205) .. controls (489.7,213.07) and (483.17,219.6) .. (475.1,219.6) .. controls (467.04,219.6) and (460.5,213.07) .. (460.5,205) -- cycle ;
\draw    (299.49,155) -- (330.5,155) ;
\draw    (359.7,155) -- (390.22,154.72) ;
\draw    (412.43,142.14) -- (459.89,107.89) ;
\draw    (415.67,165) -- (461.67,198.56) ;
\draw    (398,141.5) -- (382.78,92.56) ;
\draw    (474.25,119.6) -- (475.1,190.4) ;
\draw (311.8,160.4) -- (316.84,160.4) -- (316.84,159) -- (320.2,161.8) -- (316.84,164.6) -- (316.84,163.2) -- (311.8,163.2) -- cycle ;
\draw   [dashed] (365.44,86.78) -- (290.33,141.44) ;
\draw   (372,160.4) -- (377.04,160.4) -- (377.04,159) -- (380.4,161.8) -- (377.04,164.6) -- (377.04,163.2) -- (372,163.2) -- cycle ;
\draw   (431.8,186.48) -- (435.73,189.63) -- (436.61,188.54) -- (437.48,192.82) -- (433.11,192.91) -- (433.98,191.82) -- (430.05,188.67) -- cycle ;
\draw   (480.87,159.42) -- (480.79,154.38) -- (479.39,154.41) -- (482.13,151) -- (484.99,154.31) -- (483.59,154.34) -- (483.67,159.38) -- cycle ;
\draw   (435,118.81) -- (430.96,121.82) -- (431.8,122.95) -- (427.43,122.71) -- (428.45,118.46) -- (429.29,119.58) -- (433.33,116.57) -- cycle ;
\draw   (397.08,117.22) -- (395.19,112.55) -- (393.89,113.07) -- (395.22,108.91) -- (399.08,110.97) -- (397.78,111.5) -- (399.67,116.17) -- cycle ;

\draw (279.56,173) node [anchor=north west][inner sep=0.75pt]   [align=left] {a};
\draw (340,171) node [anchor=north west][inner sep=0.75pt]   [align=left] {b};
\draw (400,173) node [anchor=north west][inner sep=0.75pt]   [align=left] {c};
\draw (470,222) node [anchor=north west][inner sep=0.75pt]   [align=left] {d};
\draw (490,100) node [anchor=north west][inner sep=0.75pt]   [align=left] {e};
\draw (372.06,95) node [anchor=north west][inner sep=0.75pt]   [align=left] {f};

\end{tikzpicture}
}
    \caption{$P = (a, b, c, d, e, c, f)$ is a valid 1-Unbounded Hamiltonian Path with $c$ unbounded. Edge $af$ makes $P$ a 1-Unbounded Hamiltonian Cycle}
\end{figure}

\begin{defi}[Tree]
A \textit{tree} is an acyclic graph in which there is a unique path between every pair of vertices.
\end{defi}
\noindent A \textit{leaf} is a vertex of a tree with degree $1$.

\begin{defi}[Component]
A \textit{component} is a maximal induced subgraph of a graph $G$ where there exists a path between any pair of vertices of the subgraph.
\end{defi}
\noindent $C_v$ denotes the component containing vertex $v$ unless specified as otherwise.

\begin{defi}[Cut Vertex]
A \textit{cut vertex} is a vertex that when removed, increases the number of components.
\end{defi}

\begin{defi}[Degree]
A \textit{degree} of a vertex $v$, denoted as $d_v$, will be defined as the number of unvisited neighbours of $v$. 
\end{defi}

\begin{defi}[Distance]
The \textit{distance} between two vertices is the length of the shortest path between them. If the graph is weighted, then the distance refers to the least weighted distance.
\end{defi}

\noindent An undirected graph $G$ is characterized by a set of vertices $V$ and a set of edges $E$. Throughout the paper, $N=|V|$ and $E$ is used interchangeably with $|E|$. 
\subsection{Theorems}
\begin{theorem} \label{thm1}
 If there exists a bounded cut vertex, a k-Unbounded Hamiltonian Cycle does not exist
\end{theorem}

\begin{proof}
We will prove the contrapositive: ``If a k-Unbounded Hamiltonian Cycle exists, then there does not exist a bounded cut vertex." Let the cycle be $x_1, x_2,..., x_n$ and $x_v$ be a bounded cut vertex. Since $x_v$ is bounded, if vertex $x_v$ is removed, it will be the only instance removed. The remaining vertices will remain connected through $x_{v+1},..., x_n, x_1,..., x_{v-1}$. Since a cut vertex must increase the number of components when removed, there is a contradiction and bounded cut vertices cannot exist. 
\end{proof}

\begin{theorem} \label{thm2}
 If there exists a bounded cut vertex $v$ that when removed, results in k components where $k\ge3$, a k-Unbounded Hamiltonian Path does not exist
\end{theorem}

\begin{proof}
Without loss of generality, denote the components and the order that they are visited for the first time as $q_1, q_2, q_3...$ As $v$ is a cut vertex, it is necessary to visit $v$ when transitioning from $q_n \rightarrow q_{n+1}$, otherwise, they would be part of the same component. Since there are at least $2$ transitions from $q_1 \rightarrow q_2$ and $q_2 \rightarrow q_3$, $v$ is visited more than once, contradicting the definition of a bounded vertex.
\end{proof}

\subsection{Graph Traversal Algorithms}
\subsubsection{Breadth First Search (BFS)} \label{BFS}
If a BFS begins at a start vertex $S$, a vertex will only be visited until all vertices with a smaller distance from $S$ are visited. In unweighted graphs, a BFS can be used to find the distance from one vertex to all other vertices, and Dijkstra's \cite{Dijkstras}, a variation of BFS, is similarly able to find the shortest path from a vertex in weighted graphs by prioritizing the extension from vertices with the smallest distance from $S$. 

In a \textbf{multisource-BFS}, multiple starting vertices are initially pushed into $Q$: this can also be viewed as creating a new vertex $z$ adjacent to the set of starting vertices and running a BFS from $z$. If $d[v] \ne K$, $v$ can be reached by one of the starting vertices. Furthermore, $d[v]$ is the shortest distance of any starting vertex to $v$. 

Every vertex is visited exactly once, since once $v$ is pushed into $Q$, d[v] is the final distance between $S$ and $v$ and there exists no path with distance less than d[v]. All edges to adjacent vertices are considered at every visit, so the time complexity is $O(N+E) = O(E)$.

\subsubsection{Depth First Search (DFS)}
A DFS traverses a graph by beginning at a root node and traveling as long as possible until a vertex is reached with no unvisited neighbours. The search then backtracks to an ancestor and continues the process. A DFS is typically implemented with recursion alongside a visited array. Similar to the BFS, the time complexity is O(E).

\section{Non-Hamiltonian Graphs}
\subsection{Generalized Petersen Graphs}
The generalized Petersen graph $G(n, k)$ for integers $n$ and $k$ and $2\le 2k < n$ is defined by Watkins' \cite{generalized-Petersen} as a graph with vertex set $V(G(n, k))=\{u_0, u_1,..., u_{n-1}, v_0, v_1,..., v_{n-1}\}$ and the edge-set consisting of all edges in the form $u_{i}u_{i+1}$, $u_{i}v_{i}$, $v_{i}v_{i+k}$, with subscripts reduced modulo $n$. There are two well-known non-Hamiltonian instances of generalized Petersen graphs: $\{n \equiv 5 \pmod{6}, k=2\}$ or $\{4 | n, n\ge 8, k=\frac{n}{2}\}$. The former is known to be hypohamiltonian, or a non-Hamiltonian graph such that $G-v$ is Hamiltonian for $v\in V$. Consider the following lemma:

\begin{lemma} \label{lemma1-Unbounded}
For a graph G, if there exists a vertex v such that $G - v$ contains a Hamiltonian Cycle, then G contains a 1-Unbounded Hamiltonian Cycle 
\end{lemma}

\begin{proof}
Consider the Hamiltonian Cycle $P=(a_0, a_1,..., a_k)$ in $G-v$. Since $G$ is connected, there exists a neighbour $a_i$ of $v$ on $P$. It follows that the path $P'=(a_0, a_1,..., a_i, v, a_i,..., a_k)$ exists in $G$. Since all vertices are visited and $a_i$ is the only vertex that appears more than once, $P'$ is a 1-Unbounded Hamiltonian Cycle.
\end{proof}

Since the condition for the existence of a 1-Unbounded Hamiltonian Cycle in Lemma \ref{lemma1-Unbounded} is a relaxed version of hypohamiltonicity, all hypohamiltonian graphs have $m=1$ as the lower bound. For the second set of non-Hamiltonian graphs, I claim that $m\ge 2$, or equivalently, that there does not exist a 1-Unbounded Hamiltonian Cycle.

\begin{figure}[htbp] \label{figGP}
    \centering

    \tikzset{every picture/.style={line width=0.75pt}} 
    \resizebox{4.5cm}{!}{
    \begin{tikzpicture}[x=0.75pt,y=0.75pt,yscale=-1,xscale=1]
    
\draw   (220.2,69.81) .. controls (220.2,64.41) and (224.58,60.03) .. (229.99,60.03) .. controls (235.39,60.03) and (239.77,64.41) .. (239.77,69.81) .. controls (239.77,75.22) and (235.39,79.6) .. (229.99,79.6) .. controls (224.58,79.6) and (220.2,75.22) .. (220.2,69.81) -- cycle ;
\draw   (280.61,69.98) .. controls (280.61,64.58) and (284.99,60.2) .. (290.4,60.2) .. controls (295.8,60.2) and (300.18,64.58) .. (300.18,69.98) .. controls (300.18,75.39) and (295.8,79.77) .. (290.4,79.77) .. controls (284.99,79.77) and (280.61,75.39) .. (280.61,69.98) -- cycle ;
\draw   (180.4,109.61) .. controls (180.4,104.21) and (184.78,99.83) .. (190.19,99.83) .. controls (195.59,99.83) and (199.97,104.21) .. (199.97,109.61) .. controls (199.97,115.02) and (195.59,119.4) .. (190.19,119.4) .. controls (184.78,119.4) and (180.4,115.02) .. (180.4,109.61) -- cycle ;
\draw   (180.07,169.95) .. controls (180.07,164.54) and (184.45,160.16) .. (189.85,160.16) .. controls (195.26,160.16) and (199.64,164.54) .. (199.64,169.95) .. controls (199.64,175.35) and (195.26,179.73) .. (189.85,179.73) .. controls (184.45,179.73) and (180.07,175.35) .. (180.07,169.95) -- cycle ;
\draw   (219.73,209.28) .. controls (219.73,203.87) and (224.12,199.49) .. (229.52,199.49) .. controls (234.93,199.49) and (239.31,203.87) .. (239.31,209.28) .. controls (239.31,214.68) and (234.93,219.07) .. (229.52,219.07) .. controls (224.12,219.07) and (219.73,214.68) .. (219.73,209.28) -- cycle ;
\draw   (280.07,210.08) .. controls (280.07,204.67) and (284.45,200.29) .. (289.85,200.29) .. controls (295.26,200.29) and (299.64,204.67) .. (299.64,210.08) .. controls (299.64,215.48) and (295.26,219.87) .. (289.85,219.87) .. controls (284.45,219.87) and (280.07,215.48) .. (280.07,210.08) -- cycle ;
\draw   (320.73,169.95) .. controls (320.73,164.54) and (325.12,160.16) .. (330.52,160.16) .. controls (335.93,160.16) and (340.31,164.54) .. (340.31,169.95) .. controls (340.31,175.35) and (335.93,179.73) .. (330.52,179.73) .. controls (325.12,179.73) and (320.73,175.35) .. (320.73,169.95) -- cycle ;
\draw   (320.07,109.95) .. controls (320.07,104.54) and (324.45,100.16) .. (329.85,100.16) .. controls (335.26,100.16) and (339.64,104.54) .. (339.64,109.95) .. controls (339.64,115.35) and (335.26,119.73) .. (329.85,119.73) .. controls (324.45,119.73) and (320.07,115.35) .. (320.07,109.95) -- cycle ;
\draw   (236.4,100.95) .. controls (236.4,95.54) and (240.78,91.16) .. (246.19,91.16) .. controls (251.59,91.16) and (255.97,95.54) .. (255.97,100.95) .. controls (255.97,106.35) and (251.59,110.73) .. (246.19,110.73) .. controls (240.78,110.73) and (236.4,106.35) .. (236.4,100.95) -- cycle ;
\draw   (268.9,100.7) .. controls (268.9,95.29) and (273.28,90.91) .. (278.69,90.91) .. controls (284.09,90.91) and (288.47,95.29) .. (288.47,100.7) .. controls (288.47,106.1) and (284.09,110.48) .. (278.69,110.48) .. controls (273.28,110.48) and (268.9,106.1) .. (268.9,100.7) -- cycle ;
\draw   (216.07,125.61) .. controls (216.07,120.21) and (220.45,115.83) .. (225.85,115.83) .. controls (231.26,115.83) and (235.64,120.21) .. (235.64,125.61) .. controls (235.64,131.02) and (231.26,135.4) .. (225.85,135.4) .. controls (220.45,135.4) and (216.07,131.02) .. (216.07,125.61) -- cycle ;
\draw   (214.4,159.28) .. controls (214.4,153.87) and (218.78,149.49) .. (224.19,149.49) .. controls (229.59,149.49) and (233.97,153.87) .. (233.97,159.28) .. controls (233.97,164.68) and (229.59,169.07) .. (224.19,169.07) .. controls (218.78,169.07) and (214.4,164.68) .. (214.4,159.28) -- cycle ;
\draw   (238.33,179.28) .. controls (238.33,173.87) and (242.72,169.49) .. (248.12,169.49) .. controls (253.53,169.49) and (257.91,173.87) .. (257.91,179.28) .. controls (257.91,184.68) and (253.53,189.07) .. (248.12,189.07) .. controls (242.72,189.07) and (238.33,184.68) .. (238.33,179.28) -- cycle ;
\draw   (273.4,179.61) .. controls (273.4,174.21) and (277.78,169.83) .. (283.19,169.83) .. controls (288.59,169.83) and (292.97,174.21) .. (292.97,179.61) .. controls (292.97,185.02) and (288.59,189.4) .. (283.19,189.4) .. controls (277.78,189.4) and (273.4,185.02) .. (273.4,179.61) -- cycle ;
\draw   (292.07,156.95) .. controls (292.07,151.54) and (296.45,147.16) .. (301.85,147.16) .. controls (307.26,147.16) and (311.64,151.54) .. (311.64,156.95) .. controls (311.64,162.35) and (307.26,166.73) .. (301.85,166.73) .. controls (296.45,166.73) and (292.07,162.35) .. (292.07,156.95) -- cycle ;
\draw   (290.73,126.95) .. controls (290.73,121.54) and (295.12,117.16) .. (300.52,117.16) .. controls (305.93,117.16) and (310.31,121.54) .. (310.31,126.95) .. controls (310.31,132.35) and (305.93,136.73) .. (300.52,136.73) .. controls (295.12,136.73) and (290.73,132.35) .. (290.73,126.95) -- cycle ;
\draw    (239.77,69.81) -- (280.61,69.98) ;
\draw    (298,76.3) -- (325.33,101.3) ;
\draw    (220.33,72.97) -- (194,100.3) ;
\draw    (190.19,119.4) -- (189.85,160.16) ;
\draw    (189.85,179.73) -- (219.73,209.28) ;
\draw    (239.31,209.28) -- (280.07,210.08) ;
\draw    (299.64,210.08) -- (330.52,179.73) ;
\draw    (329.85,119.73) -- (330.52,160.16) ;
\draw    (234.67,78.63) -- (243.33,91.97) ;
\draw    (199.33,113.3) -- (217,120.3) ;
\draw    (199.33,166.97) -- (214.33,161.97) ;
\draw    (233.67,200.77) -- (243.29,188.2) ;
\draw    (289.85,200.29) -- (286.27,188.83) ;
\draw    (322.27,165.83) -- (311.05,160.78) ;
\draw    (321.05,113.03) -- (308.3,120.78) ;
\draw    (287.8,79.78) -- (283.55,91.78) ;
\draw    (250.05,109.48) -- (280.55,169.98) ;
\draw    (277.05,109.78) -- (253.44,171.27) ;
\draw    (290.55,129.78) -- (233.05,154.38) ;
\draw    (235.3,127.13) -- (293.05,153.38) ;

\draw (242.29,175.4) node [anchor=north west][inner sep=0.75pt]  [font=\scriptsize] [align=left] {$\displaystyle v_{0}$};
\draw (277.89,176.2) node [anchor=north west][inner sep=0.75pt]  [font=\scriptsize] [align=left] {$\displaystyle v_{1}$};
\draw (295.69,152.8) node [anchor=north west][inner sep=0.75pt]  [font=\scriptsize] [align=left] {$\displaystyle v_{2}$};
\draw (294.69,122.8) node [anchor=north west][inner sep=0.75pt]  [font=\scriptsize] [align=left] {$\displaystyle v_{3}$};
\draw (272.49,96.4) node [anchor=north west][inner sep=0.75pt]  [font=\scriptsize] [align=left] {$\displaystyle v_{4}$};
\draw (240.29,96.6) node [anchor=north west][inner sep=0.75pt]  [font=\scriptsize] [align=left] {$\displaystyle v_{5}$};
\draw (219.29,121.6) node [anchor=north west][inner sep=0.75pt]  [font=\scriptsize] [align=left] {$\displaystyle v_{6}$};
\draw (218.5,155) node [anchor=north west][inner sep=0.75pt]  [font=\scriptsize] [align=left] {$\displaystyle v_{7}$};
\draw (223,205.2) node [anchor=north west][inner sep=0.75pt]  [font=\scriptsize] [align=left] {$\displaystyle u_{0}$};
\draw (284,206.24) node [anchor=north west][inner sep=0.75pt]  [font=\scriptsize] [align=left] {$\displaystyle u_{1}$};
\draw (324.09,165.84) node [anchor=north west][inner sep=0.75pt]  [font=\scriptsize] [align=left] {$\displaystyle u_{2}$};
\draw (323.49,106.24) node [anchor=north west][inner sep=0.75pt]  [font=\scriptsize] [align=left] {$\displaystyle u_{3}$};
\draw (284,66) node [anchor=north west][inner sep=0.75pt]  [font=\scriptsize] [align=left] {$\displaystyle u_{4}$};
\draw (223,65.76) node [anchor=north west][inner sep=0.75pt]  [font=\scriptsize] [align=left] {$\displaystyle u_{5}$};
\draw (184,105.16) node [anchor=north west][inner sep=0.75pt]  [font=\scriptsize] [align=left] {$\displaystyle u_{6}$};
\draw (184,166) node [anchor=north west][inner sep=0.75pt]  [font=\scriptsize] [align=left] {$\displaystyle u_{7}$};

    \end{tikzpicture}
    }
    \caption{G(8, 4)}
\end{figure}

If there exists a 1-Unbounded Hamiltonian Cycle, the unbounded vertex is either $v_i$ or $u_i$. We first consider the case of $v_i$. Without loss of generality, let the unbounded vertex be $v_0$. Furthermore, let $P$ be a 1-Unbounded Hamiltonian Cycle such that the start vertex is not adjacent to $u_0$ and $v_{\frac{n}{2}}$ (any such cycle may always be rotated to satisfy this condition). It follows that $v_0$ may not be the last vertex. Thus, the first instance of $v_0$ in $P$ is either $(...u_0, v_0, v_{\frac{n}{2}}, x...)$ or $(...v_{\frac{n}{2}}, v_0, u_0, x...)$. If $x$ is $v_0$, it is impossible to continue the path since all of the neighbours of $v_0$ have already been visited. For any other vertex $x$, $v_0$ cannot be unbounded since it has only been visited once but all of its neighbours have already been visited. Thus, there is a contradiction and $v_i$ cannot be unbounded.

Similarly, without loss of generality, let $u_0$ be unbounded and let $P$ be a 1-Unbounded Hamiltonian Cycle such that the start vertex is not adjacent to $v_0, u_1$, and $u_{n-1}$ ($u_0$ cannot be the last vertex). For some permutation $(p_1, p_2, p_3)$ of these three vertices, there exists two possibilities for $P$: $(...p_1, u_0, p_2,..., p_3, u_0,...)$ or $(...p_1, u_0, p_2, u_0, p_3,...)$. The former is not possible since there are no unvisited neighbours of $u_0$ after its second appearance. If the latter is rotated to the path $(p_3,..., p_1, u_0, p_2, u_0)$, then $P$ is a 1-Unbounded Hamiltonian Cycle if and only if there exists a Hamiltonian Path from $p_3$ to $p_1$ in the graph $G-\{u_0, p_2\}$. If $p_2=u_1$, a Hamiltonian Path does not exist from $v_0$ to $u_{n-1}$ in the resulting graph since there exists the sequence $(...,v_{\frac{n}{2}+1}, v_1, v_{\frac{n}{2}+1}...)$ in $P$, making $v_{\frac{n}{2}+1}$ unbounded. A similar argument applies when $p_2=u_{n-1}$. If $p_2=v_0$, a Hamiltonian Path does not exist from $u_1$ to $u_{n-1}$ since there exists the sequence $(...,u_{\frac{n}{2}}, v_{\frac{n}{2}}, u_{\frac{n}{2}},...)$ in $P$, making $u_{\frac{n}{2}}$ unbounded. Thus, there does not exist a vertex $v_i$ or $u_i$ that is unbounded and $m\ge 2$.

It remains to show some upper bound for $m$. We propose the following construction: the path begins at $u_i$ for $i=0$, travels to $u_{i+4\pmod{n}}$ along the shortest path, then travels to $u_{i+5\pmod{n}}$. This process is repeated until the last unvisited vertex $u_{\frac{n}{2}-1}$ is reached. If $n=8$, the last vertex is $u_3$ and a 2-Unbounded Hamiltonian Cycle may be formed by appending $u_2$ and $u_1$ to the path. Otherwise, the path continues by traveling along the shortest path to $u_{n-1}$, which results in a 3-Unbounded Hamiltonian Cycle with vertices $v_{\frac{n}{2}-1}, v_{n-1}$, and $u_{n-1}$ unbounded. It follows that $m=2$ for $n=8$ and $2\le m\le 3$ for $n>8$.

\subsection{Trees}
We consider the problem on trees. When $n=2$, no unbounded vertices are necessary, Otherwise, the tree is rooted at a vertex $R$ with $d_R\ge2$. Furthermore, let the simple path $P$ between the start and end vertices be $(v_0, v_1,..., v_k)$. The problem can be rephrased as the maximization of the number of bounded vertices in a tree where every vertex is unbounded. 
\begin{lemma} \label{lemma1}
An m-Unbounded Hamiltonian Cycle in a tree will have every leaf unbounded
\end{lemma}

\begin{proof}
After a leaf $l$ is first visited in an m-Unbounded Hamiltonian Path or Cycle, the only option for the next vertex in the path is the parent of $l$ (unless $l$ is the last vertex, in which case $l$ is visited once). Since there are no vertices that have $l$ as a parent from the definition of a leaf, there is no need to travel back to $l$ to reach unvisited vertices and it is always possible for $l$ to be visited one time. Unbounded vertices are only necessary for vertices that are visited more than once, meaning that any leaf can be bounded.
\end{proof}

\begin{lemma} \label{lemma2}
The number of components in a tree when a vertex $v$ is removed is $d_v$
\end{lemma}
\begin{proof}
Removing $v$ can be viewed as first removing all edges to neighbours of $v$ before removing $v$. Consider a neighbour $k$ of $v$. The unique path from $v$ to $k$ is trivially ($v$, $k$). However, if edge $vk$ is removed, there is no longer a path connecting $v$ and $k$. This signifies the formation of a new component $C_k$. Repeating this process to all neighbours of $v$ results in $d_v+1$ components, and since $v$ is removed at the end, there will be $d_v$ components in the end.
\end{proof}
\subsubsection{m-Unbounded Hamiltonian Cycle}
As shown in Lemma \ref{lemma1}, every leaf may be bounded. Every other vertex $v$ has $d_v\ge 2$, otherwise $v$ would be a leaf. Lemma \ref{lemma2} shows that $v$ must be a cut vertex, since the number of components increases from $1$ to $k$ for $k\ge 2$ when $v$ is removed. It follows from theorem \ref{thm1} that $v$ must be unbounded. Thus, the optimal arrangement is to have every vertex that is not a leaf as unbounded. A DFS (i.e Euler Tour) from any vertex is sufficient to recreate the m-Unbounded Hamiltonian Cycle, so the time complexity will be $O(N)$.

\subsubsection{m-Unbounded Hamiltonian Path}

\BlankLine

\begin{lemma} \label{lemma3}
The only unbounded vertices in an m-Unbounded Hamiltonian Path are those with $degree\le 2$ on $P$ as well as every leaf. 
\end{lemma}
\begin{proof}
 It follows from Theorem \ref{thm2} that vertices on $P$ with degree $\ge$ 3 are unbounded. Otherwise, every other vertex in $P$ has degree $1$ or $2$. If a vertex has degree $1$, it is a leaf and can always be bounded from Lemma \ref{lemma1}. Consider a non-endpoint vertex $v_x$ with degree 2. Every vertex excluding end vertices in a simple path are connected to two distinct vertices, and since $v_x$ is not an end vertex and is adjacent to two vertices, both of these vertices are in P. Before visiting $v_{x}$ for the first time, it is always optimal to visit all the vertices of $C_{v_{x-1}}$ in $G-v_x$. Any configuration that doesn't do so may be rearranged such that all instances of visits in this component are moved to a continuous segment. It follows that that $v_{x+1} \ne v_{x-1}$. However, since $v_x$ is now visited, all vertices of $C_{v_x}$ in $G-v_{x+1}$ are visited and it is unnecessary to ever visit a vertex in this component after traveling to $v_{x+1}$. Thus, $v_x$ is visited once and is bounded. 

It remains to show that the only bounded vertices not on $P$ are leaves. Recall that all leaves are bounded from Lemma \ref{lemma1}. Let $v$ be a non-leaf vertex not on $P$, $l$ be the vertex directly before $v$ in the m-Unbounded Hamiltonian Path, and $v_p$ be the most recently visited vertex in $P$. Since $v$ is not a leaf, it is optimal to continue to visit all the vertices of $C_v$ in $G-l$: any configuration that doesn't may be rearranged so that all instances of visits in this component are moved to a continuous segment. The m-Unbounded Hamiltonian Path must also return to a vertex on $P$ to eventually reach $v_k$, signifying that such a vertex must be part of $C_v$ in order for $v$ to not be revisited. Let $v_c$ be such a vertex with the smallest distance from $v$. Due to the nature of a tree, there exist simple paths from $v$ to $v_c$, $v_c$ to $v_p$, and $v_p$ to $v$ such that all paths excluding endpoints are mutually exclusive. As a result, there exists a cycle $(v,...,v_c,...,v_p,...,l)$, which contradicts the definition of a tree. Thus, there cannot exist a vertex of $P$ in the component. Instead, the only other way to reach a vertex in $P$ is to ``exit" the component through $v$. Since $v$ is visited more than once, it is unbounded. 
\end{proof}

\begin{lemma} \label{lemma4}
An m-Unbounded Hamiltonian Path may always start and end at leaves
\end{lemma}
\begin{proof}
Let $v_0$ be an end that is not a leaf. Since $v_0$ has exactly one neighbour on $P$ and $d_{v_0} \ge 2$, there exists a neighbour $x$ of $v_0$ that is not on $P$. If $x$ is inserted at the beginning of $P$, the number of bounded vertices in this new simple path is $\ge P$. The number of such ``extensions" may be repeated until $d_{v_0} < 2$, signifying that $v_0$ is a leaf. A similar process applies to $v_k$ if it is also not a leaf. The number of extensions is finite since $|P| \le |G|$, otherwise there would be a cycle, contradicting the definition of a tree. 
\end{proof}

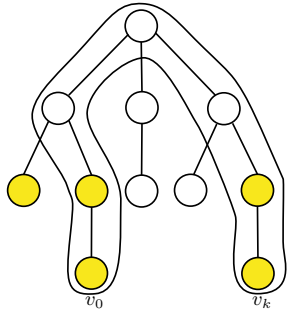
\begin{figure}[htbp] \label{figTree}
    \centering
    \tikzset{every picture/.style={line width=0.75pt}} 
\resizebox{4cm}{!}{
    \begin{tikzpicture}[x=0.75pt,y=0.75pt,yscale=-1,xscale=1]
\draw   (250.52,30.02) .. controls (250.52,24.68) and (254.85,20.36) .. (260.19,20.36) .. controls (265.53,20.36) and (269.86,24.68) .. (269.86,30.02) .. controls (269.86,35.36) and (265.53,39.69) .. (260.19,39.69) .. controls (254.85,39.69) and (250.52,35.36) .. (250.52,30.02) -- cycle ;
\draw   (200.52,80.02) .. controls (200.52,74.68) and (204.85,70.36) .. (210.19,70.36) .. controls (215.53,70.36) and (219.86,74.68) .. (219.86,80.02) .. controls (219.86,85.36) and (215.53,89.69) .. (210.19,89.69) .. controls (204.85,89.69) and (200.52,85.36) .. (200.52,80.02) -- cycle ;
\draw   (250.97,79.58) .. controls (250.97,74.24) and (255.29,69.91) .. (260.63,69.91) .. controls (265.97,69.91) and (270.3,74.24) .. (270.3,79.58) .. controls (270.3,84.92) and (265.97,89.24) .. (260.63,89.24) .. controls (255.29,89.24) and (250.97,84.92) .. (250.97,79.58) -- cycle ;
\draw   (300.3,80.24) .. controls (300.3,74.91) and (304.63,70.58) .. (309.97,70.58) .. controls (315.31,70.58) and (319.63,74.91) .. (319.63,80.24) .. controls (319.63,85.58) and (315.31,89.91) .. (309.97,89.91) .. controls (304.63,89.91) and (300.3,85.58) .. (300.3,80.24) -- cycle ;
\draw  [fill={rgb, 255:red, 248; green, 231; blue, 28 }  ,fill opacity=1 ] (179.72,129.84) .. controls (179.72,124.51) and (184.05,120.18) .. (189.39,120.18) .. controls (194.73,120.18) and (199.06,124.51) .. (199.06,129.84) .. controls (199.06,135.18) and (194.73,139.51) .. (189.39,139.51) .. controls (184.05,139.51) and (179.72,135.18) .. (179.72,129.84) -- cycle ;
\draw  [fill={rgb, 255:red, 248; green, 231; blue, 28 }  ,fill opacity=1 ] (220.66,130.29) .. controls (220.66,124.95) and (224.98,120.62) .. (230.32,120.62) .. controls (235.66,120.62) and (239.99,124.95) .. (239.99,130.29) .. controls (239.99,135.63) and (235.66,139.96) .. (230.32,139.96) .. controls (224.98,139.96) and (220.66,135.63) .. (220.66,130.29) -- cycle ;
\draw  [fill={rgb, 255:red, 248; green, 231; blue, 28 }  ,fill opacity=1 ] (220.52,180.11) .. controls (220.52,174.77) and (224.85,170.44) .. (230.19,170.44) .. controls (235.53,170.44) and (239.86,174.77) .. (239.86,180.11) .. controls (239.86,185.45) and (235.53,189.78) .. (230.19,189.78) .. controls (224.85,189.78) and (220.52,185.45) .. (220.52,180.11) -- cycle ;
\draw   (250.74,130.24) .. controls (250.74,124.91) and (255.07,120.58) .. (260.41,120.58) .. controls (265.75,120.58) and (270.08,124.91) .. (270.08,130.24) .. controls (270.08,135.58) and (265.75,139.91) .. (260.41,139.91) .. controls (255.07,139.91) and (250.74,135.58) .. (250.74,130.24) -- cycle ;
\draw   (280.24,129.74) .. controls (280.24,124.4) and (284.56,120.07) .. (289.9,120.07) .. controls (295.24,120.07) and (299.57,124.4) .. (299.57,129.74) .. controls (299.57,135.08) and (295.24,139.4) .. (289.9,139.4) .. controls (284.56,139.4) and (280.24,135.08) .. (280.24,129.74) -- cycle ;
\draw  [fill={rgb, 255:red, 248; green, 231; blue, 28 }  ,fill opacity=1 ] (320.52,129.74) .. controls (320.52,124.4) and (324.85,120.07) .. (330.19,120.07) .. controls (335.53,120.07) and (339.86,124.4) .. (339.86,129.74) .. controls (339.86,135.08) and (335.53,139.4) .. (330.19,139.4) .. controls (324.85,139.4) and (320.52,135.08) .. (320.52,129.74) -- cycle ;
\draw  [color={rgb, 255:red, 0; green, 0; blue, 0 }  ,draw opacity=1 ][fill={rgb, 255:red, 248; green, 231; blue, 28 }  ,fill opacity=1 ][line width=0.75]  (320.81,180.31) .. controls (320.81,174.97) and (325.14,170.64) .. (330.47,170.64) .. controls (335.81,170.64) and (340.14,174.97) .. (340.14,180.31) .. controls (340.14,185.65) and (335.81,189.97) .. (330.47,189.97) .. controls (325.14,189.97) and (320.81,185.65) .. (320.81,180.31) -- cycle ;
\draw    (252.43,35.14) -- (216.56,72) ;
\draw    (260.19,39.69) -- (260.63,69.91) ;
\draw    (268.78,34.22) -- (304.78,72) ;
\draw    (260.63,89.24) -- (260.41,120.58) ;
\draw    (305.67,89.56) -- (291.89,120) ;
\draw    (315.22,89.11) -- (330.19,120.07) ;
\draw    (330.19,139.4) -- (330.47,170.64) ;
\draw    (204.33,88) -- (189.39,120.18) ;
\draw    (216.56,87.33) -- (230.32,120.62) ;
\draw    (230.32,139.96) -- (230.19,170.44) ;
\draw   (232.71,192.57) .. controls (214.94,194.31) and (214.43,185.71) .. (217.29,158.57) .. controls (220.14,131.43) and (219,138) .. (215.57,127.14) .. controls (212.14,116.29) and (192.3,95.19) .. (194.71,82.29) .. controls (197.13,69.38) and (210.74,58.75) .. (226.56,40.89) .. controls (242.37,23.03) and (247.33,15.61) .. (261.33,16.5) .. controls (275.33,17.39) and (282.56,28.49) .. (298.11,43.56) .. controls (313.66,58.62) and (315.49,59.93) .. (322.33,71.78) .. controls (329.18,83.63) and (344.83,121.33) .. (345.59,122.71) .. controls (346.35,124.08) and (345.14,157.11) .. (345.82,164.59) .. controls (346.51,172.07) and (353,192.57) .. (330.14,192.57) .. controls (307.29,192.57) and (319.35,154.47) .. (316.88,145.29) .. controls (314.41,136.12) and (299.25,98.26) .. (293.89,89.78) .. controls (288.53,81.29) and (273.5,48.68) .. (261.44,49.33) .. controls (249.39,49.99) and (229.14,67.84) .. (227.67,81.78) .. controls (226.19,95.72) and (248.47,115.92) .. (247.89,136.67) .. controls (247.3,157.42) and (250.49,190.83) .. (232.71,192.57) -- cycle ;

\draw (225,193) node [anchor=north west][inner sep=0.75pt]   [align=left] {{\small $v_0$}};
\draw (325.5,193) node [anchor=north west][inner sep=0.75pt]   [align=left] {{\small $v_k$}};

\end{tikzpicture}
}
    \caption{Optimal configuration of bounded (shaded) vertices in a tree}
\end{figure}
Lemmas \ref{lemma3} and \ref{lemma4} lend a straightforward approach to find an m-Unbounded Hamiltonian Path. The algorithm will find a simple path with the largest number of vertices with degree $\le 2$. If the number of such vertices in the simple path is $k$, the exact number of unbounded vertices is $n-(nLeaves+k-2)$. The $-2$ in the second expression ensures the start and end vertices aren't overcounted, since $d_{v_0}, d_{v_k} \le 2$ from Lemma \ref{lemma4}.

Let $dp[v]$ represent the maximum number of bounded vertices on a simple path with $v$ as an endpoint and the other endpoint in the sub tree rooted at v (it is easy to prove that the second endpoint is always a leaf). $dp[v]$ can be calculated by taking the maximum dp[] value of its children and adding $1$ if $d_v$ $\le$ $2$, since every such path can be split into $v$ and a path beginning in the sub tree of one of its children.

For every simple path in a tree, there exists a vertex $l$ with minimum depth. If $l$ is fixed, the maximum number of bounded vertices $M_l$ over the simple paths containing $l$ is the sum of the two largest dp[] values of the children of $l$ with an additional unbounded vertex if $d_l$ $\le$ $2$. It follows that the final number of unbounded vertices in an m-Unbounded Hamiltonian Cycle can be found be taking the maximum $M_l$ for all possible vertices $l$. 

The code below finds an exact set of vertices that are unbounded by additionally storing the child of the vertex chosen in the dp array and utilizing backtracking. All values are initialized to zero, and if $mark[v]=1$, $v$ is bounded. 

\begin{algorithm}[!htbp]
\footnotesize
\DontPrintSemicolon
\SetKwInput{KwInput}{Input} 

\KwInput{Adjacency list $adj[N][]$}

\SetKwFunction{FMain}{Main}
\SetKwFunction{FDfs}{Dfs}
\SetKwProg{Fn}{void}{:}{}

\BlankLine
\Fn{\FDfs{$p$, $par$}}{
    M $\leftarrow$ (0, 0)\;
    \For{k in adj[p]}{
        \lIf{$k = par$}{continue}
        Dfs(k, p)\;
        dp[p] $\leftarrow$ max(dp[p], (dp[k].f, k))\;
        \uIf{$dp[k] > dp[M.f]$}{
            M.s $\leftarrow$ M.f\;
            M.f $\leftarrow$ k\;
        }
        \lElseIf{$dp[k] > dp[M.s]$}{M.s $\leftarrow$ k}
    }
    dp[p].f += (adj[p].size() $\le$ 2 ? 1 : 0)\;
    val $\leftarrow$ dp[M.f].f + dp[M.s].f + (adj[p].size() $\le$ 2 ? 1 : 0)\;
    \uIf{$val > res$}{
        fRoot $\leftarrow$ p\;
        F[0] $\leftarrow$ M.f, F[1] $\leftarrow$ M.s\;
        res $\leftarrow$ val\;
    }
}

\BlankLine
\Fn{\FMain{}}{
    \lIf{$N = 2$}{return}
    \For{$i\gets1$ \KwTo $N$}{
        \lIf{$adj[i].size > 1$}{s $\leftarrow$ i}
        \lElse{mark[i] $\leftarrow$ 1}
    }
    Dfs(s, -1)\;
    
    \BlankLine
    \lIf{$adj[fRoot].size<3$}{mark[fRoot] $\leftarrow$ 1}
    \For{$a\gets1$ \KwTo $2$}{
        x $\leftarrow$ F[a]\;
        \While{dp[x].s is not $0$}{
            \lIf{$adj[x].size() < 3$}{mark[x] $\leftarrow$ 1;}
            x $\leftarrow$ dp[x].s\;
        }
    }
}
\caption{m-Unbounded Hamiltonian Path for Trees}\label{alg:cutedge}
\end{algorithm}

Once the vertices on the path are determined, the m-Unbounded Hamiltonian Path can by constructed with a DFS beginning at $v_0$ that prioritizes neighbours not on $P$ before moving to the next vertex on $P$. The DFS terminates when $v_k$ is reached. The most costly operation is traversing the tree, so the time complexity is $O(N)$.

\section{Brute Force Algorithms} \label{section:bruteForce}
All approaches in this section will begin with an iteration over all subsets of vertices of the graph sorted in non-decreasing order of the size of the subset. Every instance of the iteration represents the set of unbounded vertices, and if it is determined that an m-Unbounded Hamiltonian Cycle exists, the iteration terminates. 

\subsection{Preliminary Approaches}
We first consider two brute force approaches utilized in the original HCP. One approach naively checks all permutations of vertices. However, this assumes that the path length is fixed at N: if unbounded vertices are included, the size of the k-Unbounded Hamiltonian Cycle and the number of times individual vertices are revisited are unknown, so this approach is not feasible. 

Another possible approach is backtracking. The search begins at some vertex and an unvisited or unbounded neighbour is continually added into the path. Let a vertex $v$ be a \textit{dead end} if all of its neighbours are on the current path. When a dead end is reached, the path reverts to the path before $v$ was added and this process repeats until the path can extend to a path that has not yet been reached. Such an approach could be implemented- in fact, the heuristic discussed in the next section could be considered a greatly optimized version of backtracking. One must be careful to not run into a never-ending cycle between consecutive unbounded vertices. A potential implementation could use a BFS/DFS to find all reachable unbounded vertices to extend a dead end. In comparison to the implementation for the original HCP which is maximally bounded by $O(N!)$, the search will be more costly since more options must be considered after a dead end is extended to an unbounded vertex. The final worst-case time complexity of a back-tracking algorithm is $O(2^{N}N!)$

\subsection{Unbounded DP Algorithm}
The Unbounded DP Algorithm is inspired by the Held-Karp Algorithm \cite{Held-Karp}, a well known dynamic programming approach to solve the TSP and HCP in $O(N^{2}2^N)$. The following description assumes a 0-indexed ordering.

Let $dp[i][j]$ represent whether there exists a k-Unbounded Hamiltonian Path beginning at $0$, ending at $j$, and containing only the vertices in bitmask $i$ (a vertex $v$ corresponds to a set bit with value $2^v$). Trivially, a graph with a single vertex always contains a k-Unbounded Hamiltonian Path so $dp[1][0]$ is initialized to true. Iterate on $i\in [0, 1,...,2^N-1]$. For vertices $k$ and $j$ in bitmask $i$, if $dp[i \bigoplus 2^j][k]$ is true and there exists an edge connecting $j$ and $k$, then $dp[i][j]$ is also true. 

This initial transition only accounts for instances when $j$ is visited exactly once. If $j$ is unbounded, then $dp[i][j]$ may also be reached from bitmask $i$ instead of $i \bigoplus 2^j$. Thus, a second transition utilizing a multisource BFS from the vertices that are marked true in the first transition is implemented. The BFS is only allowed to visit vertices in bitmask $i$ that are unbounded, otherwise a bounded vertex would be visited twice. For every vertex $v$ that the BFS visits, $dp[i][v]$ is marked as true. 

\begin{figure}[htbp] \label{figUnboundedDP}
    \centering

    \tikzset{every picture/.style={line width=0.75pt}} 
    \resizebox{10cm}{!}{
    \begin{tikzpicture}[x=0.75pt,y=0.75pt,yscale=-1,xscale=1]
    
\draw  [fill={rgb, 255:red, 248; green, 231; blue, 28 }  ,fill opacity=1 ] (200.04,214.32) .. controls (200.04,206.36) and (206.5,199.91) .. (214.46,199.91) .. controls (222.42,199.91) and (228.87,206.36) .. (228.87,214.32) .. controls (228.87,222.28) and (222.42,228.73) .. (214.46,228.73) .. controls (206.5,228.73) and (200.04,222.28) .. (200.04,214.32) -- cycle ;
\draw   (248.74,165.12) .. controls (248.74,157.16) and (255.2,150.71) .. (263.16,150.71) .. controls (271.12,150.71) and (277.57,157.16) .. (277.57,165.12) .. controls (277.57,173.08) and (271.12,179.53) .. (263.16,179.53) .. controls (255.2,179.53) and (248.74,173.08) .. (248.74,165.12) -- cycle ;
\draw  [line width=1.75]  (201.04,133.32) .. controls (201.04,125.36) and (207.5,118.91) .. (215.46,118.91) .. controls (223.42,118.91) and (229.87,125.36) .. (229.87,133.32) .. controls (229.87,141.28) and (223.42,147.73) .. (215.46,147.73) .. controls (207.5,147.73) and (201.04,141.28) .. (201.04,133.32) -- cycle ;
\draw   (139.44,125.22) .. controls (139.44,117.26) and (145.9,110.81) .. (153.86,110.81) .. controls (161.82,110.81) and (168.27,117.26) .. (168.27,125.22) .. controls (168.27,133.18) and (161.82,139.63) .. (153.86,139.63) .. controls (145.9,139.63) and (139.44,133.18) .. (139.44,125.22) -- cycle ;
\draw  [fill={rgb, 255:red, 248; green, 231; blue, 28 }  ,fill opacity=1 ][line width=1.75]  (208.44,64.72) .. controls (208.44,56.76) and (214.9,50.31) .. (222.86,50.31) .. controls (230.82,50.31) and (237.27,56.76) .. (237.27,64.72) .. controls (237.27,72.68) and (230.82,79.13) .. (222.86,79.13) .. controls (214.9,79.13) and (208.44,72.68) .. (208.44,64.72) -- cycle ;
\draw  [fill={rgb, 255:red, 248; green, 231; blue, 28 }  ,fill opacity=1 ] (270.94,225.55) .. controls (270.94,217.59) and (277.4,211.14) .. (285.36,211.14) .. controls (293.32,211.14) and (299.77,217.59) .. (299.77,225.55) .. controls (299.77,233.51) and (293.32,239.97) .. (285.36,239.97) .. controls (277.4,239.97) and (270.94,233.51) .. (270.94,225.55) -- cycle ;
\draw  [fill={rgb, 255:red, 248; green, 231; blue, 28 }  ,fill opacity=1 ] (119.44,194.72) .. controls (119.44,186.76) and (125.9,180.31) .. (133.86,180.31) .. controls (141.82,180.31) and (148.27,186.76) .. (148.27,194.72) .. controls (148.27,202.68) and (141.82,209.13) .. (133.86,209.13) .. controls (125.9,209.13) and (119.44,202.68) .. (119.44,194.72) -- cycle ;
\draw    (168.2,128.5) -- (201.04,133.32) ;
\draw    (215.46,147.73) -- (214.46,199.91) ;
\draw    (158.53,111.5) -- (210.2,72.77) ;
\draw    (147.87,138.17) -- (137.2,180.5) ;
\draw    (146.87,200.63) -- (200.04,214.32) ;
\draw    (226.93,206.7) -- (252.93,175.83) ;
\draw    (228.6,139.37) -- (253.93,154.5) ;
\draw    (228.27,219.17) -- (270.94,225.55) ;
\draw    (236.27,59.77) .. controls (268.93,39.77) and (328.27,199.7) .. (298.93,220.17) ;
\draw  [fill={rgb, 255:red, 248; green, 231; blue, 28 }  ,fill opacity=1 ][line width=1.75]  (440.04,213.99) .. controls (440.04,206.03) and (446.5,199.57) .. (454.46,199.57) .. controls (462.42,199.57) and (468.87,206.03) .. (468.87,213.99) .. controls (468.87,221.95) and (462.42,228.4) .. (454.46,228.4) .. controls (446.5,228.4) and (440.04,221.95) .. (440.04,213.99) -- cycle ;
\draw   (488.74,164.79) .. controls (488.74,156.83) and (495.2,150.37) .. (503.16,150.37) .. controls (511.12,150.37) and (517.57,156.83) .. (517.57,164.79) .. controls (517.57,172.75) and (511.12,179.2) .. (503.16,179.2) .. controls (495.2,179.2) and (488.74,172.75) .. (488.74,164.79) -- cycle ;
\draw  [line width=1.75]  (441.04,132.99) .. controls (441.04,125.03) and (447.5,118.57) .. (455.46,118.57) .. controls (463.42,118.57) and (469.87,125.03) .. (469.87,132.99) .. controls (469.87,140.95) and (463.42,147.4) .. (455.46,147.4) .. controls (447.5,147.4) and (441.04,140.95) .. (441.04,132.99) -- cycle ;
\draw   (379.44,124.89) .. controls (379.44,116.93) and (385.9,110.47) .. (393.86,110.47) .. controls (401.82,110.47) and (408.27,116.93) .. (408.27,124.89) .. controls (408.27,132.85) and (401.82,139.3) .. (393.86,139.3) .. controls (385.9,139.3) and (379.44,132.85) .. (379.44,124.89) -- cycle ;
\draw  [fill={rgb, 255:red, 248; green, 231; blue, 28 }  ,fill opacity=1 ][line width=1.75]  (448.44,64.39) .. controls (448.44,56.43) and (454.9,49.97) .. (462.86,49.97) .. controls (470.82,49.97) and (477.27,56.43) .. (477.27,64.39) .. controls (477.27,72.35) and (470.82,78.8) .. (462.86,78.8) .. controls (454.9,78.8) and (448.44,72.35) .. (448.44,64.39) -- cycle ;
\draw  [fill={rgb, 255:red, 248; green, 231; blue, 28 }  ,fill opacity=1 ][line width=1.75]  (510.94,225.22) .. controls (510.94,217.26) and (517.4,210.81) .. (525.36,210.81) .. controls (533.32,210.81) and (539.77,217.26) .. (539.77,225.22) .. controls (539.77,233.18) and (533.32,239.63) .. (525.36,239.63) .. controls (517.4,239.63) and (510.94,233.18) .. (510.94,225.22) -- cycle ;
\draw  [fill={rgb, 255:red, 248; green, 231; blue, 28 }  ,fill opacity=1 ][line width=1.75]  (359.44,194.39) .. controls (359.44,186.43) and (365.9,179.97) .. (373.86,179.97) .. controls (381.82,179.97) and (388.27,186.43) .. (388.27,194.39) .. controls (388.27,202.35) and (381.82,208.8) .. (373.86,208.8) .. controls (365.9,208.8) and (359.44,202.35) .. (359.44,194.39) -- cycle ;
\draw    (408.2,128.17) -- (441.04,132.99) ;
\draw [color={rgb, 255:red, 208; green, 2; blue, 27 }  ,draw opacity=1 ][line width=1.5]    (455.46,147.4) -- (454.46,199.57) ;
\draw    (398.53,111.17) -- (450.2,72.43) ;
\draw    (387.87,137.83) -- (377.2,180.17) ;
\draw [color={rgb, 255:red, 208; green, 2; blue, 27 }  ,draw opacity=1 ][line width=1.5]    (386.87,200.3) -- (440.04,213.99) ;
\draw    (466.93,206.37) -- (492.93,175.5) ;
\draw    (468.6,139.03) -- (493.93,154.17) ;
\draw    (468.27,218.83) -- (510.94,225.22) ;
\draw [color={rgb, 255:red, 208; green, 2; blue, 27 }  ,draw opacity=1 ][line width=1.5]    (476.27,59.43) .. controls (508.93,39.43) and (568.27,199.37) .. (538.93,219.83) ;
\draw   (325.53,142.22) -- (339.29,142.22) -- (339.29,138.4) -- (348.47,146.03) -- (339.29,153.67) -- (339.29,149.85) -- (325.53,149.85) -- cycle ;
    \end{tikzpicture}
    }
    \caption{Second transition using a multisource BFS along red edges. Shaded vertices are unbounded and bolded vertices are marked true}
\end{figure}
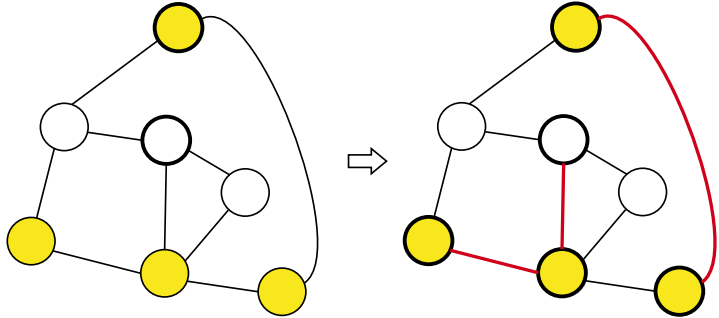

Recall that a k-Unbounded Hamiltonian Cycle is also a k-Unbounded Hamiltonian Path with adjacent end vertices. Importantly, a k-Unbounded Hamiltonian Cycle can always be rotated such that any vertex can be an end vertex. Since the path is fixed at vertex $0$, it suffices to check all $dp[2^N-1][j]$ for all neighbours $j$ of $0$.

The code below stores a pair of values of the previous state instead of a boolean in the dp array to recreate the cycle stored in $path$; $.f$ and $.s$ refer to the first and second values in a pair respectively; .push() inserts an element to the end and .pop() deletes the first element; dp values are intialized to (-1, -1); $ord$ contains integers from $0$ to $2^N-1$ sorted by the number of set bits; $UB[v]=1$ if $v$ is unbounded.

\newpage
\begin{algorithm}[htbp]
\DontPrintSemicolon
\small
\textbf{Input:} Adjacency list $adj[N][]$, adjacency matrix $A[N][N]$

\SetKwFunction{FMain}{Main}
\SetKwFunction{FBfs}{multiBfs}
\SetKwFunction{FCheck}{check}
\SetKwProg{Fn}{void}{:}{}
\SetKwProg{TF}{bool}{:}{}

\BlankLine

\Fn{\FBfs{$x$, $S$}}{
    \lFor{k in S}{
        Q.push(k)
    }
    \While{Q is not empty}{
        v $\leftarrow$ Q.front()\;
        Q.pop()\;
        
        \For{k in adj[v]}{
            nv $\leftarrow$ x $|$ $2^k$\; 
            \lIf{nv $\neq$ x OR dp[x][k].f $\neq$ -1 OR !UB[k]}{continue}
            dp[x][k] $\leftarrow$ (x, v)\;
            Q.push(k)\;
        }
    }
}
\BlankLine
\TF{\FCheck{}}{
    dp[1][0] $\leftarrow$ (0, 0)\;
    \For{$i\gets0$ \KwTo $2^N-1$}{
        v $\leftarrow$ [ ]\;
        \For{$j\gets0$ \KwTo $N-1$}{
            \lIf{i $\&$ $2^j$ is 0}{continue} 
            \For{$k\gets0$ \KwTo $N-1$}{
                \uIf{i $\&$ $2^k$ and A[j][k] and dp[i $\bigoplus 2^j$][k] $\neq$ (-1, -1)}{
                    dp[i][j] $\leftarrow$ (i $\bigoplus 2^j$, k)\;
                    v.push(j)\;
                    break\;
                } 
            }
        }
        Bfs(i, v)\;
    }
    
    \For{$i\gets0$ \KwTo $N$}{
        \uIf{A[0][i] and dp[ $2^N-1$][i] $\ne$ (-1, -1)}{
            V $\leftarrow$ i\;
            return true\;
        }
    }
    return false\;
}
\BlankLine
\Fn{\FMain{}}{
    \For{k in ord}{
        \lFor{$j\gets0$ \KwTo $N-1$}{UB[j] $\leftarrow$ 0}
        \For{$j\gets0$ \KwTo $N-1$}{
            \lIf{k $\&$ $2^j$}{UB[j] $\leftarrow$ 1}
        }
        \uIf{check()}{
            C $\leftarrow$ $2^N-1$\;
            \While{C $\ne$ 0}{
                path.push(V)\;
                tC $\leftarrow$ C, tV $\leftarrow$ V\;
                C $\leftarrow$ dp[tC][tV].f, V $\leftarrow$ dp[tC][tV].s\;
            }
            return\;
        }
    }
}
\caption{Unbounded DP Algorithm}\label{alg:unboundedDP}
\end{algorithm}
The iteration over subsets of $G$ in $ord$ contributes $O(2^N)$. The loop over all bitmasks $i$ from $0$ to $2^N-1$ contributes another factor of $O(2^N)$. Within this loop is the initial transition, which runs in $O(N^2)$ from the nested loops iterating over the bits of $i$, as well as the BFS, which runs in $O(E)$. The final time complexity is $O(2^N)O(2^N)O(N^2+E) = O(4^{N}N^2)$. $O(2^{N}N)$ memory is required to maintain the dp matrix.

\section{Conversions to k-Unbounded HCP}
The k-Unbounded Hamiltonian Cycle Problem can trivially be converted to the Hamiltonian Cycle Problem through the determination of a 0-Unbounded Hamiltonian Cycle. In this section, we explore the inverse relationship and propose two different constructions in which the HCP and ATSP can be applied to solve the k-Unbounded HCP. 

\begin{lemma} \label{lemma5}
For a given m-Unbounded Hamiltonian Cycle, there exists an m-Unbounded Hamiltonian Cycle with the same set of unbounded vertices such that every vertex appears no more than N times
\end{lemma}
\begin{proof}
Let there be a vertex $v$ that appears $k$ times in an m-Unbounded Hamiltonian Cycle $P$ for $k>N$. Consider the set $S=(P_1, P_2,...,P_k)$ that is composed of the paths strictly between consecutive instances of $v$ in $P$. If every vertex in a given $P_i$ has already appeared in a path $P_j$ for $j<i$, then the path formed from removing $P_i$ and the following instance of $v$ from P is also an m-Unbounded Hamiltonian Cycle, since every removed vertex appears in $P$ at least two times. If every such $P_i$ is removed from $S$ and the corresponding vertices are removed from $P$, for every path $P_u$, there exists at least one vertex that does not appear in any $P_j$ for $j<u$. It follows that there are at least $|S|+1$ distinct vertices in $P$ (including $v$), which is also bounded by $N$. Thus, $|S|\le N-1$, and since $v$ appears $|S|+1$ times in $P$, there exists an m-Unbounded Hamiltonian Cycle that contains no more than $N$ occurrences of $v$. The same process may be applied to all vertices that appear greater than $N$ times to construct an m-Unbounded Hamiltonian Cycle that satisfies the constraints. 
\end{proof}
\subsection{HCP Conversion}
Applying the HCP on this construction verifies whether a k-Unbounded Hamiltonian Cycle exists given a set of unbounded vertices. The construction works as follows: replace every unbounded vertex $u$ with $C_u$, a cycle of length $N$. For every vertex in $C_u$, add an edge to an adjacent vertex $k$, or if $k$ is unbounded, add an edge to every vertex in $C_k$.

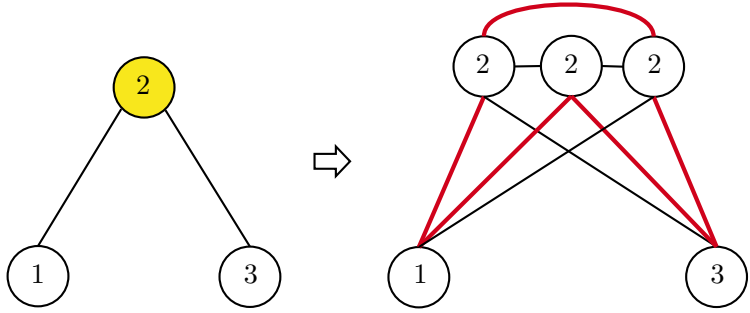
\begin{figure}[htbp] \label{figHCPConversion}
    \centering

    \tikzset{every picture/.style={line width=0.75pt}} 
    \resizebox{10cm}{!}{
    \begin{tikzpicture}[x=0.75pt,y=0.75pt,yscale=-1,xscale=1]
    
\draw  [fill={rgb, 255:red, 248; green, 231; blue, 28 }  ,fill opacity=1 ] (271.27,105.05) .. controls (271.27,97.13) and (277.68,90.72) .. (285.6,90.72) .. controls (293.52,90.72) and (299.93,97.13) .. (299.93,105.05) .. controls (299.93,112.97) and (293.52,119.38) .. (285.6,119.38) .. controls (277.68,119.38) and (271.27,112.97) .. (271.27,105.05) -- cycle ;
\draw   (321.15,195.03) .. controls (321.15,187.12) and (327.57,180.7) .. (335.48,180.7) .. controls (343.4,180.7) and (349.82,187.12) .. (349.82,195.03) .. controls (349.82,202.95) and (343.4,209.37) .. (335.48,209.37) .. controls (327.57,209.37) and (321.15,202.95) .. (321.15,195.03) -- cycle ;
\draw   (221.27,194.75) .. controls (221.27,186.83) and (227.68,180.42) .. (235.6,180.42) .. controls (243.52,180.42) and (249.93,186.83) .. (249.93,194.75) .. controls (249.93,202.67) and (243.52,209.08) .. (235.6,209.08) .. controls (227.68,209.08) and (221.27,202.67) .. (221.27,194.75) -- cycle ;
\draw    (274.85,115.33) -- (235.6,180.42) ;
\draw    (295.65,115.83) -- (335.48,180.7) ;
\draw   (365.93,136.63) -- (375.98,136.63) -- (375.98,133.1) -- (382.68,140.17) -- (375.98,147.24) -- (375.98,143.7) -- (365.93,143.7) -- cycle ;
\draw   (431.45,95.09) .. controls (431.45,87.18) and (437.87,80.76) .. (445.79,80.76) .. controls (453.7,80.76) and (460.12,87.18) .. (460.12,95.09) .. controls (460.12,103.01) and (453.7,109.43) .. (445.79,109.43) .. controls (437.87,109.43) and (431.45,103.01) .. (431.45,95.09) -- cycle ;
\draw   (541.09,195.09) .. controls (541.09,187.17) and (547.51,180.76) .. (555.43,180.76) .. controls (563.34,180.76) and (569.76,187.17) .. (569.76,195.09) .. controls (569.76,203) and (563.34,209.42) .. (555.43,209.42) .. controls (547.51,209.42) and (541.09,203) .. (541.09,195.09) -- cycle ;
\draw   (400.77,195.25) .. controls (400.77,187.33) and (407.18,180.92) .. (415.1,180.92) .. controls (423.02,180.92) and (429.43,187.33) .. (429.43,195.25) .. controls (429.43,203.17) and (423.02,209.58) .. (415.1,209.58) .. controls (407.18,209.58) and (400.77,203.17) .. (400.77,195.25) -- cycle ;
\draw [color={rgb, 255:red, 208; green, 2; blue, 27 }  ,draw opacity=1 ][line width=1.5]    (445.79,109.43) -- (415.1,180.92) ;
\draw    (445.79,109.43) -- (555.43,180.76) ;
\draw   (472.64,94.88) .. controls (472.64,86.96) and (479.06,80.54) .. (486.97,80.54) .. controls (494.89,80.54) and (501.31,86.96) .. (501.31,94.88) .. controls (501.31,102.79) and (494.89,109.21) .. (486.97,109.21) .. controls (479.06,109.21) and (472.64,102.79) .. (472.64,94.88) -- cycle ;
\draw   (511.43,95.19) .. controls (511.43,87.28) and (517.85,80.86) .. (525.77,80.86) .. controls (533.68,80.86) and (540.1,87.28) .. (540.1,95.19) .. controls (540.1,103.11) and (533.68,109.53) .. (525.77,109.53) .. controls (517.85,109.53) and (511.43,103.11) .. (511.43,95.19) -- cycle ;
\draw [color={rgb, 255:red, 208; green, 2; blue, 27 }  ,draw opacity=1 ][line width=1.5]    (486.97,109.21) -- (415.1,180.92) ;
\draw [draw opacity=0][line width=0.75]    (525.77,109.53) -- (415.1,180.92) ;
\draw [color={rgb, 255:red, 208; green, 2; blue, 27 }  ,draw opacity=1 ][line width=1.5]    (486.97,109.21) -- (555.43,180.76) ;
\draw [color={rgb, 255:red, 208; green, 2; blue, 27 }  ,draw opacity=1 ][line width=1.5]    (525.77,109.53) -- (555.43,180.76) ;
\draw    (460.12,95.09) -- (472.64,94.88) ;
\draw    (501.31,94.88) -- (511.43,95.19) ;
\draw [color={rgb, 255:red, 208; green, 2; blue, 27 }  ,draw opacity=1 ][line width=1.5]    (445.79,80.76) .. controls (445.4,60.6) and (525.69,60.6) .. (525.77,80.86) ;
\draw    (415.1,180.92) -- (525.77,109.53) ;

\draw (230.73,188) node [anchor=north west][inner sep=0.75pt]   [align=left] {1};
\draw (331.07,188) node [anchor=north west][inner sep=0.75pt]   [align=left] {3};
\draw (551.07,188) node [anchor=north west][inner sep=0.75pt]   [align=left] {3};
\draw (411,188) node [anchor=north west][inner sep=0.75pt]   [align=left] {1};
\draw (280.4,98) node [anchor=north west][inner sep=0.75pt]   [align=left] {2};
\draw (440.4,88) node [anchor=north west][inner sep=0.75pt]   [align=left] {2};
\draw (482.07,88) node [anchor=north west][inner sep=0.75pt]   [align=left] {2};
\draw (521.4,88) node [anchor=north west][inner sep=0.75pt]   [align=left] {2};

    \end{tikzpicture}
    }
    \caption{HCP Construction for $N=3$ with the shaded vertex unbounded}
\end{figure}

If the vertices in $C_u$ are labeled with the corresponding unbounded vertex, the Hamiltonian Cycle in which consecutive instances of the same vertex are reduced to a single vertex is a k-Unbounded Hamiltonian Cycle: the set of bounded vertices remains the same and every vertex appears at least once in this constructed path. Furthermore, any k-Unbounded Hamiltonian Cycle may be converted to a Hamiltonian Cycle in the construction. Let the k-Unbounded Hamiltonian Cycle formed after applying the process outlined in Lemma \ref{lemma5} be $P$, and let $C_u=(u_1, u_2,...,u_{N})$ for a given unbounded vertex $u$. We propose the following transformation: for every unbounded vertex $u$ that appears $m$ times in $P$, replace the ordered sequence of all $m$ instances of $u$ with the sequence $(u_1, u_2,...,u_m)$ and insert the sequence $(u_{m+1}, u_{m+2},...,u_{N})$ following $u_m$. This new path is a Hamiltonian Cycle since all edges exist and every vertex is visited in the construction. It follows that a k-Unbounded Hamiltonian Cycle exists if and only if a Hamiltonian Cycle exists in the construction.

An m-Unbounded Hamiltonian Cycle may be determined by utilizing the same brute force approach in Section \ref{section:bruteForce}- iterating over the subsets of vertices in non-decreasing order of size. Since the size of the construction can have $N^2$ vertices, if the Held-Karp Algorithm is used to identify a Hamiltonian Cycle, the worst-case time complexity is $O(2^{N+N^2}N^4)$.

\subsection{ATSP Conversion}
Let there exist an m-Unbounded Hamiltonian Cycle $P$ that satisfies the conditions outlined in Lemma \ref{lemma5} in a graph $G$. Consider the directed instance of $G$ in which every bidirectional edge is replaced with 2 opposing directed edges. We propose the following transformation for each vertex $v$. $v$ is replaced with the vertex set $\{v_{a_1}, v_{a_2},..., v_{a_N}, v_{b_1}, v_{b_2},..., v_{b_N}, v_{c_1}, v_{c_2},..., v_{c_N}\}$ of size $3N$. All outgoing directed edges originate from the vertex set $\{v_{c_1}, v_{c_2},..., v_{c_N}, v_{a_1}\}$ and all incoming directed edges end at the vertex set $\{v_{a_1}, v_{a_2},..., v_{a_N}, v_{c_N}\}$. More formally, for every neighbour $k$ of $v$ in $G$, there exist edges $v_{i}k_{j}$ for all $i \in \{c_1, c_2,..., c_N, a_1\}$ and $j \in \{a_1, a_2,..., a_N, c_N\}$. Furthermore, there exist the undirected (or two directed) edges $v_{a_i}v_{b_i}$, $v_{b_i}v_{c_i}$ for $1\le i\le N$ and $v_{a_i}v_{c_{i-1}}$ for $2\le i\le N$. Every edge has a weight of $0$ except for a single ``downward" directed edge in the form $v_{a_i}v_{b_i}$ or $v_{a_j}v_{b_j}$ with weight $1$.

\begin{figure}[htbp] \label{figATSPConversion}
    \centering

    \tikzset{every picture/.style={line width=0.75pt}} 
    \resizebox{10cm}{!}{
    \begin{tikzpicture}[x=0.75pt,y=0.75pt,yscale=-1,xscale=1]
\draw   (120.72,80.13) .. controls (120.72,69.52) and (129.33,60.91) .. (139.94,60.91) .. controls (150.56,60.91) and (159.17,69.52) .. (159.17,80.13) .. controls (159.17,90.75) and (150.56,99.36) .. (139.94,99.36) .. controls (129.33,99.36) and (120.72,90.75) .. (120.72,80.13) -- cycle ;
\draw   (120.06,150.53) .. controls (120.06,139.92) and (128.66,131.31) .. (139.28,131.31) .. controls (149.89,131.31) and (158.5,139.92) .. (158.5,150.53) .. controls (158.5,161.15) and (149.89,169.76) .. (139.28,169.76) .. controls (128.66,169.76) and (120.06,161.15) .. (120.06,150.53) -- cycle ;
\draw   (120.44,220.13) .. controls (120.44,209.52) and (129.04,200.91) .. (139.66,200.91) .. controls (150.27,200.91) and (158.88,209.52) .. (158.88,220.13) .. controls (158.88,230.75) and (150.27,239.36) .. (139.66,239.36) .. controls (129.04,239.36) and (120.44,230.75) .. (120.44,220.13) -- cycle ;
\draw   (200.99,80) .. controls (200.99,69.38) and (209.59,60.78) .. (220.21,60.78) .. controls (230.83,60.78) and (239.43,69.38) .. (239.43,80) .. controls (239.43,90.62) and (230.83,99.22) .. (220.21,99.22) .. controls (209.59,99.22) and (200.99,90.62) .. (200.99,80) -- cycle ;
\draw   (200.32,150.13) .. controls (200.32,139.52) and (208.93,130.91) .. (219.54,130.91) .. controls (230.16,130.91) and (238.77,139.52) .. (238.77,150.13) .. controls (238.77,160.75) and (230.16,169.36) .. (219.54,169.36) .. controls (208.93,169.36) and (200.32,160.75) .. (200.32,150.13) -- cycle ;
\draw   (200.32,220.53) .. controls (200.32,209.92) and (208.93,201.31) .. (219.54,201.31) .. controls (230.16,201.31) and (238.77,209.92) .. (238.77,220.53) .. controls (238.77,231.15) and (230.16,239.76) .. (219.54,239.76) .. controls (208.93,239.76) and (200.32,231.15) .. (200.32,220.53) -- cycle ;
\draw   (281.12,80.53) .. controls (281.12,69.92) and (289.73,61.31) .. (300.34,61.31) .. controls (310.96,61.31) and (319.57,69.92) .. (319.57,80.53) .. controls (319.57,91.15) and (310.96,99.76) .. (300.34,99.76) .. controls (289.73,99.76) and (281.12,91.15) .. (281.12,80.53) -- cycle ;
\draw   (280.32,150.53) .. controls (280.32,139.92) and (288.93,131.31) .. (299.54,131.31) .. controls (310.16,131.31) and (318.77,139.92) .. (318.77,150.53) .. controls (318.77,161.15) and (310.16,169.76) .. (299.54,169.76) .. controls (288.93,169.76) and (280.32,161.15) .. (280.32,150.53) -- cycle ;
\draw   (280.72,220.53) .. controls (280.72,209.92) and (289.33,201.31) .. (299.94,201.31) .. controls (310.56,201.31) and (319.17,209.92) .. (319.17,220.53) .. controls (319.17,231.15) and (310.56,239.76) .. (299.94,239.76) .. controls (289.33,239.76) and (280.72,231.15) .. (280.72,220.53) -- cycle ;
\draw   (360.72,80.13) .. controls (360.72,69.52) and (369.33,60.91) .. (379.94,60.91) .. controls (390.56,60.91) and (399.17,69.52) .. (399.17,80.13) .. controls (399.17,90.75) and (390.56,99.36) .. (379.94,99.36) .. controls (369.33,99.36) and (360.72,90.75) .. (360.72,80.13) -- cycle ;
\draw   (361.52,150.53) .. controls (361.52,139.92) and (370.13,131.31) .. (380.74,131.31) .. controls (391.36,131.31) and (399.97,139.92) .. (399.97,150.53) .. controls (399.97,161.15) and (391.36,169.76) .. (380.74,169.76) .. controls (370.13,169.76) and (361.52,161.15) .. (361.52,150.53) -- cycle ;
\draw   (361.18,220.53) .. controls (361.18,209.92) and (369.79,201.31) .. (380.4,201.31) .. controls (391.02,201.31) and (399.62,209.92) .. (399.62,220.53) .. controls (399.62,231.15) and (391.02,239.76) .. (380.4,239.76) .. controls (369.79,239.76) and (361.18,231.15) .. (361.18,220.53) -- cycle ;
\draw    (139.94,99.36) -- (139.28,131.31) ;
\draw    (138.99,169.76) -- (139.66,200.91) ;
\draw    (220.21,99.22) -- (219.54,130.91) ;
\draw    (219.54,169.36) -- (219.54,201.31) ;
\draw    (299.54,169.76) -- (299.94,201.31) ;
\draw    (379.94,99.36) -- (380.74,131.31) ;
\draw    (380.74,169.76) -- (380.4,201.31) ;
\draw    (209.07,95.57) -- (149.73,203.43) ;
\draw    (288.4,96.1) -- (229.73,203.77) ;
\draw    (367.07,94.1) -- (309.73,204.1) ;
\draw [fill={rgb, 255:red, 0; green, 0; blue, 0 }  ,fill opacity=1 ]   (120.16,30.6) -- (129.17,55.22) ;
\draw [shift={(130.2,58.04)}, rotate = 249.89] [fill={rgb, 255:red, 0; green, 0; blue, 0 }  ][line width=0.08]  [draw opacity=0] (8.93,-4.29) -- (0,0) -- (8.93,4.29) -- cycle    ;
\draw [fill={rgb, 255:red, 0; green, 0; blue, 0 }  ,fill opacity=1 ]   (140.6,30.44) -- (140.25,53.44) ;
\draw [shift={(140.2,56.44)}, rotate = 270.88] [fill={rgb, 255:red, 0; green, 0; blue, 0 }  ][line width=0.08]  [draw opacity=0] (8.93,-4.29) -- (0,0) -- (8.93,4.29) -- cycle    ;
\draw [fill={rgb, 255:red, 0; green, 0; blue, 0 }  ,fill opacity=1 ]   (160.38,30.6) -- (151.97,55.2) ;
\draw [shift={(151,58.04)}, rotate = 288.87] [fill={rgb, 255:red, 0; green, 0; blue, 0 }  ][line width=0.08]  [draw opacity=0] (8.93,-4.29) -- (0,0) -- (8.93,4.29) -- cycle    ;
\draw [fill={rgb, 255:red, 0; green, 0; blue, 0 }  ,fill opacity=1 ]   (200.16,30.27) -- (209.17,54.89) ;
\draw [shift={(210.2,57.71)}, rotate = 249.89] [fill={rgb, 255:red, 0; green, 0; blue, 0 }  ][line width=0.08]  [draw opacity=0] (8.93,-4.29) -- (0,0) -- (8.93,4.29) -- cycle    ;
\draw [fill={rgb, 255:red, 0; green, 0; blue, 0 }  ,fill opacity=1 ]   (220.6,30.11) -- (220.25,53.11) ;
\draw [shift={(220.2,56.11)}, rotate = 270.88] [fill={rgb, 255:red, 0; green, 0; blue, 0 }  ][line width=0.08]  [draw opacity=0] (8.93,-4.29) -- (0,0) -- (8.93,4.29) -- cycle    ;
\draw [fill={rgb, 255:red, 0; green, 0; blue, 0 }  ,fill opacity=1 ]   (240.38,30.27) -- (231.97,54.87) ;
\draw [shift={(231,57.71)}, rotate = 288.87] [fill={rgb, 255:red, 0; green, 0; blue, 0 }  ][line width=0.08]  [draw opacity=0] (8.93,-4.29) -- (0,0) -- (8.93,4.29) -- cycle    ;
\draw [fill={rgb, 255:red, 0; green, 0; blue, 0 }  ,fill opacity=1 ]   (280.16,29.6) -- (289.17,54.22) ;
\draw [shift={(290.2,57.04)}, rotate = 249.89] [fill={rgb, 255:red, 0; green, 0; blue, 0 }  ][line width=0.08]  [draw opacity=0] (8.93,-4.29) -- (0,0) -- (8.93,4.29) -- cycle    ;
\draw [fill={rgb, 255:red, 0; green, 0; blue, 0 }  ,fill opacity=1 ]   (300.6,29.44) -- (300.25,52.44) ;
\draw [shift={(300.2,55.44)}, rotate = 270.88] [fill={rgb, 255:red, 0; green, 0; blue, 0 }  ][line width=0.08]  [draw opacity=0] (8.93,-4.29) -- (0,0) -- (8.93,4.29) -- cycle    ;
\draw [fill={rgb, 255:red, 0; green, 0; blue, 0 }  ,fill opacity=1 ]   (320.38,29.6) -- (311.97,54.2) ;
\draw [shift={(311,57.04)}, rotate = 288.87] [fill={rgb, 255:red, 0; green, 0; blue, 0 }  ][line width=0.08]  [draw opacity=0] (8.93,-4.29) -- (0,0) -- (8.93,4.29) -- cycle    ;
\draw [fill={rgb, 255:red, 0; green, 0; blue, 0 }  ,fill opacity=1 ]   (360.2,29.6) -- (369.21,54.22) ;
\draw [shift={(370.24,57.04)}, rotate = 249.89] [fill={rgb, 255:red, 0; green, 0; blue, 0 }  ][line width=0.08]  [draw opacity=0] (8.93,-4.29) -- (0,0) -- (8.93,4.29) -- cycle    ;
\draw [fill={rgb, 255:red, 0; green, 0; blue, 0 }  ,fill opacity=1 ]   (380.64,29.44) -- (380.29,52.44) ;
\draw [shift={(380.24,55.44)}, rotate = 270.88] [fill={rgb, 255:red, 0; green, 0; blue, 0 }  ][line width=0.08]  [draw opacity=0] (8.93,-4.29) -- (0,0) -- (8.93,4.29) -- cycle    ;
\draw [fill={rgb, 255:red, 0; green, 0; blue, 0 }  ,fill opacity=1 ]   (400.42,29.6) -- (392.01,54.2) ;
\draw [shift={(391.04,57.04)}, rotate = 288.87] [fill={rgb, 255:red, 0; green, 0; blue, 0 }  ][line width=0.08]  [draw opacity=0] (8.93,-4.29) -- (0,0) -- (8.93,4.29) -- cycle    ;
\draw [fill={rgb, 255:red, 0; green, 0; blue, 0 }  ,fill opacity=1 ]   (370.09,237.29) -- (361.3,264.43) ;
\draw [shift={(360.37,267.29)}, rotate = 287.94] [fill={rgb, 255:red, 0; green, 0; blue, 0 }  ][line width=0.08]  [draw opacity=0] (8.93,-4.29) -- (0,0) -- (8.93,4.29) -- cycle    ;
\draw [fill={rgb, 255:red, 0; green, 0; blue, 0 }  ,fill opacity=1 ]   (380.4,239.76) -- (380.12,266.86) ;
\draw [shift={(380.09,269.86)}, rotate = 270.6] [fill={rgb, 255:red, 0; green, 0; blue, 0 }  ][line width=0.08]  [draw opacity=0] (8.93,-4.29) -- (0,0) -- (8.93,4.29) -- cycle    ;
\draw [fill={rgb, 255:red, 0; green, 0; blue, 0 }  ,fill opacity=1 ]   (390.33,237.4) -- (399.37,263.89) ;
\draw [shift={(400.33,266.73)}, rotate = 251.18] [fill={rgb, 255:red, 0; green, 0; blue, 0 }  ][line width=0.08]  [draw opacity=0] (8.93,-4.29) -- (0,0) -- (8.93,4.29) -- cycle    ;
\draw [fill={rgb, 255:red, 0; green, 0; blue, 0 }  ,fill opacity=1 ]   (290.6,237.57) -- (281.28,265.02) ;
\draw [shift={(280.31,267.86)}, rotate = 288.76] [fill={rgb, 255:red, 0; green, 0; blue, 0 }  ][line width=0.08]  [draw opacity=0] (8.93,-4.29) -- (0,0) -- (8.93,4.29) -- cycle    ;
\draw [fill={rgb, 255:red, 0; green, 0; blue, 0 }  ,fill opacity=1 ]   (299.94,239.76) -- (300.02,267.14) ;
\draw [shift={(300.03,270.14)}, rotate = 269.84] [fill={rgb, 255:red, 0; green, 0; blue, 0 }  ][line width=0.08]  [draw opacity=0] (8.93,-4.29) -- (0,0) -- (8.93,4.29) -- cycle    ;
\draw [fill={rgb, 255:red, 0; green, 0; blue, 0 }  ,fill opacity=1 ]   (309.74,237.57) -- (319.05,264.45) ;
\draw [shift={(320.03,267.29)}, rotate = 250.91] [fill={rgb, 255:red, 0; green, 0; blue, 0 }  ][line width=0.08]  [draw opacity=0] (8.93,-4.29) -- (0,0) -- (8.93,4.29) -- cycle    ;
\draw [fill={rgb, 255:red, 0; green, 0; blue, 0 }  ,fill opacity=1 ]   (210.03,237.57) -- (201.22,265.57) ;
\draw [shift={(200.31,268.43)}, rotate = 287.47] [fill={rgb, 255:red, 0; green, 0; blue, 0 }  ][line width=0.08]  [draw opacity=0] (8.93,-4.29) -- (0,0) -- (8.93,4.29) -- cycle    ;
\draw [fill={rgb, 255:red, 0; green, 0; blue, 0 }  ,fill opacity=1 ]   (219.54,239.76) -- (219.98,266.86) ;
\draw [shift={(220.03,269.86)}, rotate = 269.08] [fill={rgb, 255:red, 0; green, 0; blue, 0 }  ][line width=0.08]  [draw opacity=0] (8.93,-4.29) -- (0,0) -- (8.93,4.29) -- cycle    ;
\draw [fill={rgb, 255:red, 0; green, 0; blue, 0 }  ,fill opacity=1 ]   (230.03,236.43) -- (238.59,264.7) ;
\draw [shift={(239.46,267.57)}, rotate = 253.16] [fill={rgb, 255:red, 0; green, 0; blue, 0 }  ][line width=0.08]  [draw opacity=0] (8.93,-4.29) -- (0,0) -- (8.93,4.29) -- cycle    ;
\draw [fill={rgb, 255:red, 0; green, 0; blue, 0 }  ,fill opacity=1 ]   (130.26,237.46) -- (121.18,265.81) ;
\draw [shift={(120.27,268.67)}, rotate = 287.75] [fill={rgb, 255:red, 0; green, 0; blue, 0 }  ][line width=0.08]  [draw opacity=0] (8.93,-4.29) -- (0,0) -- (8.93,4.29) -- cycle    ;
\draw [fill={rgb, 255:red, 0; green, 0; blue, 0 }  ,fill opacity=1 ]   (139.66,239.36) -- (140,267.36) ;
\draw [shift={(140.04,270.36)}, rotate = 269.3] [fill={rgb, 255:red, 0; green, 0; blue, 0 }  ][line width=0.08]  [draw opacity=0] (8.93,-4.29) -- (0,0) -- (8.93,4.29) -- cycle    ;
\draw [fill={rgb, 255:red, 0; green, 0; blue, 0 }  ,fill opacity=1 ]   (149.46,236.43) -- (159.31,265.82) ;
\draw [shift={(160.27,268.67)}, rotate = 251.46] [fill={rgb, 255:red, 0; green, 0; blue, 0 }  ][line width=0.08]  [draw opacity=0] (8.93,-4.29) -- (0,0) -- (8.93,4.29) -- cycle    ;
\draw [fill={rgb, 255:red, 0; green, 0; blue, 0 }  ,fill opacity=1 ]   (431.3,199.1) -- (407.06,209.09) ;
\draw [shift={(404.29,210.23)}, rotate = 337.6] [fill={rgb, 255:red, 0; green, 0; blue, 0 }  ][line width=0.08]  [draw opacity=0] (8.93,-4.29) -- (0,0) -- (8.93,4.29) -- cycle    ;
\draw [fill={rgb, 255:red, 0; green, 0; blue, 0 }  ,fill opacity=1 ]   (432.28,219.52) -- (409.28,220.09) ;
\draw [shift={(406.28,220.16)}, rotate = 358.59] [fill={rgb, 255:red, 0; green, 0; blue, 0 }  ][line width=0.08]  [draw opacity=0] (8.93,-4.29) -- (0,0) -- (8.93,4.29) -- cycle    ;
\draw [fill={rgb, 255:red, 0; green, 0; blue, 0 }  ,fill opacity=1 ]   (432.91,239.29) -- (407.99,231.87) ;
\draw [shift={(405.12,231.02)}, rotate = 16.58] [fill={rgb, 255:red, 0; green, 0; blue, 0 }  ][line width=0.08]  [draw opacity=0] (8.93,-4.29) -- (0,0) -- (8.93,4.29) -- cycle    ;
\draw [fill={rgb, 255:red, 0; green, 0; blue, 0 }  ,fill opacity=1 ]   (122.5,70.85) -- (94.44,60.91) ;
\draw [shift={(91.62,59.91)}, rotate = 19.51] [fill={rgb, 255:red, 0; green, 0; blue, 0 }  ][line width=0.08]  [draw opacity=0] (8.93,-4.29) -- (0,0) -- (8.93,4.29) -- cycle    ;
\draw [fill={rgb, 255:red, 0; green, 0; blue, 0 }  ,fill opacity=1 ]   (120.32,80.19) -- (92.32,79.68) ;
\draw [shift={(89.32,79.62)}, rotate = 1.06] [fill={rgb, 255:red, 0; green, 0; blue, 0 }  ][line width=0.08]  [draw opacity=0] (8.93,-4.29) -- (0,0) -- (8.93,4.29) -- cycle    ;
\draw [fill={rgb, 255:red, 0; green, 0; blue, 0 }  ,fill opacity=1 ]   (122.94,90.07) -- (93.26,99.02) ;
\draw [shift={(90.39,99.89)}, rotate = 343.22] [fill={rgb, 255:red, 0; green, 0; blue, 0 }  ][line width=0.08]  [draw opacity=0] (8.93,-4.29) -- (0,0) -- (8.93,4.29) -- cycle    ;
\draw [fill={rgb, 255:red, 0; green, 0; blue, 0 }  ,fill opacity=1 ]   (296.36,99.48) -- (296.36,126.08) ;
\draw [shift={(296.36,129.08)}, rotate = 270] [fill={rgb, 255:red, 0; green, 0; blue, 0 }  ][line width=0.08]  [draw opacity=0] (10.72,-5.15) -- (0,0) -- (10.72,5.15) -- (7.12,0) -- cycle    ;
\draw [fill={rgb, 255:red, 0; green, 0; blue, 0 }  ,fill opacity=1 ]   (304.14,131.97) -- (304.14,105.26) ;
\draw [shift={(304.14,102.26)}, rotate = 90] [fill={rgb, 255:red, 0; green, 0; blue, 0 }  ][line width=0.08]  [draw opacity=0] (10.72,-5.15) -- (0,0) -- (10.72,5.15) -- (7.12,0) -- cycle    ;

\draw (130,73) node [anchor=north west][inner sep=0.75pt] [font=\large]    [align=left] {$\displaystyle v_{a_{1}}$};
\draw (210,73) node [anchor=north west][inner sep=0.75pt] [font=\large]    [align=left] {$\displaystyle v_{a_{2}}$};
\draw (290,73) node [anchor=north west][inner sep=0.75pt] [font=\large]     [align=left] {$\displaystyle v_{a_{3}}$};
\draw (370,73) node [anchor=north west][inner sep=0.75pt] [font=\large]    [align=left] {$\displaystyle v_{a_{4}}$};
\draw (130,143) node [anchor=north west][inner sep=0.75pt] [font=\large]    [align=left] {$\displaystyle v_{b_{1}}$};
\draw (210,143) node [anchor=north west][inner sep=0.75pt] [font=\large]    [align=left] {$\displaystyle v_{b_{2}}$};
\draw (290,143) node [anchor=north west][inner sep=0.75pt] [font=\large]    [align=left] {$\displaystyle v_{b_{3}}$};
\draw (370,143) node [anchor=north west][inner sep=0.75pt] [font=\large]    [align=left] {$\displaystyle v_{b_{4}}$};
\draw (130,213) node [anchor=north west][inner sep=0.75pt] [font=\large]    [align=left] {$\displaystyle v_{c_{1}}$};
\draw (210,213) node [anchor=north west][inner sep=0.75pt] [font=\large]    [align=left] {$\displaystyle v_{c_{2}}$};
\draw (290,213) node [anchor=north west][inner sep=0.75pt] [font=\large]    [align=left] {$\displaystyle v_{c_{3}}$};
\draw (370,213) node [anchor=north west][inner sep=0.75pt] [font=\large]  [align=left] {$\displaystyle v_{c_{4}}$};
\draw (286,106.16) node [anchor=north west][inner sep=0.75pt]  [font=\scriptsize] [align=left] {$\displaystyle 1$};

    \end{tikzpicture}
    }
    \caption{ATSP transformation of $v$ for $N=4$}
\end{figure}
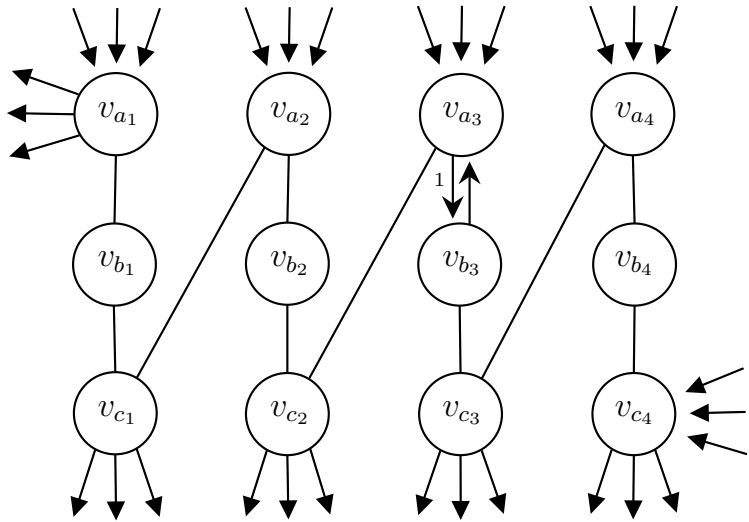

If the construction is entered through an incoming directed edge to $v_{c_N}$, every other vertex in the construction must be visited before it is exited (otherwise there will exist some $v_{b_i}$ that is either never visited or becomes a dead end). Additionally, this sequence of consecutive visits $(v_{c_N}, v_{b_N}, v_{a_N}, v_{c_{N-1}},..., v_{b_1}, v_{a_1})$ does not add to the weight of the tour since this path travels ``up" along $v_{c_i}v_{b_i}$ or $v_{b_i}v_{a_i}$. Thus, entering the construction from $v_{c_N}$ corresponds to visiting $v$ exactly once. 

Alternatively, if the construction is visited through $v_{a_i}$, the path must continue vertically downward along $(v_{b_i}, v_{c_i})$. At this point, the path may either leave the construction through an outgoing directed edge or move right along the edge $v_{c_i}v_{a_{i+1}}$. The former action represents a single instance of $v$ in $P$, and the latter action allows every vertex in the construction to be visited even if $v$ appears less than $N$ times in $P$. Every downward directed edge must be used to visit every vertex in the construction, so the contributed weight is $1$. Thus, this scenario corresponds to $v$ being unbounded as the construction can be visited up to $N$ times for every vertex $v_{a_i}$ for $1\le i\le N$.

It follows that the weight of the optimal ATSP tour is $m$. Given the tour, the m-Unbounded Hamiltonian Cycle may be reconstructed by reducing consecutive instances of the same vertex to a single vertex. The transformation has a total of $N^2$ vertices and $2E(N+1)^2+N^2(3N-1)$ edges.
\section{Unbounded Heuristic}
\subsection{Methods}
\subsubsection{Rotations and Cycle Extensions}
Pósa \cite{Posa-Rotation} describes a method called a \textit{rotation} that can be used to reroute a path from a dead end. Figure 2 illustrates this process: the end vertex of the path is shifted from $6$ to $3$ by traveling along the edge from $2$ to $6$.

\newpage
\begin{figure}[htbp] \label{fig2}
    \centering

    \tikzset{every picture/.style={line width=0.75pt}} 
        \resizebox{5cm}{!}{
    \begin{tikzpicture}[x=0.75pt,y=0.75pt,yscale=-1,xscale=1]
    
    \draw   (251.54,60.33) .. controls (251.54,55.14) and (255.74,50.94) .. (260.93,50.94) .. controls (266.11,50.94) and (270.31,55.14) .. (270.31,60.33) .. controls (270.31,65.51) and (266.11,69.71) .. (260.93,69.71) .. controls (255.74,69.71) and (251.54,65.51) .. (251.54,60.33) -- cycle ;
    \draw   (210.97,60.61) .. controls (210.97,55.43) and (215.17,51.23) .. (220.36,51.23) .. controls (225.54,51.23) and (229.74,55.43) .. (229.74,60.61) .. controls (229.74,65.8) and (225.54,70) .. (220.36,70) .. controls (215.17,70) and (210.97,65.8) .. (210.97,60.61) -- cycle ;
    \draw   (170.97,60.33) .. controls (170.97,55.14) and (175.17,50.94) .. (180.36,50.94) .. controls (185.54,50.94) and (189.74,55.14) .. (189.74,60.33) .. controls (189.74,65.51) and (185.54,69.71) .. (180.36,69.71) .. controls (175.17,69.71) and (170.97,65.51) .. (170.97,60.33) -- cycle ;
    \draw   (291.54,60.61) .. controls (291.54,55.43) and (295.74,51.23) .. (300.93,51.23) .. controls (306.11,51.23) and (310.31,55.43) .. (310.31,60.61) .. controls (310.31,65.8) and (306.11,70) .. (300.93,70) .. controls (295.74,70) and (291.54,65.8) .. (291.54,60.61) -- cycle ;
    \draw   (331.83,60.61) .. controls (331.83,55.43) and (336.03,51.23) .. (341.21,51.23) .. controls (346.4,51.23) and (350.6,55.43) .. (350.6,60.61) .. controls (350.6,65.8) and (346.4,70) .. (341.21,70) .. controls (336.03,70) and (331.83,65.8) .. (331.83,60.61) -- cycle ;
    \draw   (371.26,60.61) .. controls (371.26,55.43) and (375.46,51.23) .. (380.64,51.23) .. controls (385.83,51.23) and (390.03,55.43) .. (390.03,60.61) .. controls (390.03,65.8) and (385.83,70) .. (380.64,70) .. controls (375.46,70) and (371.26,65.8) .. (371.26,60.61) -- cycle ;
    \draw    (189.74,60.33) -- (210.97,60.61) ;
    \draw    (229.74,60.61) -- (250.97,60.9) ;
    \draw    (270.31,60.33) -- (291.54,60.61) ;
    \draw    (310.6,60.33) -- (331.83,60.61) ;
    \draw    (350.03,60.33) -- (371.26,60.61) ;
    \draw [dashed, color={rgb, 255:red, 0; green, 0; blue, 0 }  ,draw opacity=1 ][line width=0.75]    (220.36,51.23) .. controls (220.57,10.86) and (380.29,11.14) .. (380.64,51.23) ;
    \draw   (279.14,109.19) -- (283.25,109.19) -- (283.25,94.79) -- (291.46,94.79) -- (291.46,109.19) -- (295.57,109.19) -- (287.36,118.79) -- cycle ;
    \draw   (252.11,180.04) .. controls (252.11,174.86) and (256.32,170.66) .. (261.5,170.66) .. controls (266.68,170.66) and (270.89,174.86) .. (270.89,180.04) .. controls (270.89,185.23) and (266.68,189.43) .. (261.5,189.43) .. controls (256.32,189.43) and (252.11,185.23) .. (252.11,180.04) -- cycle ;
    \draw   (211.54,180.33) .. controls (211.54,175.14) and (215.74,170.94) .. (220.93,170.94) .. controls (226.11,170.94) and (230.31,175.14) .. (230.31,180.33) .. controls (230.31,185.51) and (226.11,189.71) .. (220.93,189.71) .. controls (215.74,189.71) and (211.54,185.51) .. (211.54,180.33) -- cycle ;
    \draw   (171.54,180.04) .. controls (171.54,174.86) and (175.74,170.66) .. (180.93,170.66) .. controls (186.11,170.66) and (190.31,174.86) .. (190.31,180.04) .. controls (190.31,185.23) and (186.11,189.43) .. (180.93,189.43) .. controls (175.74,189.43) and (171.54,185.23) .. (171.54,180.04) -- cycle ;
    \draw   (292.11,180.33) .. controls (292.11,175.14) and (296.32,170.94) .. (301.5,170.94) .. controls (306.68,170.94) and (310.89,175.14) .. (310.89,180.33) .. controls (310.89,185.51) and (306.68,189.71) .. (301.5,189.71) .. controls (296.32,189.71) and (292.11,185.51) .. (292.11,180.33) -- cycle ;
    \draw   (332.4,180.33) .. controls (332.4,175.14) and (336.6,170.94) .. (341.79,170.94) .. controls (346.97,170.94) and (351.17,175.14) .. (351.17,180.33) .. controls (351.17,185.51) and (346.97,189.71) .. (341.79,189.71) .. controls (336.6,189.71) and (332.4,185.51) .. (332.4,180.33) -- cycle ;
    \draw   (371.83,180.33) .. controls (371.83,175.14) and (376.03,170.94) .. (381.21,170.94) .. controls (386.4,170.94) and (390.6,175.14) .. (390.6,180.33) .. controls (390.6,185.51) and (386.4,189.71) .. (381.21,189.71) .. controls (376.03,189.71) and (371.83,185.51) .. (371.83,180.33) -- cycle ;
    \draw    (190.31,180.04) -- (211.54,180.33) ;
    \draw [dashed]   (230.31,180.33) -- (251.54,180.61) ;
    \draw    (270.89,180.04) -- (292.11,180.33) ;
    \draw    (311.17,180.04) -- (332.4,180.33) ;
    \draw    (350.6,180.04) -- (371.83,180.33) ;
    \draw [color={rgb, 255:red, 0; green, 0; blue, 0 }  ,draw opacity=1 ][line width=0.75]    (220.93,170.94) .. controls (221.14,130.57) and (380.86,130.86) .. (381.21,170.94) ;
    
    \draw (176,55) node [anchor=north west][inner sep=0.75pt]   [align=left] {{\footnotesize 1}};
    \draw (216,55) node [anchor=north west][inner sep=0.75pt]   [align=left] {{\footnotesize 2}};
    \draw (256.5,55) node [anchor=north west][inner sep=0.75pt]   [align=left] {{\footnotesize 3}};
    \draw (296,55) node [anchor=north west][inner sep=0.75pt]   [align=left] {{\footnotesize 4}};
    \draw (337,55) node [anchor=north west][inner sep=0.75pt]   [align=left] {{\footnotesize 5}};
    \draw (376,55) node [anchor=north west][inner sep=0.75pt]   [align=left] {{\footnotesize 6}};
    \draw (176.5,175) node [anchor=north west][inner sep=0.75pt]   [align=left] {{\footnotesize 1}};
    \draw (216.5,175) node [anchor=north west][inner sep=0.75pt]   [align=left] {{\footnotesize 2}};
    \draw (256.5,175) node [anchor=north west][inner sep=0.75pt]   [align=left] {{\footnotesize 3}};
    \draw (296.5,175) node [anchor=north west][inner sep=0.75pt]   [align=left] {{\footnotesize 4}};
    \draw (336.5,175) node [anchor=north west][inner sep=0.75pt]   [align=left] {{\footnotesize 5}};
    \draw (376.5,175) node [anchor=north west][inner sep=0.75pt]   [align=left] {{\footnotesize 6}};
    \end{tikzpicture}
    }
    \caption{Pósa  Rotation}
\end{figure}
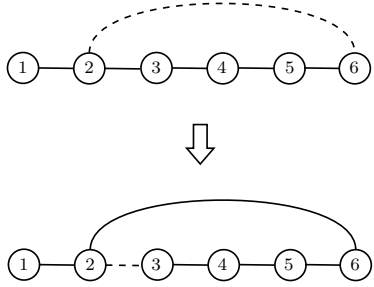

More formally, if there is a path $P$ = $(v_0, v_1,..., v_k)$, $k < N$ in a graph $G$ and there exists a neighbour $v_i$ of $v_k$ in $P$, then $(v_0, v_1,..., v_i, v_k, v_{k-1},..., v_{i+1})$ is a valid path in $G$. A similar idea may be applied to $v_0$. Note that rotations also function when unbounded vertices exist in $P$. A rotation does not affect the number of times a vertex appears, so the number of unbounded vertices remains constant. Furthermore, a rotation will not create consecutive visits of the same vertex since $v_i \ne v_k$ because there are no self-loops. Rotations are the basis of a great number of HCP heuristics due to their ability to extend to a vertex with unvisited neighbours.

Another extension technique is the \textit{cycle extension}. Let $v_0$ and $v_k$ be adjacent dead ends. There exists a vertex $x$ adjacent to a vertex $v_i$ where $0<i<k$, otherwise the graph would not be connected. It follows that the path $(v_{i-1},..., v_0, v_k,..., v_{i+1}, v_i, x)$ with size $|P|+1$ can be created. The justification for cycle extensions with unbounded vertices is very similar to that of rotations.

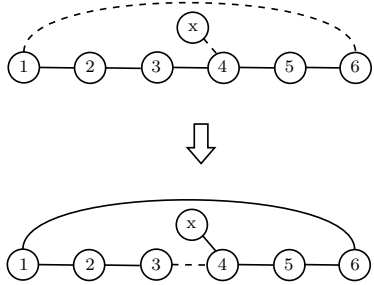
\begin{figure}[htbp]
    \centering
    \tikzset{every picture/.style={line width=0.75pt}} 
        \resizebox{5cm}{!}{
    \begin{tikzpicture}[x=0.75pt,y=0.75pt,yscale=-1,xscale=1]
    
    \draw   (251.54,60.33) .. controls (251.54,55.14) and (255.74,50.94) .. (260.93,50.94) .. controls (266.11,50.94) and (270.31,55.14) .. (270.31,60.33) .. controls (270.31,65.51) and (266.11,69.71) .. (260.93,69.71) .. controls (255.74,69.71) and (251.54,65.51) .. (251.54,60.33) -- cycle ;
    \draw   (210.97,60.61) .. controls (210.97,55.43) and (215.17,51.23) .. (220.36,51.23) .. controls (225.54,51.23) and (229.74,55.43) .. (229.74,60.61) .. controls (229.74,65.8) and (225.54,70) .. (220.36,70) .. controls (215.17,70) and (210.97,65.8) .. (210.97,60.61) -- cycle ;
    \draw   (170.97,60.33) .. controls (170.97,55.14) and (175.17,50.94) .. (180.36,50.94) .. controls (185.54,50.94) and (189.74,55.14) .. (189.74,60.33) .. controls (189.74,65.51) and (185.54,69.71) .. (180.36,69.71) .. controls (175.17,69.71) and (170.97,65.51) .. (170.97,60.33) -- cycle ;
    \draw   (291.54,60.61) .. controls (291.54,55.43) and (295.74,51.23) .. (300.93,51.23) .. controls (306.11,51.23) and (310.31,55.43) .. (310.31,60.61) .. controls (310.31,65.8) and (306.11,70) .. (300.93,70) .. controls (295.74,70) and (291.54,65.8) .. (291.54,60.61) -- cycle ;
    \draw   (331.83,60.61) .. controls (331.83,55.43) and (336.03,51.23) .. (341.21,51.23) .. controls (346.4,51.23) and (350.6,55.43) .. (350.6,60.61) .. controls (350.6,65.8) and (346.4,70) .. (341.21,70) .. controls (336.03,70) and (331.83,65.8) .. (331.83,60.61) -- cycle ;
    \draw   (371.26,60.61) .. controls (371.26,55.43) and (375.46,51.23) .. (380.64,51.23) .. controls (385.83,51.23) and (390.03,55.43) .. (390.03,60.61) .. controls (390.03,65.8) and (385.83,70) .. (380.64,70) .. controls (375.46,70) and (371.26,65.8) .. (371.26,60.61) -- cycle ;
    \draw    (189.74,60.33) -- (210.97,60.61) ;
    \draw    (229.74,60.61) -- (250.97,60.9) ;
    \draw    (270.31,60.33) -- (291.54,60.61) ;
    \draw    (310.6,60.33) -- (331.83,60.61) ;
    \draw    (350.03,60.33) -- (371.26,60.61) ;
    \draw [dashed, color={rgb, 255:red, 0; green, 0; blue, 0 }  ,draw opacity=1 ][line width=0.75]    (180.36,50.94) .. controls (180.57,10.57) and (380.29,11.14) .. (380.64,51.23) ;
    \draw   (279.14,109.19) -- (283.25,109.19) -- (283.25,94.79) -- (291.46,94.79) -- (291.46,109.19) -- (295.57,109.19) -- (287.36,118.79) -- cycle ;
    \draw   (272.97,36.11) .. controls (272.97,30.93) and (277.17,26.73) .. (282.36,26.73) .. controls (287.54,26.73) and (291.74,30.93) .. (291.74,36.11) .. controls (291.74,41.3) and (287.54,45.5) .. (282.36,45.5) .. controls (277.17,45.5) and (272.97,41.3) .. (272.97,36.11) -- cycle ;
    \draw [dashed]   (288.71,42.57) -- (297,52.29) ;
    \draw   (250.97,180.33) .. controls (250.97,175.14) and (255.17,170.94) .. (260.36,170.94) .. controls (265.54,170.94) and (269.74,175.14) .. (269.74,180.33) .. controls (269.74,185.51) and (265.54,189.71) .. (260.36,189.71) .. controls (255.17,189.71) and (250.97,185.51) .. (250.97,180.33) -- cycle ;
    \draw   (210.4,180.61) .. controls (210.4,175.43) and (214.6,171.23) .. (219.79,171.23) .. controls (224.97,171.23) and (229.17,175.43) .. (229.17,180.61) .. controls (229.17,185.8) and (224.97,190) .. (219.79,190) .. controls (214.6,190) and (210.4,185.8) .. (210.4,180.61) -- cycle ;
    \draw   (170.4,180.33) .. controls (170.4,175.14) and (174.6,170.94) .. (179.79,170.94) .. controls (184.97,170.94) and (189.17,175.14) .. (189.17,180.33) .. controls (189.17,185.51) and (184.97,189.71) .. (179.79,189.71) .. controls (174.6,189.71) and (170.4,185.51) .. (170.4,180.33) -- cycle ;
    \draw   (290.97,180.61) .. controls (290.97,175.43) and (295.17,171.23) .. (300.36,171.23) .. controls (305.54,171.23) and (309.74,175.43) .. (309.74,180.61) .. controls (309.74,185.8) and (305.54,190) .. (300.36,190) .. controls (295.17,190) and (290.97,185.8) .. (290.97,180.61) -- cycle ;
    \draw   (331.26,180.61) .. controls (331.26,175.43) and (335.46,171.23) .. (340.64,171.23) .. controls (345.83,171.23) and (350.03,175.43) .. (350.03,180.61) .. controls (350.03,185.8) and (345.83,190) .. (340.64,190) .. controls (335.46,190) and (331.26,185.8) .. (331.26,180.61) -- cycle ;
    \draw   (370.69,180.61) .. controls (370.69,175.43) and (374.89,171.23) .. (380.07,171.23) .. controls (385.26,171.23) and (389.46,175.43) .. (389.46,180.61) .. controls (389.46,185.8) and (385.26,190) .. (380.07,190) .. controls (374.89,190) and (370.69,185.8) .. (370.69,180.61) -- cycle ;
    \draw    (189.17,180.33) -- (210.4,180.61) ;
    \draw    (229.17,180.61) -- (250.4,180.9) ;
    \draw [dashed]   (269.74,180.33) -- (290.97,180.61) ;
    \draw    (310.03,180.33) -- (331.26,180.61) ;
    \draw    (349.46,180.33) -- (370.69,180.61) ;
    \draw [color={rgb, 255:red, 0; green, 0; blue, 0 }  ,draw opacity=1 ][line width=0.75]    (179.79,170.94) .. controls (180,130.57) and (379.71,131.14) .. (380.07,171.23) ;
    \draw   (272.4,156.11) .. controls (272.4,150.93) and (276.6,146.73) .. (281.79,146.73) .. controls (286.97,146.73) and (291.17,150.93) .. (291.17,156.11) .. controls (291.17,161.3) and (286.97,165.5) .. (281.79,165.5) .. controls (276.6,165.5) and (272.4,161.3) .. (272.4,156.11) -- cycle ;
    \draw    (288.14,162.57) -- (296.43,172.29) ;
    
    \draw (176,55) node [anchor=north west][inner sep=0.75pt]   [align=left] {{\footnotesize 1}};
    \draw (216,55) node [anchor=north west][inner sep=0.75pt]   [align=left] {{\footnotesize 2}};
    \draw (256,55) node [anchor=north west][inner sep=0.75pt]   [align=left] {{\footnotesize 3}};
    \draw (296,55) node [anchor=north west][inner sep=0.75pt]   [align=left] {{\footnotesize 4}};
    \draw (336.5,55) node [anchor=north west][inner sep=0.75pt]   [align=left] {{\footnotesize 5}};
    \draw (376,55) node [anchor=north west][inner sep=0.75pt]   [align=left] {{\footnotesize 6}};
    \draw (278,32) node [anchor=north west][inner sep=0.75pt]  [font=\footnotesize] [align=left] {{\footnotesize x}};
    \draw (176,175) node [anchor=north west][inner sep=0.75pt]   [align=left] {{\footnotesize 1}};
    \draw (216,175) node [anchor=north west][inner sep=0.75pt]   [align=left] {{\footnotesize 2}};
    \draw (256,175) node [anchor=north west][inner sep=0.75pt]   [align=left] {{\footnotesize 3}};
    \draw (296,175) node [anchor=north west][inner sep=0.75pt]   [align=left] {{\footnotesize 4}};
    \draw (336,175) node [anchor=north west][inner sep=0.75pt]   [align=left] {{\footnotesize 5}};
    \draw (376,175) node [anchor=north west][inner sep=0.75pt]   [align=left] {{\footnotesize 6}};
    \draw (277,152) node [anchor=north west][inner sep=0.75pt]  [font=\footnotesize] [align=left] {{\footnotesize x}};

    \end{tikzpicture}
    }
    \caption{Cycle Extension}
    \label{fig3}
\end{figure}

A combination of rotations and subsequently checking for cycle extensions has been shown to be effective in finding Hamiltonian Paths and Cycles in undirected graphs \cite{HAM} \cite{Sparse-HAM}. 

\subsubsection{Dijkstra's Algorithm}
As mentioned in Section \ref{BFS}, Dijkstra's is a greedy, single-source shortest path algorithm utilized in weighted graphs. Below is a description of its implementation.

Let $d[v]$ represent the current distance from the source $s$ to a vertex $v$ with all its indices except for $s$ (which is set to $0$) initialized to an arbitrarily large constant. Additionally, let $Q$ contain the current set of vertices that the BFS is considering to travel to with $s$ as the only vertex in $Q$ initially. At every instance of the BFS, a vertex $u$ with $d[u] = \text{min}_{v\in Q}(d[v])$ is selected and deleted from $Q$.  Denote $w_{ab}$ as the edge weight between vertices $a$ and $b$. For all neighbours $k$ of $u$, if $d[u] + w_{uk} < d[k]$, then $k$ is added to $Q$ and $d[k]$ is updated. The process terminates once $Q$ is empty.

The implementation greatly resembles the BFS, with the primary differences being the use of a Fibonacci heap in $Q$ to find the current minimum distance and the inclusion of edge weights other than 1. Once a vertex $v$ is popped from $Q$, $d[v]$ is finalized and it is unnecessary to consider another iteration at $v$, hence the inclusion of a visited array. The use of a Fibonacci heap or a similar data structure to maintain a non-decreasing sequence of distances in $Q$ results in a time complexity of $O(E\log(N))$.

\subsubsection{0-1 BFS}
If the edge weights in a graph are limited to $0$ and $1$, a more efficient variation of Dijkstra's called the ``0-1 BFS" can be implemented.

Consider the set of variables used in the previous section. Let $Q=(v_0, v_1,..., v_m)$ such that $d[v_i] \le d[v_{i+1}]$, $0\le i<m$ and $d[v_m]-d[v_0]\le 1$. Since Q is sorted in non-decreasing order, it is optimal for $u$ to be $v_0$. Recall that for an adjacent vertex $k$, $d[k]$ may be updated to $d[u]+w_{uk}$ before being added to $Q$. If $w_{uk}=0$, adding $k$ to the front of $Q$ maintains a non-decreasing sequence by ``replacing" $v$ with a vertex having an identical distance. Otherwise, $w_{uk}=1$ and $d[k]=d[u]+1=d[v_0]+1 \ge d[v_m]$. Thus, adding $k$ to the back of $Q$ maintains a non-decreasing sequence.

$Q$ begins with a single vertex $s$, which is trivially non-decreasing and having $d[s]-d[s]=0\le 1$. As a result, it is always possible to maintain a sorted $Q$ through a series of constant-time operations, eliminating the need to use a Fibonacci Heap that requires a logarithmic factor. A visited array is unnecessary since a vertex can only appear in $Q$ twice. Thus, the time complexity is reduced to $O(E)$.

\begin{algorithm}[!htbp]
\DontPrintSemicolon
\SetKwInput{KwInput}{Input} 

\KwInput{Adjacency list $adj[N][]$}
\BlankLine

\SetKwFunction{FMain}{Main}
\SetKwFunction{fDijk}{0-1BFS}
\SetKwProg{Fn}{void}{:}{}

\Fn{\fDijk{$s$}}{

    \While{Q is not empty}{
        u $\leftarrow$ Q.front()\;
        Q.pop()\;
        \BlankLine
        
        \For{k in adj[v]}{
            \uIf{$d[u] + w_{uk} < d[k]$}{
                d[k] $\leftarrow$ $d[u] + w_{uk}$\;
                \lIf{$w_{uk}$ is 0}{Q.pushFront(k)}
                \lElse{Q.pushBack(k)}
            }
        }
    }
}
\caption{0-1 BFS}
\end{algorithm}
It is worth mentioning that with slight modifications, 0-1 BFS can be generalized to graphs with edge weights bounded by $C$, where the time complexity is $O(EC)$. This is known as Dial's Algorithm, and it is unnecessary for the scope of the heuristic.

\subsubsection{Cut Vertices} \label{section:cutVertices}
As shown in theorem \ref{thm1}, every cut vertex must be unbounded for a k-Unbounded Hamiltonian cycle to exist. The number of cut vertices serves as a rough estimate of the accuracy of the heuristic in a Section \ref{section:randomGraphs} due to its ability to serve as a lower bound for the number of required unbounded vertices. It is important to note that making all cut vertices unbounded does not necessarily guarantee the existence of a k-Unbounded Hamiltonian Cycle, so the number of cut vertices is not a tight lower bound. 

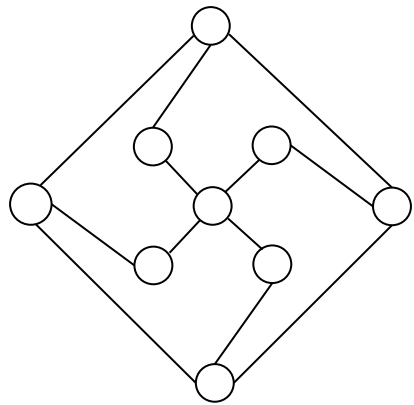
\begin{figure}[htbp] \label{figNoCut}
    \centering
        \tikzset{every picture/.style={line width=0.75pt}} 

        \begin{tikzpicture}[x=0.75pt,y=0.75pt,yscale=-1,xscale=1]
        
        \draw   (328.37,86.18) .. controls (328.37,80.95) and (332.61,76.71) .. (337.84,76.71) .. controls (343.06,76.71) and (347.3,80.95) .. (347.3,86.18) .. controls (347.3,91.41) and (343.06,95.64) .. (337.84,95.64) .. controls (332.61,95.64) and (328.37,91.41) .. (328.37,86.18) -- cycle ;
        \draw   (328.65,145.97) .. controls (328.65,140.74) and (332.89,136.51) .. (338.11,136.51) .. controls (343.34,136.51) and (347.58,140.74) .. (347.58,145.97) .. controls (347.58,151.2) and (343.34,155.44) .. (338.11,155.44) .. controls (332.89,155.44) and (328.65,151.2) .. (328.65,145.97) -- cycle ;
        \draw   (387.97,145.36) .. controls (387.97,140.14) and (392.21,135.9) .. (397.44,135.9) .. controls (402.66,135.9) and (406.9,140.14) .. (406.9,145.36) .. controls (406.9,150.59) and (402.66,154.83) .. (397.44,154.83) .. controls (392.21,154.83) and (387.97,150.59) .. (387.97,145.36) -- cycle ;
        \draw   (358.2,116.05) .. controls (358.2,110.82) and (362.44,106.59) .. (367.66,106.59) .. controls (372.89,106.59) and (377.13,110.82) .. (377.13,116.05) .. controls (377.13,121.28) and (372.89,125.51) .. (367.66,125.51) .. controls (362.44,125.51) and (358.2,121.28) .. (358.2,116.05) -- cycle ;
        \draw   (357.29,25.48) .. controls (357.29,20.25) and (361.52,16.01) .. (366.75,16.01) .. controls (371.98,16.01) and (376.21,20.25) .. (376.21,25.48) .. controls (376.21,30.71) and (371.98,34.94) .. (366.75,34.94) .. controls (361.52,34.94) and (357.29,30.71) .. (357.29,25.48) -- cycle ;
        \draw   (266.6,115.16) .. controls (266.6,109.44) and (271.24,104.8) .. (276.96,104.8) .. controls (282.68,104.8) and (287.31,109.44) .. (287.31,115.16) .. controls (287.31,120.88) and (282.68,125.51) .. (276.96,125.51) .. controls (271.24,125.51) and (266.6,120.88) .. (266.6,115.16) -- cycle ;
        \draw   (387.63,85.54) .. controls (387.63,80.31) and (391.87,76.07) .. (397.09,76.07) .. controls (402.32,76.07) and (406.56,80.31) .. (406.56,85.54) .. controls (406.56,90.76) and (402.32,95) .. (397.09,95) .. controls (391.87,95) and (387.63,90.76) .. (387.63,85.54) -- cycle ;
        \draw [line width=0.75]    (281.8,105.6) -- (358.2,30.4) ;
        \draw    (391,92.8) -- (374.2,108.8) ;
        \draw    (359.29,205.08) -- (279.4,125.2) ;
        \draw    (346.2,139.6) -- (361.4,123.2) ;
        \draw    (360.2,110.4) -- (344.2,93.2) ;
        \draw    (337.84,76.71) -- (366.75,34.94) ;
        \draw    (375.8,29.6) -- (457.15,106.81) ;
        \draw    (457.15,125.74) -- (378.21,205.08) ;
        \draw    (375.4,122.4) -- (392.6,138) ;
        \draw   (359.29,205.08) .. controls (359.29,199.85) and (363.52,195.61) .. (368.75,195.61) .. controls (373.98,195.61) and (378.21,199.85) .. (378.21,205.08) .. controls (378.21,210.31) and (373.98,214.54) .. (368.75,214.54) .. controls (363.52,214.54) and (359.29,210.31) .. (359.29,205.08) -- cycle ;
        \draw   (447.69,116.28) .. controls (447.69,111.05) and (451.92,106.81) .. (457.15,106.81) .. controls (462.38,106.81) and (466.61,111.05) .. (466.61,116.28) .. controls (466.61,121.51) and (462.38,125.74) .. (457.15,125.74) .. controls (451.92,125.74) and (447.69,121.51) .. (447.69,116.28) -- cycle ;
        \draw    (406.56,85.54) -- (447.69,116.28) ;
        \draw    (397.44,154.83) -- (368.75,195.61) ;
        \draw    (328.65,145.97) -- (287.31,115.16) ;
        \end{tikzpicture}
    \caption{Graph without cut vertices but no Hamiltonian Path or Cycle}
\end{figure}

\subsection{The Heuristic}
The Unbounded Heuristic can be summarized by the following steps:

\begin{enumerate}
    \item Mark every cut vertex as unbounded 
    \item Begin at a start vertex with the highest degree
    \item Select neighbouring vertices to add to the path until both end vertices are dead ends
    \item If the path is not a k-Unbounded Hamiltonian Path, reroute an end vertex to a vertex that is not a dead end
    \item Repeat steps 3 and 4 until the path is a k-Unbounded Hamiltonian Path
    \item Convert the k-Unbounded Hamiltonian Path into a k-Unbounded Hamiltonian Cycle
\end{enumerate}

\subsubsection{Marking Cut Vertices as Unbounded}
Since every cut vertex must be unbounded, it is possible to preemptively mark every cut vertex as unbounded to allow for a greater number of options throughout earlier stages of the algorithm. The most straightforward approach to checking whether a vertex is a cut-vertex is by removing it from the graph and running a BFS/DFS to determine whether the remaining graph is connected.

In the implementation below, curCut is the vertex that is being considered and the check function returns true if there is more than one component after curCut is removed. To identify all cut vertices, the check function is run over all vertices. 

\begin{algorithm}[!htbp]
\DontPrintSemicolon
\SetKwInput{KwInput}{Input} 

\KwInput{Adjacency list $adj[N][]$, $curCut$}
\BlankLine

\SetKwFunction{FDfs}{Dfs}
\SetKwFunction{FCheck}{check}
\SetKwProg{Fn}{void}{:}{}
\SetKwProg{TF}{bool}{:}{}

\Fn{\FDfs{$v$}}{
    vis[v] $\leftarrow$ true\;
    \For{k in adj[v]}{
        \uIf{$k \ne $ curCut and !vis[k]}{Dfs(k)\;} 
    }
}
\BlankLine

\TF{\FCheck{curCut}}{
    inc  $\leftarrow$ 0\;
    \For{$i\gets1$ \KwTo $N$}{
        \uIf{i $\ne$ curCut and !vis[i]}{
            Dfs(i)\;
            inc++\;
            \lIf{inc $>$ 1}{return true}
        }
    }
    return false\;
}

\caption{Cut-Vertex Check}
\end{algorithm}

\subsubsection{Greedy Vertex Selection}
A vertex with maximum degree is initially chosen to increase the chances of extending the path if a dead end is reached by extending from the start vertex- this will also give more options for rotations and unbounded vertices described in the next section. Furthermore, the lowest degree neighbour is chosen at every instance. Intuitively, choosing the neighbour with the least number of edges to unvisited vertices maximizes the number of ``usable" edges, or edges connecting two unvisited vertices.

\newpage
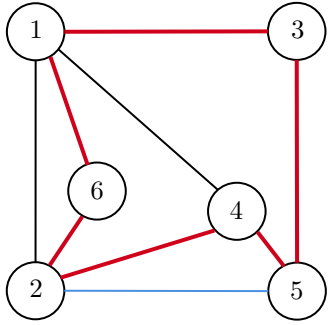
\begin{figure}[htbp] 
    \centering
        \tikzset{every picture/.style={line width=0.75pt}} 

        \begin{tikzpicture}[x=0.75pt,y=0.75pt,yscale=-1,xscale=1]
        
\draw   (280.65,205.19) .. controls (280.65,197.26) and (287.08,190.84) .. (295,190.84) .. controls (302.93,190.84) and (309.36,197.26) .. (309.36,205.19) .. controls (309.36,213.12) and (302.93,219.55) .. (295,219.55) .. controls (287.08,219.55) and (280.65,213.12) .. (280.65,205.19) -- cycle ;
\draw   (280.7,74.91) .. controls (280.7,66.98) and (287.12,60.56) .. (295.05,60.56) .. controls (302.98,60.56) and (309.41,66.98) .. (309.41,74.91) .. controls (309.41,82.84) and (302.98,89.26) .. (295.05,89.26) .. controls (287.12,89.26) and (280.7,82.84) .. (280.7,74.91) -- cycle ;
\draw   (410.95,74.71) .. controls (410.95,66.79) and (417.38,60.36) .. (425.31,60.36) .. controls (433.23,60.36) and (439.66,66.79) .. (439.66,74.71) .. controls (439.66,82.64) and (433.23,89.07) .. (425.31,89.07) .. controls (417.38,89.07) and (410.95,82.64) .. (410.95,74.71) -- cycle ;
\draw   (411.13,205.38) .. controls (411.13,197.45) and (417.56,191.03) .. (425.49,191.03) .. controls (433.42,191.03) and (439.84,197.45) .. (439.84,205.38) .. controls (439.84,213.31) and (433.42,219.73) .. (425.49,219.73) .. controls (417.56,219.73) and (411.13,213.31) .. (411.13,205.38) -- cycle ;
\draw   (381.07,165.23) .. controls (381.07,157.3) and (387.5,150.87) .. (395.43,150.87) .. controls (403.36,150.87) and (409.78,157.3) .. (409.78,165.23) .. controls (409.78,173.16) and (403.36,179.58) .. (395.43,179.58) .. controls (387.5,179.58) and (381.07,173.16) .. (381.07,165.23) -- cycle ;
\draw   (311.32,155.05) .. controls (311.32,147.12) and (317.74,140.69) .. (325.67,140.69) .. controls (333.6,140.69) and (340.02,147.12) .. (340.02,155.05) .. controls (340.02,162.97) and (333.6,169.4) .. (325.67,169.4) .. controls (317.74,169.4) and (311.32,162.97) .. (311.32,155.05) -- cycle ;
\draw    (295.05,89.26) -- (295,190.84) ;
\draw [color={rgb, 255:red, 74; green, 144; blue, 226 }  ,draw opacity=1 ]   (309.36,205.19) -- (411.13,205.38) ;
\draw [color={rgb, 255:red, 208; green, 2; blue, 27 }  ,draw opacity=1 ][line width=1.5]    (309.41,74.91) -- (410.95,74.71) ;
\draw [color={rgb, 255:red, 208; green, 2; blue, 27 }  ,draw opacity=1 ][fill={rgb, 255:red, 208; green, 2; blue, 27 }  ,fill opacity=1 ][line width=1.5]    (425.31,89.07) -- (425.49,191.03) ;
\draw    (306.56,83.89) -- (386.11,154.56) ;
\draw [color={rgb, 255:red, 208; green, 2; blue, 27 }  ,draw opacity=1 ][line width=1.5]    (405.75,175.63) -- (419,192.88) ;
\draw [color={rgb, 255:red, 208; green, 2; blue, 27 }  ,draw opacity=1 ][line width=1.5]    (302.05,192.38) -- (318.3,167.88) ;
\draw [color={rgb, 255:red, 208; green, 2; blue, 27 }  ,draw opacity=1 ][line width=1.5]    (307.5,198.63) -- (384.25,174.88) ;
\draw [color={rgb, 255:red, 208; green, 2; blue, 27 }  ,draw opacity=1 ][line width=1.5]    (302.4,87.5) -- (321,141.9) ;

\draw (290.5,68) node [anchor=north west][inner sep=0.75pt]   [align=left] {1};
\draw (421,68) node [anchor=north west][inner sep=0.75pt]   [align=left] {3};
\draw (420.5,199.5) node [anchor=north west][inner sep=0.75pt]   [align=left] {5};
\draw (290.5,198) node [anchor=north west][inner sep=0.75pt]   [align=left] {2};
\draw (321,149) node [anchor=north west][inner sep=0.75pt]   [align=left] {6};
\draw (390.5,159) node [anchor=north west][inner sep=0.75pt]   [align=left] {4};
        \end{tikzpicture}
    \caption{Low degree selection}
    \label{figLowDegree}
\end{figure}

Consider the graph shown in Figure \ref{figLowDegree}. If the path $P$ is $(1, 3, 5)$, then $P$ can extend to $2$ or $4$. $4$ is adjacent to $\{2\}$ and has degree $1$ while $2$ is adjacent to $\{4, 6\}$ and has degree 2. With the low degree selection process, $P$ would extend to $4$ and proceed to find the Hamiltonian Cycle $(1, 3, 5, 4, 2, 6)$. If instead $P$ was extended to $2$, $P$ would reach a dead end after traveling to $6$ or $4$ due to a lack of usable edges.  

If a dead end is reached, the start vertex is checked: if it is not a dead end, then the path is reversed and the greedy vertex selection continues. 

\subsubsection{Rerouting}
If both end vertices are dead ends and the path is not a k-Unbounded Hamiltonian Path, the path will be rerouted to an unvisited vertex through a series of rotations, cycle extensions, and assignments of unbounded vertices. The objective is to minimize the number of newly assigned unbounded vertices used to reroute the path. 

This section proposes a novel idea to escape a dead end when unbounded vertices are allowed. Consider the subgraph of $G$ composed of the vertices and edges of the path $P = (v_0, v_1,..., v_k)$. A new, weighted graph may be formed such that all edges from the end vertices to other vertices on the path $v_i$, $0<i<k$ have a weight of $1$ if $v_i$ is bounded and a weight of $0$ if $v_i$ is unbounded. A weight corresponds to whether or not moving along an edge contributes a new unbounded vertex. $v_i$ can be prepended to $P$ if $v_i$ is adjacent to $v_0$ and appended to $P$ if $v_i$ is adjacent to $v_k$, effectively serving to replace endpoints. 
\begin{figure}[htbp] \label{figRoute1}
    \centering
        \tikzset{every picture/.style={line width=0.75pt}} 

        \begin{tikzpicture}[x=0.75pt,y=0.75pt,yscale=-1,xscale=1]
        
\draw   (210.53,130.15) .. controls (210.53,124.88) and (214.8,120.61) .. (220.07,120.61) .. controls (225.34,120.61) and (229.61,124.88) .. (229.61,130.15) .. controls (229.61,135.42) and (225.34,139.69) .. (220.07,139.69) .. controls (214.8,139.69) and (210.53,135.42) .. (210.53,130.15) -- cycle ;
\draw   (260.31,130.15) .. controls (260.31,124.88) and (264.58,120.61) .. (269.85,120.61) .. controls (275.12,120.61) and (279.39,124.88) .. (279.39,130.15) .. controls (279.39,135.42) and (275.12,139.69) .. (269.85,139.69) .. controls (264.58,139.69) and (260.31,135.42) .. (260.31,130.15) -- cycle ;
\draw   (310.53,130.15) .. controls (310.53,124.88) and (314.8,120.61) .. (320.07,120.61) .. controls (325.34,120.61) and (329.61,124.88) .. (329.61,130.15) .. controls (329.61,135.42) and (325.34,139.69) .. (320.07,139.69) .. controls (314.8,139.69) and (310.53,135.42) .. (310.53,130.15) -- cycle ;
\draw   (360.53,129.93) .. controls (360.53,124.66) and (364.8,120.39) .. (370.07,120.39) .. controls (375.34,120.39) and (379.61,124.66) .. (379.61,129.93) .. controls (379.61,135.2) and (375.34,139.47) .. (370.07,139.47) .. controls (364.8,139.47) and (360.53,135.2) .. (360.53,129.93) -- cycle ;
\draw  [fill={rgb, 255:red, 248; green, 231; blue, 28 }  ,fill opacity=1 ] (410.22,129.93) .. controls (410.22,124.66) and (414.49,120.39) .. (419.76,120.39) .. controls (425.03,120.39) and (429.31,124.66) .. (429.31,129.93) .. controls (429.31,135.2) and (425.03,139.47) .. (419.76,139.47) .. controls (414.49,139.47) and (410.22,135.2) .. (410.22,129.93) -- cycle ;
\draw   (460.44,130.15) .. controls (460.44,124.88) and (464.72,120.61) .. (469.99,120.61) .. controls (475.26,120.61) and (479.53,124.88) .. (479.53,130.15) .. controls (479.53,135.42) and (475.26,139.69) .. (469.99,139.69) .. controls (464.72,139.69) and (460.44,135.42) .. (460.44,130.15) -- cycle ;
\draw   (510.67,130.15) .. controls (510.67,124.88) and (514.94,120.61) .. (520.21,120.61) .. controls (525.48,120.61) and (529.75,124.88) .. (529.75,130.15) .. controls (529.75,135.42) and (525.48,139.69) .. (520.21,139.69) .. controls (514.94,139.69) and (510.67,135.42) .. (510.67,130.15) -- cycle ;
\draw   (560.89,129.93) .. controls (560.89,124.66) and (565.16,120.39) .. (570.43,120.39) .. controls (575.7,120.39) and (579.97,124.66) .. (579.97,129.93) .. controls (579.97,135.2) and (575.7,139.47) .. (570.43,139.47) .. controls (565.16,139.47) and (560.89,135.2) .. (560.89,129.93) -- cycle ;
\draw    (229.61,130.15) -- (260.31,130.15) ;
\draw    (279.39,130.15) -- (310.53,130.15) ;
\draw    (329.61,130.15) -- (360.53,129.93) ;
\draw    (379.61,129.93) -- (410.22,129.93) ;
\draw    (429.31,129.93) -- (460.44,130.15) ;
\draw    (479.53,130.15) -- (510.67,130.15) ;
\draw    (529.75,130.15) -- (560.89,129.93) ;
\draw [color={rgb, 255:red, 74; green, 144; blue, 226 }  ,draw opacity=1 ]   (220.07,120.61) .. controls (220.2,90.6) and (320.33,90.11) .. (320.07,120.61) ;
\draw [color={rgb, 255:red, 76; green, 42; blue, 12 }  ,draw opacity=1 ]   (269.85,120.61) .. controls (270.27,80.09) and (570.27,80.27) .. (570.43,120.39) ;
\draw [color={rgb, 255:red, 245; green, 166; blue, 35 }  ,draw opacity=1 ]   (220.07,120.61) .. controls (220.33,69.89) and (420.11,69.89) .. (419.76,120.39) ;
\draw (319.2,73) node [anchor=north west][inner sep=0.75pt]   [align=left] {{\scriptsize 0}};
\draw (416.4,80) node [anchor=north west][inner sep=0.75pt]   [align=left] {{\scriptsize 1}};
\draw (269.2,88) node [anchor=north west][inner sep=0.75pt]   [align=left] {{\scriptsize 1}};
\draw (241,120) node [anchor=north west][inner sep=0.75pt]   [align=left] {{\scriptsize 1}};
\draw (542.8,120) node [anchor=north west][inner sep=0.75pt]   [align=left] {{\scriptsize 1}};
\draw (213.5,126) node [anchor=north west][inner sep=0.75pt]  [font=\tiny] [align=left] {$\displaystyle v_{0}$};
\draw (263.5,126) node [anchor=north west][inner sep=0.75pt]  [font=\tiny] [align=left] {$\displaystyle v_{1}$};
\draw (313.5,126) node [anchor=north west][inner sep=0.75pt]  [font=\tiny] [align=left] {$\displaystyle v_{2}$};
\draw (363.5,126) node [anchor=north west][inner sep=0.75pt]  [font=\tiny] [align=left] {$\displaystyle v_{3}$};
\draw (413,126) node [anchor=north west][inner sep=0.75pt]  [font=\tiny] [align=left] {$\displaystyle v_{4}$};
\draw (463,126) node [anchor=north west][inner sep=0.75pt]  [font=\tiny] [align=left] {$\displaystyle v_{5}$};
\draw (513.5,126) node [anchor=north west][inner sep=0.75pt]  [font=\tiny] [align=left] {$\displaystyle v_{6}$};
\draw (564.5,126) node [anchor=north west][inner sep=0.75pt]  [font=\tiny] [align=left] {$\displaystyle v_{7}$};

        \end{tikzpicture}
    \caption{New graph with assigned edge weights. The shaded vertex is unbounded.}
\end{figure}
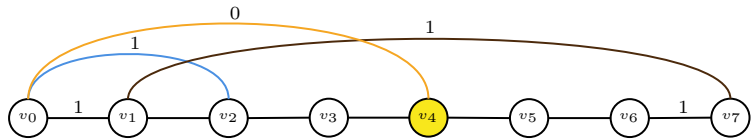

Furthermore, if a vertex can become an endpoint through a single rotation, then an edge with weight $0$ may be added between this vertex and the endpoint of rotation. Recall that if there exists an edge from $v_i$ to $v_k$, $v_{i+1}$ may become an endpoint by reversing the order of visits from $i+1$ to $k$. Similarly, an edge connecting $v_0$ and $v_i$ can be rotated about the start vertex to reroute the start vertex to $v_{i-1}$

\begin{figure}[htbp] \label{figRoute0.5}
    \centering
        \tikzset{every picture/.style={line width=0.75pt}} 

        \begin{tikzpicture}[x=0.75pt,y=0.75pt,yscale=-1,xscale=1]
        
\draw   (210.53,130.15) .. controls (210.53,124.88) and (214.8,120.61) .. (220.07,120.61) .. controls (225.34,120.61) and (229.61,124.88) .. (229.61,130.15) .. controls (229.61,135.42) and (225.34,139.69) .. (220.07,139.69) .. controls (214.8,139.69) and (210.53,135.42) .. (210.53,130.15) -- cycle ;
\draw   (260.31,130.15) .. controls (260.31,124.88) and (264.58,120.61) .. (269.85,120.61) .. controls (275.12,120.61) and (279.39,124.88) .. (279.39,130.15) .. controls (279.39,135.42) and (275.12,139.69) .. (269.85,139.69) .. controls (264.58,139.69) and (260.31,135.42) .. (260.31,130.15) -- cycle ;
\draw   (310.53,130.15) .. controls (310.53,124.88) and (314.8,120.61) .. (320.07,120.61) .. controls (325.34,120.61) and (329.61,124.88) .. (329.61,130.15) .. controls (329.61,135.42) and (325.34,139.69) .. (320.07,139.69) .. controls (314.8,139.69) and (310.53,135.42) .. (310.53,130.15) -- cycle ;
\draw   (360.53,129.93) .. controls (360.53,124.66) and (364.8,120.39) .. (370.07,120.39) .. controls (375.34,120.39) and (379.61,124.66) .. (379.61,129.93) .. controls (379.61,135.2) and (375.34,139.47) .. (370.07,139.47) .. controls (364.8,139.47) and (360.53,135.2) .. (360.53,129.93) -- cycle ;
\draw  [fill={rgb, 255:red, 248; green, 231; blue, 28}  ,fill opacity=1 ] (410.22,129.93) .. controls (410.22,124.66) and (414.49,120.39) .. (419.76,120.39) .. controls (425.03,120.39) and (429.31,124.66) .. (429.31,129.93) .. controls (429.31,135.2) and (425.03,139.47) .. (419.76,139.47) .. controls (414.49,139.47) and (410.22,135.2) .. (410.22,129.93) -- cycle ;
\draw   (460.44,130.15) .. controls (460.44,124.88) and (464.72,120.61) .. (469.99,120.61) .. controls (475.26,120.61) and (479.53,124.88) .. (479.53,130.15) .. controls (479.53,135.42) and (475.26,139.69) .. (469.99,139.69) .. controls (464.72,139.69) and (460.44,135.42) .. (460.44,130.15) -- cycle ;
\draw   (510.67,130.15) .. controls (510.67,124.88) and (514.94,120.61) .. (520.21,120.61) .. controls (525.48,120.61) and (529.75,124.88) .. (529.75,130.15) .. controls (529.75,135.42) and (525.48,139.69) .. (520.21,139.69) .. controls (514.94,139.69) and (510.67,135.42) .. (510.67,130.15) -- cycle ;
\draw   (560.89,129.93) .. controls (560.89,124.66) and (565.16,120.39) .. (570.43,120.39) .. controls (575.7,120.39) and (579.97,124.66) .. (579.97,129.93) .. controls (579.97,135.2) and (575.7,139.47) .. (570.43,139.47) .. controls (565.16,139.47) and (560.89,135.2) .. (560.89,129.93) -- cycle ;
\draw    (229.61,130.15) -- (260.31,130.15) ;
\draw    (279.39,130.15) -- (310.53,130.15) ;
\draw    (329.61,130.15) -- (360.53,129.93) ;
\draw    (379.61,129.93) -- (410.22,129.93) ;
\draw    (429.31,129.93) -- (460.44,130.15) ;
\draw    (479.53,130.15) -- (510.67,130.15) ;
\draw    (529.75,130.15) -- (560.89,129.93) ;
\draw [color={rgb, 255:red, 74; green, 144; blue, 226 }  ,draw opacity=1 ]   (220.07,120.61) .. controls (220.2,90.6) and (320.33,90.11) .. (320.07,120.61) ;
\draw [color={rgb, 255:red, 76; green, 42; blue, 12 }  ,draw opacity=1 ]   (269.85,120.61) .. controls (270.27,80.09) and (570.27,80.27) .. (570.43,120.39) ;
\draw [color={rgb, 255:red, 245; green, 166; blue, 35 }  ,draw opacity=1 ]   (220.07,120.61) .. controls (220.33,69.89) and (420.11,69.89) .. (419.76,120.39) ;
\draw [dashed, color={rgb, 255:red, 74; green, 144; blue, 226 }  ,draw opacity=1 ]   (220.07,139.69) .. controls (220.25,160.13) and (270,160.13) .. (269.85,139.69) ;
\draw [dashed, color={rgb, 255:red, 245; green, 166; blue, 35 }  ,draw opacity=1 ]   (220.07,139.69) .. controls (220.5,180.13) and (370,180.13) .. (370.07,139.47) ;
\draw [dashed, color={rgb, 255:red, 76; green, 42; blue, 12 }  ,draw opacity=1 ]   (320.07,139.69) .. controls (320,170.63) and (570.5,170.63) .. (570.43,139.47) ;
\draw (319.2,73) node [anchor=north west][inner sep=0.75pt]   [align=left] {{\scriptsize 0}};
\draw (416.4,80) node [anchor=north west][inner sep=0.75pt]   [align=left] {{\scriptsize 1}};
\draw (269.2,88) node [anchor=north west][inner sep=0.75pt]   [align=left] {{\scriptsize 1}};
\draw (242,156) node [anchor=north west][inner sep=0.75pt]  [font=\scriptsize] [align=left] {0};
\draw (289.2,171) node [anchor=north west][inner sep=0.75pt]   [align=left] {{\scriptsize 0}};
\draw (441.2,165) node [anchor=north west][inner sep=0.75pt]   [align=left] {{\scriptsize 0}};
\draw (241,120) node [anchor=north west][inner sep=0.75pt]   [align=left] {{\scriptsize 1}};
\draw (542.8,120) node [anchor=north west][inner sep=0.75pt]   [align=left] {{\scriptsize 1}};
\draw (213.5,126) node [anchor=north west][inner sep=0.75pt]  [font=\tiny] [align=left] {$\displaystyle v_{0}$};
\draw (263.5,126) node [anchor=north west][inner sep=0.75pt]  [font=\tiny] [align=left] {$\displaystyle v_{1}$};
\draw (313.5,126) node [anchor=north west][inner sep=0.75pt]  [font=\tiny] [align=left] {$\displaystyle v_{2}$};
\draw (363.5,126) node [anchor=north west][inner sep=0.75pt]  [font=\tiny] [align=left] {$\displaystyle v_{3}$};
\draw (413,126) node [anchor=north west][inner sep=0.75pt]  [font=\tiny] [align=left] {$\displaystyle v_{4}$};
\draw (463,126) node [anchor=north west][inner sep=0.75pt]  [font=\tiny] [align=left] {$\displaystyle v_{5}$};
\draw (513.5,126) node [anchor=north west][inner sep=0.75pt]  [font=\tiny] [align=left] {$\displaystyle v_{6}$};
\draw (564.5,126) node [anchor=north west][inner sep=0.75pt]  [font=\tiny] [align=left] {$\displaystyle v_{7}$};

        \end{tikzpicture}
    \caption{New graph with rotations (dashed lines) added}
\end{figure}
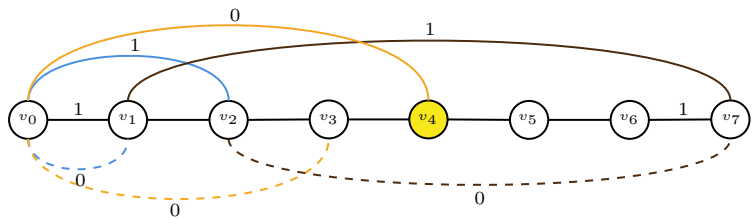

If a vertex can become an endpoint by following a sequence of edges, then the number of new unbounded vertices required to make this vertex an endpoint is the sum of the weights of the edges. This idea motivates running a shortest path algorithm to find a vertex with the lowest distance with degree greater than 0. The algorithm is guaranteed to find such a vertex since the edges connecting $P$, $v_{0}v_{1}, v_{1}v_{2},..., v_{k-1}v_{k}$, ensure that any vertex on $P$ can be reached. This method is not perfect, as once a vertex is visited by the shortest path algorithm, it cannot be visited again. Since rotations dynamically alter the order of the path, a different series of rotations may lead to a different set of edges. However, the majority of the edges are ``fixed" since they don't use rotations and are part of the original graph. This method allows for an effective blend of rotations, traveling to already unbounded vertices, and making bounded vertices unbounded while maintaining an optimal time complexity.

An optimization utilized in the heuristic is preemptively checking for the presence of a cycle at every arrangement. If there exists $v_i$ for $1 \le i \le k-2$ such that edges $v_{i}v_{k}$ and $v_{i+1}v_{0}$ exist, then a cycle can be formed. Alternatively, this can be configured as a rotation followed by a normal cycle extension with adjacent end vertices. Since vertices cannot be revisited in the 0-1 BFS, a rotation that induces a normal cycle extension may not follow through, whereas a preemptive cycle check will be able to detect the cycle. An implementation can be found in Appendix \ref{appendix:cycleCheck}, in which a preemptive cycle check is utilized in every state of the BFS.

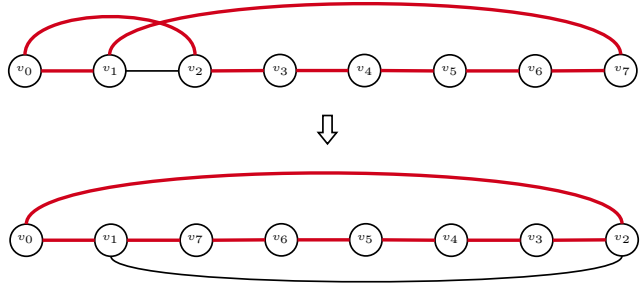
\begin{figure}[htbp] \label{figRoute2}
    \centering
        \tikzset{every picture/.style={line width=0.75pt}} 
\resizebox{8.5cm}{!}{
        \begin{tikzpicture}[x=0.75pt,y=0.75pt,yscale=-1,xscale=1]
\draw   (220.13,70.15) .. controls (220.13,64.88) and (224.4,60.61) .. (229.67,60.61) .. controls (234.94,60.61) and (239.21,64.88) .. (239.21,70.15) .. controls (239.21,75.42) and (234.94,79.69) .. (229.67,79.69) .. controls (224.4,79.69) and (220.13,75.42) .. (220.13,70.15) -- cycle ;
\draw   (269.91,70.15) .. controls (269.91,64.88) and (274.18,60.61) .. (279.45,60.61) .. controls (284.72,60.61) and (288.99,64.88) .. (288.99,70.15) .. controls (288.99,75.42) and (284.72,79.69) .. (279.45,79.69) .. controls (274.18,79.69) and (269.91,75.42) .. (269.91,70.15) -- cycle ;
\draw   (320.13,70.15) .. controls (320.13,64.88) and (324.4,60.61) .. (329.67,60.61) .. controls (334.94,60.61) and (339.21,64.88) .. (339.21,70.15) .. controls (339.21,75.42) and (334.94,79.69) .. (329.67,79.69) .. controls (324.4,79.69) and (320.13,75.42) .. (320.13,70.15) -- cycle ;
\draw   (370.13,69.93) .. controls (370.13,64.66) and (374.4,60.39) .. (379.67,60.39) .. controls (384.94,60.39) and (389.21,64.66) .. (389.21,69.93) .. controls (389.21,75.2) and (384.94,79.47) .. (379.67,79.47) .. controls (374.4,79.47) and (370.13,75.2) .. (370.13,69.93) -- cycle ;
\draw   (419.82,69.93) .. controls (419.82,64.66) and (424.09,60.39) .. (429.36,60.39) .. controls (434.63,60.39) and (438.91,64.66) .. (438.91,69.93) .. controls (438.91,75.2) and (434.63,79.47) .. (429.36,79.47) .. controls (424.09,79.47) and (419.82,75.2) .. (419.82,69.93) -- cycle ;
\draw   (470.04,70.15) .. controls (470.04,64.88) and (474.32,60.61) .. (479.59,60.61) .. controls (484.86,60.61) and (489.13,64.88) .. (489.13,70.15) .. controls (489.13,75.42) and (484.86,79.69) .. (479.59,79.69) .. controls (474.32,79.69) and (470.04,75.42) .. (470.04,70.15) -- cycle ;
\draw   (520.27,70.15) .. controls (520.27,64.88) and (524.54,60.61) .. (529.81,60.61) .. controls (535.08,60.61) and (539.35,64.88) .. (539.35,70.15) .. controls (539.35,75.42) and (535.08,79.69) .. (529.81,79.69) .. controls (524.54,79.69) and (520.27,75.42) .. (520.27,70.15) -- cycle ;
\draw   (570.49,69.93) .. controls (570.49,64.66) and (574.76,60.39) .. (580.03,60.39) .. controls (585.3,60.39) and (589.57,64.66) .. (589.57,69.93) .. controls (589.57,75.2) and (585.3,79.47) .. (580.03,79.47) .. controls (574.76,79.47) and (570.49,75.2) .. (570.49,69.93) -- cycle ;
\draw [color={rgb, 255:red, 208; green, 2; blue, 27 }  ,draw opacity=1 ][line width=1.5]    (239.21,70.15) -- (269.91,70.15) ;
\draw    (288.99,70.15) -- (320.13,70.15) ;
\draw [color={rgb, 255:red, 208; green, 2; blue, 27 }  ,draw opacity=1 ][line width=1.5]    (339.21,70.15) -- (370.13,69.93) ;
\draw [color={rgb, 255:red, 208; green, 2; blue, 27 }  ,draw opacity=1 ][line width=1.5]    (389.21,69.93) -- (419.82,69.93) ;
\draw [color={rgb, 255:red, 208; green, 2; blue, 27 }  ,draw opacity=1 ][line width=1.5]    (438.91,69.93) -- (470.04,70.15) ;
\draw [color={rgb, 255:red, 208; green, 2; blue, 27 }  ,draw opacity=1 ][line width=1.5]    (489.13,70.15) -- (520.27,70.15) ;
\draw [color={rgb, 255:red, 208; green, 2; blue, 27 }  ,draw opacity=1 ][line width=1.5]    (539.35,70.15) -- (570.49,69.93) ;
\draw [color={rgb, 255:red, 208; green, 2; blue, 27 }  ,draw opacity=1 ][line width=1.5]    (229.67,60.61) .. controls (229.8,30.6) and (329.93,30.11) .. (329.67,60.61) ;
\draw [color={rgb, 255:red, 208; green, 2; blue, 27 }  ,draw opacity=1 ][line width=1.5]    (279.45,60.61) .. controls (279.87,20.09) and (579.87,20.27) .. (580.03,60.39) ;
\draw   (220.93,170.55) .. controls (220.93,165.28) and (225.2,161.01) .. (230.47,161.01) .. controls (235.74,161.01) and (240.01,165.28) .. (240.01,170.55) .. controls (240.01,175.82) and (235.74,180.09) .. (230.47,180.09) .. controls (225.2,180.09) and (220.93,175.82) .. (220.93,170.55) -- cycle ;
\draw   (270.71,170.55) .. controls (270.71,165.28) and (274.98,161.01) .. (280.25,161.01) .. controls (285.52,161.01) and (289.79,165.28) .. (289.79,170.55) .. controls (289.79,175.82) and (285.52,180.09) .. (280.25,180.09) .. controls (274.98,180.09) and (270.71,175.82) .. (270.71,170.55) -- cycle ;
\draw   (320.93,170.55) .. controls (320.93,165.28) and (325.2,161.01) .. (330.47,161.01) .. controls (335.74,161.01) and (340.01,165.28) .. (340.01,170.55) .. controls (340.01,175.82) and (335.74,180.09) .. (330.47,180.09) .. controls (325.2,180.09) and (320.93,175.82) .. (320.93,170.55) -- cycle ;
\draw   (370.93,170.33) .. controls (370.93,165.06) and (375.2,160.79) .. (380.47,160.79) .. controls (385.74,160.79) and (390.01,165.06) .. (390.01,170.33) .. controls (390.01,175.6) and (385.74,179.87) .. (380.47,179.87) .. controls (375.2,179.87) and (370.93,175.6) .. (370.93,170.33) -- cycle ;
\draw   (420.62,170.33) .. controls (420.62,165.06) and (424.89,160.79) .. (430.16,160.79) .. controls (435.43,160.79) and (439.71,165.06) .. (439.71,170.33) .. controls (439.71,175.6) and (435.43,179.87) .. (430.16,179.87) .. controls (424.89,179.87) and (420.62,175.6) .. (420.62,170.33) -- cycle ;
\draw   (470.84,170.55) .. controls (470.84,165.28) and (475.12,161.01) .. (480.39,161.01) .. controls (485.66,161.01) and (489.93,165.28) .. (489.93,170.55) .. controls (489.93,175.82) and (485.66,180.09) .. (480.39,180.09) .. controls (475.12,180.09) and (470.84,175.82) .. (470.84,170.55) -- cycle ;
\draw   (521.07,170.55) .. controls (521.07,165.28) and (525.34,161.01) .. (530.61,161.01) .. controls (535.88,161.01) and (540.15,165.28) .. (540.15,170.55) .. controls (540.15,175.82) and (535.88,180.09) .. (530.61,180.09) .. controls (525.34,180.09) and (521.07,175.82) .. (521.07,170.55) -- cycle ;
\draw   (571.29,170.33) .. controls (571.29,165.06) and (575.56,160.79) .. (580.83,160.79) .. controls (586.1,160.79) and (590.37,165.06) .. (590.37,170.33) .. controls (590.37,175.6) and (586.1,179.87) .. (580.83,179.87) .. controls (575.56,179.87) and (571.29,175.6) .. (571.29,170.33) -- cycle ;
\draw [color={rgb, 255:red, 208; green, 2; blue, 27 }  ,draw opacity=1 ][line width=1.5]    (240.01,170.55) -- (270.71,170.55) ;
\draw [color={rgb, 255:red, 208; green, 2; blue, 27 }  ,draw opacity=1 ][line width=1.5]    (289.79,170.55) -- (320.93,170.55) ;
\draw [color={rgb, 255:red, 208; green, 2; blue, 27 }  ,draw opacity=1 ][line width=1.5]    (340.01,170.55) -- (370.93,170.33) ;
\draw [color={rgb, 255:red, 208; green, 2; blue, 27 }  ,draw opacity=1 ][line width=1.5]    (390.01,170.33) -- (420.62,170.33) ;
\draw [color={rgb, 255:red, 208; green, 2; blue, 27 }  ,draw opacity=1 ][line width=1.5]    (439.71,170.33) -- (470.84,170.55) ;
\draw [color={rgb, 255:red, 208; green, 2; blue, 27 }  ,draw opacity=1 ][line width=1.5]    (489.93,170.55) -- (521.07,170.55) ;
\draw [color={rgb, 255:red, 208; green, 2; blue, 27 }  ,draw opacity=1 ][line width=1.5]    (540.15,170.55) -- (571.29,170.33) ;
\draw [color={rgb, 255:red, 0; green, 0; blue, 0 }  ,draw opacity=1 ][line width=0.75]    (280.25,180.09) .. controls (280.14,200.14) and (580.14,200.43) .. (580.83,179.87) ;
\draw [color={rgb, 255:red, 208; green, 2; blue, 27 }  ,draw opacity=1 ][line width=1.5]    (230.47,161.01) .. controls (230.89,120.49) and (580.67,120.67) .. (580.83,160.79) ;
\draw   (402.2,106.68) -- (404.9,106.68) -- (404.9,96.6) -- (410.3,96.6) -- (410.3,106.68) -- (413,106.68) -- (407.6,113.4) -- cycle ;

\draw (223,66) node [anchor=north west][inner sep=0.75pt]  [font=\tiny] [align=left] {$\displaystyle v_{0}$};
\draw (273.5,66) node [anchor=north west][inner sep=0.75pt]  [font=\tiny] [align=left] {$\displaystyle v_{1}$};
\draw (323.5,66) node [anchor=north west][inner sep=0.75pt]  [font=\tiny] [align=left] {$\displaystyle v_{2}$};
\draw (373.5,66) node [anchor=north west][inner sep=0.75pt]  [font=\tiny] [align=left] {$\displaystyle v_{3}$};
\draw (423,66) node [anchor=north west][inner sep=0.75pt]  [font=\tiny] [align=left] {$\displaystyle v_{4}$};
\draw (473,66) node [anchor=north west][inner sep=0.75pt]  [font=\tiny] [align=left] {$\displaystyle v_{5}$};
\draw (523.5,66) node [anchor=north west][inner sep=0.75pt]  [font=\tiny] [align=left] {$\displaystyle v_{6}$};
\draw (574.5,66) node [anchor=north west][inner sep=0.75pt]  [font=\tiny] [align=left] {$\displaystyle v_{7}$};

\draw (223.5,166) node [anchor=north west][inner sep=0.75pt]  [font=\tiny] [align=left] {$\displaystyle v_{0}$};
\draw (273.5,166) node [anchor=north west][inner sep=0.75pt]  [font=\tiny] [align=left] {$\displaystyle v_{1}$};
\draw (323.5,166) node [anchor=north west][inner sep=0.75pt]  [font=\tiny] [align=left] {$\displaystyle v_{7}$};
\draw (373.5,166) node [anchor=north west][inner sep=0.75pt]  [font=\tiny] [align=left] {$\displaystyle v_{6}$};
\draw (423.5,166) node [anchor=north west][inner sep=0.75pt]  [font=\tiny] [align=left] {$\displaystyle v_{5}$};
\draw (473.5,166) node [anchor=north west][inner sep=0.75pt]  [font=\tiny] [align=left] {$\displaystyle v_{4}$};
\draw (523.5,166) node [anchor=north west][inner sep=0.75pt]  [font=\tiny] [align=left] {$\displaystyle v_{3}$};
\draw (574.5,166) node [anchor=north west][inner sep=0.75pt]  [font=\tiny] [align=left] {$\displaystyle v_{2}$};

        \end{tikzpicture}
}
    \caption{Preemptive cycle extension}
\end{figure}

The only possible edge weights are $0$ and $1$, so a 0-1 BFS can be run instead of Dijkstra's. The distance array stores both the distance and path to a vertex which are initialized to $0$ and $P$ for $v_0$ and $v_k$. Due to the dynamic nature of rotations, a visited array is included so that vertices are not revisited. If the preemptive cycle check finds a cycle, the path is normalized so that the end vertices are neighbours and an adjacent vertex to the path with minimal degree is returned. As soon as a vertex with positive degree or a cycle is found, the BFS can immediately exit because vertices are visited in increasing order of distance. An implementation can be found in Appendix \ref{appendix:rerouting}. 

\subsubsection{k-Unbounded Hamiltonian Path to Cycle}
If the end vertices of the path are not adjacent, it is necessary to convert the path into a k-Unbounded Hamiltonian Cycle. One approach is to adjust the rerouting process by fixing the start vertex and adjust only the last vertex. 

Instead, we introduce an idea that utilizes 0-1 BFS but with ``nodes" that are represented by a 2-d distance matrix $D[ ][ ]$, where $D[i][j]$ is the smallest number of new unbounded vertices needed for $i$ and $j$ to be endpoints. Instead of fixing an end of the path, both ends may be changed. This method allows for a considerably larger number of arrangements to be considered, as there exist $N^2$ instead of $N$ states in the BFS tree. Thus, the chance of a k-Unbounded Hamiltonian Cycle being encountered with a lower distance is increased. The preemptive cycle check from the rerouting algorithm is maintained and it is able to determine the existence of a k-Unbounded Hamiltonian Cycle. An implementation can be found in Appendix \ref{appendix:pathToCycle}.

\subsection{Time Complexity and Memory}
Marking cut vertices unbounded requires a DFS that could traverse the entire graph, and this DFS is called $N$ times, resulting in a time complexity of $O(NE)$. It is worth mentioning that there exists an $O(E)$ implementation of locating all cut vertices in a graph involving the creation of a DFS Tree and keeping track of when vertices are visited. However, in the context of the Unbounded Heuristic, both implementations have a negligibly small run time. The greedy vertex selection traverses the graph with a time complexity of $O(E)$. Reroute can visit up to $N$ vertices where every visit iterates over a path. Since Reroute can be called up to $N$ times, the time complexity is $O(N^2M)$, where M is the final length of the path. Converting from a k-Unbounded Hamiltonian Path to Cycle utilizes a BFS that can visit $N^2$ nodes where every visit iterates over a path, leading to a time complexity of $O(N^2M)$.

Theoretically, $M$ is bounded by $N^2$ since rerouting can add up to $N$ vertices to the path. Thus, the worst case time complexity is $O(NE+E+N^2M+N^2M)=O(N^4)$. However, in practice, it is nearly impossible for $M$ to be this large because the number of dead ends encountered is fairly small and the rerouting process is able to escape a dead end by visiting a small number of vertices; all trials had a path length lower than $7N$ in extremely sparse graphs and $2N$ in every other graph. This leads to an effective time complexity of $O(N^3)$.

The bulk of the memory will be allocated to storing paths in the 2-D distance matrix when converting to a k-Unbounded Hamiltonian Cycle. The heuristic can visit up to $N^2$ nodes, resulting in a memory of $O(N^2M)$, or a worst case memory of $O(N^4)$ and effective memory of $O(N^3)$. If the conversion to a k-Unbounded Hamiltonian Cycle is instead processed by calling reroute, the effective memory is $O(N^2)$.

\subsection{Experiments}
We consider two versions of the heuristic: the Unbounded Heuristic, which uses the implementation in Appendix \ref{appendix:pathToCycle} to convert to a k-Unbounded Hamiltonian Cycle, and the Fast Unbounded Heuristic, which uses the implementation in Appendix \ref{appendix:rerouting} to convert to a k-Unbounded Hamiltonian Cycle. All trials were conducted in C++ using an Intel Pentium Silver J5005 processor in an environment with 2GB RAM. It is worth mentioning that this CPU has comparatively low processing capabilities to the CPUs used to measure performance on the HCP heuristics, which may result in inflated runtimes.  

\subsubsection{Random Graphs} \label{section:randomGraphs}
The only method that is known to invariably find an m-Unbounded Hamiltonian Cycle is a brute-force algorithm. However, these approaches cannot feasibly be run on large graphs, resulting in the trials being conducted on graphs with 20 vertices. In Table \ref{table:N20}, 9 different average degrees are tested and each number represents the average over 1000 randomly generated, connected graphs. Average degree is tested instead of a uniform distribution to create more difficult instances and to allow for non-integer degrees. For each degree, the Unbounded and Fast Unbounded Heuristics are run over the same set of graphs and the Unbounded DP Algorithm is used to calculate the number of unbounded vertices needed in each graph.
\begin{table}[htbp]
\centering
\resizebox{11cm}{!}{
\begin{tabular}{|c|c|c|c|c|c|}
\hline
\begin{tabular}[c]{@{}c@{}}Average\\ Degree\end{tabular} & \begin{tabular}[c]{@{}c@{}}Unbounded \\ Heuristic\\ Vertices\end{tabular} & \begin{tabular}[c]{@{}c@{}}Fast\\ Unbounded\\ Heuristic\\ Vertices\end{tabular} & \begin{tabular}[c]{@{}c@{}}Unbounded \\ Vertices\\ Needed\end{tabular} & \begin{tabular}[c]{@{}c@{}}Unbounded \\ Heuristic\\ Difference\end{tabular} & \begin{tabular}[c]{@{}c@{}}Fast \\ Unbounded\\ Heuristic\\ Difference\end{tabular} \\ \hline
2                                                        & 10.458                                                                    & 10.458                                                                          & 10.403                                                                 & 0.055                                                                       & 0.055                                                                              \\
2.5                                                      & 5.811                                                                     & 5.955                                                                           & 5.743                                                                  & 0.068                                                                       & 0.212                                                                              \\
3                                                        & 3.409                                                                     & 3.601                                                                           & 3.338                                                                  & 0.071                                                                       & 0.263                                                                              \\
3.5                                                      & 1.979                                                                     & 2.136                                                                           & 1.942                                                                  & 0.037                                                                       & 0.194                                                                              \\
4                                                        & 1.072                                                                     & 1.200                                                                           & 1.057                                                                  & 0.015                                                                       & 0.143                                                                              \\
4.5                                                      & 0.526                                                                     & 0.587                                                                           & 0.524                                                                  & 0.002                                                                       & 0.063                                                                              \\
5                                                        & 0.290                                                                     & 0.330                                                                           & 0.290                                                                  & 0.000                                                                       & 0.040                                                                              \\
5.5                                                      & 0.152                                                                     & 0.164                                                                           & 0.152                                                                  & 0.000                                                                       & 0.012                                                                              \\
6                                                        & 0.081                                                                     & 0.088                                                                           & 0.081                                                                  & 0.000                                                                       & 0.007                                                                              \\ \hline
\end{tabular}
}
\caption{Accuracy comparison of Unbounded Heuristic and Fast Unbounded Heuristic ($N=20$)}
\label{table:N20}
\end{table}

As the average degree increases, $m$ decreases exponentially. Furthermore, the Unbounded Heuristic performs increasingly well compared to the Fast Unbounded Heuristic as the average degree increases. Both heuristics perform the worst around an average degree of 3, with the Unbounded Heuristic using $\approx 0.36\%$ extra unbounded vertices and the fast Unbounded Heuristic using $\approx 1.32\%$ extra unbounded vertices. However, accuracy improves in average degrees greater than 3, and the Unbounded Heuristic could correctly identify m-Unbounded Cycles in every graph with average degree 5 or larger. 

Additional tests were conducted at $N=500$ to gain a better understanding of performance on larger graphs. Both heuristics are tested on 10 different average degrees with the numbers representing the average over 5000 randomly generated, connected graphs for each degree. Since a brute-force algorithm is not able to run on such large graphs, the number of cut vertices is included to serve as a rough estimate for the accuracy of the heuristic. However, as previously discussed in Section \ref{section:cutVertices}, this number is not necessarily a tight lower bound for the number of unbounded vertices and is not nearly as representative of the accuracy of the heuristics as Table \ref{table:N20}. The heuristics were tested on different sets of graphs, which is the reason for discrepancies such as the number of unbounded vertices being larger in the Unbounded Heuristic compared to the Fast Unbounded Heuristic in some instances.

\begin{table}[htbp]
\centering
\resizebox{11cm}{!}{
\begin{tabular}{|c|c|c|c|c|c|}
\hline
Average Degree & Unbounded Vertices & Cut Vertices & Path Length & Time (ms) \\ \hline
2              & 310.825            & 305.636      & 1530        & 1807      \\
2.5            & 146.71             & 130.88       & 752         & 728       \\
3              & 86.46              & 73.85        & 643         & 534       \\
3.5            & 50.11              & 42.98        & 579         & 366       \\
4              & 28.67              & 25.68        & 543         & 230       \\
4.5            & 16.40              & 15.33        & 522         & 178       \\
5              & 9.73               & 9.37         & 512         & 127       \\
5.5            & 5.65               & 5.54         & 507         & 105       \\
6              & 3.42               & 3.38         & 504         & 103       \\
6.5            & 2.04               & 2.03         & 502         & 107       \\ \hline
\end{tabular}
}
\caption{Unbounded Heuristic performance ($N=500$)}
\end{table}

\begin{table}[htbp]
\centering
\resizebox{11cm}{!}{
\begin{tabular}{|c|c|c|c|c|c|}
\hline
Average Degree & Unbounded Vertices & Cut Vertices & Path Length & Time (ms) \\ \hline
2              & 310.54             & 305.40       & 989         & 291       \\
2.5            & 147.36             & 130.83       & 727         & 180       \\
3              & 87.11              & 73.94        & 637         & 132       \\
3.5            & 50.72              & 42.99        & 579         & 112       \\
4              & 29.09              & 25.71        & 543         & 87       \\
4.5            & 16.67              & 15.33        & 523         & 79       \\
5              & 9.74               & 9.26         & 513         & 66       \\
5.5            & 5.75               & 5.58         & 507         & 64       \\
6              & 3.39               & 3.33         & 504         & 59       \\
6.5            & 2.05               & 2.02         & 502         & 61       \\ \hline
\end{tabular}
}
\caption{Fast Unbounded Heuristic performance ($N=500$)}
\end{table}

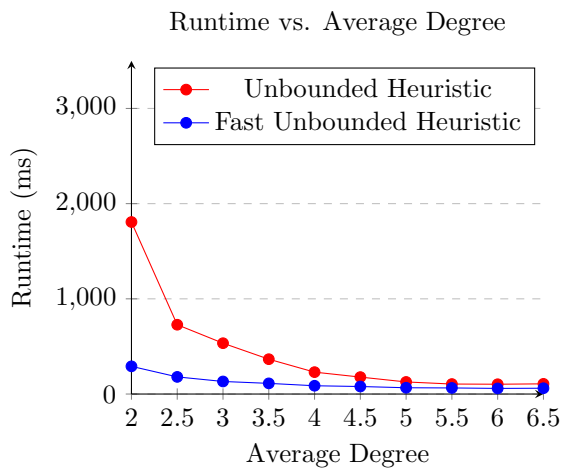
\begin{figure}[htbp]
\centering
\begin{tikzpicture}
\begin{axis}[
title=Runtime vs. Average Degree,
xlabel={Average Degree}, 
ylabel={Runtime (ms)}, 
xmin=2, xmax=6.5, 
ymin=0, ymax=3500, 
xtick={2, 2.5, 3, 3.5, 4, 4.5, 5, 5.5, 6, 6.5},
axis lines=left,
ymajorgrids=true,
grid style=dashed
]
\addplot[
    color=red,
    mark=*,
    ]
    coordinates {
    (2, 1807)(2.5, 728)(3, 534)(3.5, 366)(4, 230)(4.5, 178)(5, 127)(5.5, 105)(6, 103)(6.5,107)
    };

\addplot[
    color=blue,
    mark=*,
    ]
    coordinates {
    (2, 291)(2.5, 180)(3, 132)(3.5, 112)(4, 87)(4.5, 79)(5, 66)(5.5, 64)(6, 59)(6.5,61)
    };

\legend{Unbounded Heuristic, Fast Unbounded Heuristic}

\end{axis}
\end{tikzpicture}
\caption{Graph comparing runtimes ($N=500$)}
\end{figure}

\newpage
Both the runtimes and the performance gap between the two heuristics decrease as the average degree increases. As the number of edges increases, the 2-dimensional 0-1 BFS of the Unbounded Heuristic is able to find the m-Unbounded Hamiltonian Cycle with a smaller distance, and thus, considers a smaller number of states. Unsurprisingly, the number of unbounded vertices increases by a factor of around 25 as the graph size is increased from 20 to 500, showing how performance is independent of size and is instead influenced by the average degree.
\subsubsection{TSPLIB}
The Unbounded Heuristic can also function as a Hamiltonian Cycle heuristic, as a 0-Unbounded Hamiltonian Cycle is synonymous with a Hamiltonian Cycle. The TSPLIB dataset \cite{TSPLIB} contains $9$ instances of graphs with Hamiltonian Cycles. These graphs have served as a benchmark for a number of state-of-the-art HCP heuristics and are noteworthy for their large sizes of up to 5000 vertices and sparsity. Both Unbounded Heuristics were able to solve all instances and the runtimes are compared to the Concorde TSP Solver \cite{Concorde}, the Snakes and Ladders Heuristic \cite{SnakesLadders}, and HybridHAM \cite{HybridHAM}.
\begin{table}[htbp]
\centering
\resizebox{11cm}{!}{
\begin{tabular}{|c|c|ccccc|}
\hline
\multirow{2}{*}{Name} & \multirow{2}{*}{No. of Vertices} & \multicolumn{5}{c|}{Running Time (sec)}                                                                                                                                                                                                                                                                                    \\ \cline{3-7} 
                      &                                  & \multicolumn{1}{c|}{Concorde} & \multicolumn{1}{c|}{\begin{tabular}[c]{@{}c@{}}Snakes and \\ Ladders Heuristic\end{tabular}} & \multicolumn{1}{c|}{Hybrid Ham} & \multicolumn{1}{c|}{\begin{tabular}[c]{@{}c@{}}Unbounded \\ Heuristic\end{tabular}} & \begin{tabular}[c]{@{}c@{}}Fast Unbounded\\  Heuristic\end{tabular} \\ \hline
alb1000               & 1000                             & \multicolumn{1}{c|}{4.95}     & \multicolumn{1}{c|}{0.1}                                                                     & \multicolumn{1}{c|}{0.2656}     & \multicolumn{1}{c|}{0.223}                                                          & 0.056                                                               \\
alb2000               & 2000                             & \multicolumn{1}{c|}{7.30}     & \multicolumn{1}{c|}{0.8}                                                                     & \multicolumn{1}{c|}{1.4375}     & \multicolumn{1}{c|}{0.961}                                                          & 0.806                                                               \\
alb3000a              & 3000                             & \multicolumn{1}{c|}{9.56}     & \multicolumn{1}{c|}{3.44}                                                                    & \multicolumn{1}{c|}{2.7656}     & \multicolumn{1}{c|}{8.131}                                                          & 2.713                                                               \\
alb3000b              & 3000                             & \multicolumn{1}{c|}{9.94}     & \multicolumn{1}{c|}{3.64}                                                                    & \multicolumn{1}{c|}{1.5781}     & \multicolumn{1}{c|}{3.660}                                                          & 3.184                                                               \\
alb3000c              & 3000                             & \multicolumn{1}{c|}{9.95}     & \multicolumn{1}{c|}{4.31}                                                                    & \multicolumn{1}{c|}{1.8438}     & \multicolumn{1}{c|}{5.787}                                                          & 1.934                                                               \\
alb3000d              & 3000                             & \multicolumn{1}{c|}{10.14}    & \multicolumn{1}{c|}{4.03}                                                                    & \multicolumn{1}{c|}{1.6406}     & \multicolumn{1}{c|}{3.544}                                                          & 2.460                                                               \\
alb3000e              & 3000                             & \multicolumn{1}{c|}{10.44}    & \multicolumn{1}{c|}{4.29}                                                                    & \multicolumn{1}{c|}{1.6719}     & \multicolumn{1}{c|}{3.676}                                                          & 2.615                                                               \\
alb4000               & 4000                             & \multicolumn{1}{c|}{13.45}    & \multicolumn{1}{c|}{13.89}                                                                   & \multicolumn{1}{c|}{3.0625}     & \multicolumn{1}{c|}{7.017}                                                          & 3.771                                                               \\
alb5000               & 5000                             & \multicolumn{1}{c|}{17.24}    & \multicolumn{1}{c|}{14.12}                                                                   & \multicolumn{1}{c|}{8.9844}     & \multicolumn{1}{c|}{13.571}                                                         & 7.092                                                               \\ \hline
\end{tabular}
}
\caption{Comparison of runtimes on TSPLIB graphs}
\end{table}

The Fast Unbounded Heuristic performs quicker than the Concorde and the Snakes and Ladders Heuristics and is comparable to HybridHAM. The Unbounded Heuristic takes several factors longer than the Fast Unbounded Heuristic but is still able to solve all instances relatively quickly. Interestingly, despite the Unbounded Heuristic having memory requirements of up to $O(N^3)$, the Unbounded Heuristic is able to run on graphs with over 1000 nodes. This highlights the dynamic nature of the memory, and the quicker the conversion from a k-Unbounded Hamiltonian Path to Cycle takes, the less memory used.

\subsubsection{FHCP Challenge Set}
Another dataset of graphs containing Hamiltonian Cycles is the FHCP Challenge Set \cite{FHCP}. The FHCP graphs are unique in the sense that they are specifically designed to resist existing heuristics through a complex underlying structure. During the one year period that the challenge was running, 16 of the 1001 instances remained unsolved and only two teams were able to solve over half of the instances. In addition, the challenge was not limited to the performance of a single program and the top performing teams utilized multiple algorithms that took advantage of properties of the graph over many months.

Out of the first 250 instances, the Unbounded Heuristic found 88 Hamiltonian Cycles, the Fast Unbounded Heuristic found 11 Hamiltonian Cycles, and each heuristic found the same set of 131 Hamiltonian Paths. In comparison, HybridHAM found 13 Hamiltonian Cycles and 75 Hamiltonian Paths in the first 250 instances. The conversion from a Hamiltonian Path to a Hamiltonian Cycle implemented in Appendix \ref{appendix:pathToCycle} worked exceptionally well and was able to convert 88 of the 131 Hamiltonian Paths to Hamiltonian Cycles. The Fast Unbounded Heuristic was able to run up to hundreds of times faster than both the Unbounded Heuristic and HybridHAM while marking less than 1$\%$ of the vertices unbounded in the vast majority of graphs. As Table \ref{tableFHCP} demonstrates, the Unbounded Heuristic was able to reduce the number of unbounded vertices by a notable amount compared to the Fast Unbounded Heuristic throughout many instances.

Due to the recency of the FHCP Challenge Set, many other established heuristics either have not been tested on this set or do not have tests publicly available. The lack of published data and the similarities of HybridHAM as a cubic time complexity heuristic that utilizes rotations are the primary reasons HybridHAM is used as the main point of reference. 

\begin{table}[htbp]
\begin{minipage}{.5\linewidth}
\resizebox{6.5cm}{!}{
\begin{tabular}{|c|c|c|c|c|c|}
\hline
Graph \# & \begin{tabular}[c]{@{}c@{}}No. of \\ Vertices\end{tabular} & \begin{tabular}[c]{@{}c@{}}Unbounded \\ Heuristic\\ Vertices\end{tabular} & \begin{tabular}[c]{@{}c@{}}Unbounded\\  Heuristic\\  Time (ms)\end{tabular} & \begin{tabular}[c]{@{}c@{}}Fast \\ Unbounded\\ Heuristic\\  Vertices\end{tabular} & \begin{tabular}[c]{@{}c@{}}Fast\\  Unbounded\\ Heuristic\\  Time (ms)\end{tabular} \\ \hline
1        & 66                                                         & 1                                                                         & 25                                                                          & 1                                                                           & 0                                                                               \\
2        & 70                                                         & 0                                                                         & 1                                                                           & 1                                                                           & 0                                                                               \\
3        & 78                                                         & 0                                                                         & 19                                                                          & 1                                                                           & 0                                                                               \\
4        & 84                                                         & 0                                                                         & 29                                                                          & 2                                                                           & 1                                                                               \\
5        & 90                                                         & 1                                                                         & 46                                                                          & 1                                                                           & 1                                                                               \\
6        & 94                                                         & 0                                                                         & 7                                                                           & 1                                                                           & 1                                                                               \\
7        & 102                                                        & 3                                                                         & 64                                                                          & 3                                                                           & 1                                                                               \\
8        & 108                                                        & 0                                                                         & 96                                                                          & 2                                                                           & 1                                                                               \\
9        & 114                                                        & 1                                                                         & 96                                                                          & 1                                                                           & 0                                                                               \\
10       & 118                                                        & 0                                                                         & 51                                                                          & 1                                                                           & 1                                                                               \\
11       & 126                                                        & 1                                                                         & 109                                                                         & 1                                                                           & 1                                                                               \\
12       & 132                                                        & 0                                                                         & 93                                                                          & 2                                                                           & 2                                                                               \\
13       & 138                                                        & 1                                                                         & 140                                                                         & 1                                                                           & 0                                                                               \\
14       & 138                                                        & 1                                                                         & 135                                                                         & 1                                                                           & 1                                                                               \\
15       & 150                                                        & 1                                                                         & 172                                                                         & 1                                                                           & 0                                                                               \\
16       & 156                                                        & 0                                                                         & 173                                                                         & 2                                                                           & 2                                                                               \\
17       & 162                                                        & 1                                                                         & 235                                                                         & 1                                                                           & 0                                                                               \\
18       & 166                                                        & 1                                                                         & 3                                                                           & 3                                                                           & 1                                                                               \\
19       & 170                                                        & 3                                                                         & 7                                                                           & 4                                                                           & 6                                                                               \\
20       & 174                                                        & 1                                                                         & 272                                                                         & 1                                                                           & 1                                                                               \\
21       & 180                                                        & 0                                                                         & 217                                                                         & 2                                                                           & 2                                                                               \\
22       & 186                                                        & 1                                                                         & 358                                                                         & 1                                                                           & 1                                                                               \\
23       & 190                                                        & 0                                                                         & 116                                                                         & 1                                                                           & 2                                                                               \\
24       & 198                                                        & 1                                                                         & 92                                                                          & 3                                                                           & 2                                                                               \\
25       & 204                                                        & 1                                                                         & 627                                                                         & 2                                                                           & 1                                                                               \\
26       & 210                                                        & 1                                                                         & 695                                                                         & 1                                                                           & 2                                                                               \\
27       & 214                                                        & 1                                                                         & 3                                                                           & 1                                                                           & 1                                                                               \\
28       & 222                                                        & 1                                                                         & 577                                                                         & 1                                                                           & 1                                                                               \\
29       & 228                                                        & 0                                                                         & 271                                                                         & 2                                                                           & 2                                                                               \\
30       & 234                                                        & 1                                                                         & 702                                                                         & 1                                                                           & 2                                                                               \\
31       & 238                                                        & 1                                                                         & 5                                                                           & 1                                                                           & 1                                                                               \\
32       & 246                                                        & 1                                                                         & 839                                                                         & 1                                                                           & 1                                                                               \\
33       & 252                                                        & 0                                                                         & 661                                                                         & 2                                                                           & 3                                                                               \\
34       & 258                                                        & 1                                                                         & 3                                                                           & 1                                                                           & 1                                                                               \\
35       & 262                                                        & 0                                                                         & 46                                                                          & 1                                                                           & 2                                                                               \\
36       & 270                                                        & 1                                                                         & 940                                                                         & 1                                                                           & 1                                                                               \\
37       & 276                                                        & 0                                                                         & 951                                                                         & 2                                                                           & 2                                                                               \\
38       & 282                                                        & 3                                                                         & 1108                                                                        & 3                                                                           & 3                                                                               \\
39       & 286                                                        & 1                                                                         & 51                                                                          & 3                                                                           & 2                                                                               \\
40       & 294                                                        & 1                                                                         & 1456                                                                        & 1                                                                           & 1                                                                               \\
41       & 300                                                        & 0                                                                         & 1317                                                                        & 2                                                                           & 4                                                                               \\
42       & 306                                                        & 1                                                                         & 1438                                                                        & 1                                                                           & 2                                                                               \\
43       & 310                                                        & 0                                                                         & 72                                                                          & 1                                                                           & 2                                                                               \\
44       & 318                                                        & 1                                                                         & 6                                                                           & 1                                                                           & 2                                                                               \\
45       & 324                                                        & 6                                                                         & 99                                                                          & 7                                                                           & 16                                                                              \\
46       & 330                                                        & 1                                                                         & 1865                                                                        & 1                                                                           & 3                                                                               \\
47       & 334                                                        & 6                                                                         & 87                                                                          & 11                                                                          & 14                                                                              \\
48       & 338                                                        & 6                                                                         & 2053                                                                        & 7                                                                           & 22                                                                              \\
49       & 342                                                        & 1                                                                         & 581                                                                         & 3                                                                           & 4                                                                               \\
50       & 348                                                        & 0                                                                         & 2478                                                                        & 1                                                                           & 4                                                                               \\ \hline
\end{tabular}
}
\end{minipage}%
\begin{minipage}{.5\linewidth}
\resizebox{6.5cm}{!}{
\begin{tabular}{|c|c|c|c|c|c|}
\hline
Graph \# & \begin{tabular}[c]{@{}c@{}}No. of \\ Vertices\end{tabular} & \begin{tabular}[c]{@{}c@{}}Unbounded \\ Heuristic\\ Vertices\end{tabular} & \begin{tabular}[c]{@{}c@{}}Unbounded\\  Heuristic\\  Time (ms)\end{tabular} & \begin{tabular}[c]{@{}c@{}}Fast \\ Unbounded\\ Heuristic\\  Vertices\end{tabular} & \begin{tabular}[c]{@{}c@{}}Fast\\  Unbounded\\ Heuristic\\  Time (ms)\end{tabular} \\ \hline
51       & 354                                                        & 1                                                                         & 2945                                                                        & 1                                                                           & 4                                                                               \\
52       & 358                                                        & 0                                                                         & 595                                                                         & 1                                                                           & 3                                                                               \\
53       & 366                                                        & 1                                                                         & 2662                                                                        & 1                                                                           & 2                                                                               \\
54       & 372                                                        & 4                                                                         & 22                                                                          & 6                                                                           & 10                                                                              \\
55       & 378                                                        & 2                                                                         & 1800                                                                        & 3                                                                           & 6                                                                               \\
56       & 382                                                        & 12                                                                        & 1472                                                                        & 17                                                                          & 28                                                                              \\
57       & 390                                                        & 1                                                                         & 3533                                                                        & 1                                                                           & 3                                                                               \\
58       & 396                                                        & 0                                                                         & 3096                                                                        & 2                                                                           & 5                                                                               \\
59       & 400                                                        & 0                                                                         & 8                                                                           & 0                                                                           & 3                                                                               \\
60       & 402                                                        & 3                                                                         & 2409                                                                        & 3                                                                           & 7                                                                               \\
61       & 406                                                        & 2                                                                         & 3416                                                                        & 3                                                                           & 5                                                                               \\
62       & 408                                                        & 5                                                                         & 2234                                                                        & 7                                                                           & 44                                                                              \\
63       & 414                                                        & 2                                                                         & 2258                                                                        & 3                                                                           & 7                                                                               \\
64       & 416                                                        & 0                                                                         & 1466                                                                        & 3                                                                           & 13                                                                              \\
65       & 419                                                        & 0                                                                         & 4340                                                                        & 2                                                                           & 8                                                                               \\
66       & 426                                                        & 1                                                                         & 9                                                                           & 1                                                                           & 4                                                                               \\
67       & 430                                                        & 1                                                                         & 16                                                                          & 1                                                                           & 8                                                                               \\
68       & 438                                                        & 1                                                                         & 5183                                                                        & 1                                                                           & 3                                                                               \\
69       & 444                                                        & 0                                                                         & 4802                                                                        & 2                                                                           & 10                                                                              \\
70       & 450                                                        & 3                                                                         & 19                                                                          & 3                                                                           & 14                                                                              \\
71       & 454                                                        & 0                                                                         & 47                                                                          & 1                                                                           & 4                                                                               \\
72       & 460                                                        & 0                                                                         & 12                                                                          & 0                                                                           & 4                                                                               \\
73       & 462                                                        & 1                                                                         & 6422                                                                        & 1                                                                           & 7                                                                               \\
74       & 462                                                        & 12                                                                        & 138                                                                         & 13                                                                          & 34                                                                              \\
75       & 468                                                        & 0                                                                         & 4404                                                                        & 2                                                                           & 10                                                                              \\
76       & 471                                                        & 1                                                                         & 5895                                                                        & 2                                                                           & 19                                                                              \\
77       & 474                                                        & 20                                                                        & 3461                                                                        & 25                                                                          & 41                                                                              \\
78       & 478                                                        & 0                                                                         & 359                                                                         & 1                                                                           & 7                                                                               \\
79       & 480                                                        & 0                                                                         & 11                                                                          & 0                                                                           & 5                                                                               \\
80       & 486                                                        & 10                                                                        & 1553                                                                        & 21                                                                          & 54                                                                              \\
81       & 492                                                        & 7                                                                         & 840                                                                         & 10                                                                          & 24                                                                              \\
82       & 496                                                        & 2                                                                         & 5546                                                                        & 3                                                                           & 22                                                                              \\
83       & 498                                                        & 1                                                                         & 11                                                                          & 1                                                                           & 5                                                                               \\
84       & 500                                                        & 0                                                                         & 11                                                                          & 0                                                                           & 6                                                                               \\
85       & 502                                                        & 0                                                                         & 1647                                                                        & 1                                                                           & 5                                                                               \\
86       & 503                                                        & 1                                                                         & 6650                                                                        & 3                                                                           & 49                                                                              \\
87       & 507                                                        & 12                                                                        & 3054                                                                        & 15                                                                          & 76                                                                              \\
88       & 507                                                        & 1                                                                         & 251                                                                         & 2                                                                           & 62                                                                              \\
89       & 510                                                        & 1                                                                         & 18                                                                          & 1                                                                           & 9                                                                               \\
90       & 510                                                        & 0                                                                         & 14                                                                          & 0                                                                           & 5                                                                               \\
91       & 516                                                        & 0                                                                         & 6038                                                                        & 2                                                                           & 11                                                                              \\
92       & 522                                                        & 12                                                                        & 3168                                                                        & 21                                                                          & 77                                                                              \\
93       & 526                                                        & 2                                                                         & 5348                                                                        & 3                                                                           & 14                                                                              \\
94       & 534                                                        & 1                                                                         & 10972                                                                       & 1                                                                           & 12                                                                              \\
95       & 540                                                        & 0                                                                         & 8525                                                                        & 2                                                                           & 12                                                                              \\
96       & 540                                                        & 0                                                                         & 15                                                                          & 0                                                                           & 5                                                                               \\
97       & 546                                                        & 1                                                                         & 11777                                                                       & 1                                                                           & 10                                                                              \\
98       & 546                                                        & 23                                                                        & 927                                                                         & 25                                                                          & 90                                                                              \\
99       & 550                                                        & 1                                                                         & 736                                                                         & 2                                                                           & 7                                                                               \\
100      & 558                                                        & 1                                                                         & 13108                                                                       & 1                                                                           & 4                                                                               \\ \hline
\end{tabular}
}
\end{minipage}%
\caption{Performance of Unbounded and Fast Unbounded Heuristics on the first 100 instances of the FHCP Challenge Set}
\label{tableFHCP}
\end{table}

\section{Conclusions}
In this paper we explore a novel variant of the Hamiltonian Cycle Problem in which the objective is to find an m-Unbounded Hamiltonian Cycle. We first consider the problem on the only two non-Hamiltonian instances of the generalized Petersen graphs. We introduce a sufficient condition for the existence of a 1-Unbounded Hamiltonian Cycle and use it to show that $m=1$ for every hypohamiltonian graph, including the first set of instances. For the second set of instances, we prove that $m=2$ for $n=8$ and $2\le m\le 3$ for $n>8$. Additionally, we discuss linear-time algorithms to find an m-Unbounded Hamiltonian Cycle or Path in trees: finding an m-Unbounded Hamiltonian Path turned out to be quite a bit more complicated than a cycle. We then consider the problem in general graphs and show that even brute-force approaches are non-trivial. We propose the Unbounded DP Algorithm, a brute-force dynamic programming algorithm inspired by the Held-Karp Algorithm to find an m-Unbounded Hamiltonian Cycle with a worst case time complexity of $O(4^{N}N^2)$ and memory $O(2^{N}N)$. We also propose two constructions: the first is a construction in which the HCP may be applied to verify the existence of a k-Unbounded Hamiltonian Cycle given a set of unbounded vertices, and the second is a construction in which the ATSP may be applied to directly solve the k-Unbounded HCP. However, there are notable tradeoffs to these conversions, including an increase of the size of the graph by a factor of N as well as a brute force approach having to be adopted in the HCP conversion, both of which drastically increase runtime.

Finally, we propose two deterministic heuristics with an effective cubic time complexity: the Unbounded Heuristic and the Fast Unbounded Heuristic. The goal of the heuristics is to find an m-Unbounded Hamiltonian Cycle, or if this is not possible, get as close as possible to doing so in a significantly shorter time than a brute force approach. Table \ref{table:N20} shows that the Unbounded Heuristic is able to find m-Unbounded Hamiltonian Cycles with a very low margin of error in random graphs, with the Fast Unbounded Heuristic performing worse but still reasonably well. The Fast Unbounded Heuristic can perform up to six times faster on average than the Unbounded Heuristic in sparse graphs. The heuristics are also shown to be competitive Hamiltonian Cycle solvers. The Unbounded Heuristic could find Hamiltonian Cycles in all TSPLIB instances of up to 5000 nodes and 88 of the first 250 instances of the very challenging FHCP graphs. In comparison, HybridHAM, an HCP heuristic with similarities in both runtime and implementation, could only find Hamiltonian Cycles in 13 of the first 250 instances. The novel idea of converting a k-Unbounded Hamiltonian Path to Cycle using a 0-1 BFS with up to $N^2$ states performed surprisingly well and was able to convert over two-thirds of the Hamiltonian Paths to Hamiltonian Cycles in the FHCP graphs. Thus, it is very likely that if a preprocessing phase is implemented to first find a Hamiltonian Path, the Unbounded Heuristic would be able to solve many more instances of the FHCP challenge set. The Fast Unbounded Heuristic is comparable to HybridHAM in general graphs, but is able to consistently run hundreds of times faster in the harder instances of the FHCP set with a similar accuracy. 

\subsection{Future Work}
The concept of unbounded vertices in paths is unexplored, so there are vast possibilities for research on this topic. In the context of m-Unbounded Hamiltonian Cycles, we present several avenues for future research. 

In the second set of instances of the non-Hamiltonian generalized Petersen graphs, we establish that $2\le m\le 3$ when $n>8$. However, tests utilizing the algorithms discussed in this paper strongly suggest that $m=3$ for these instances. It remains to prove that there does not exist a 2-Unbounded Hamiltonian Cycle when $n>8$.

There is the possibility that there exists an algorithm with a better time complexity than the Unbounded DP Algorithm that does not begin with the costly ``brute force" approach of iterating over all subsets of vertices. There may also be a Monte Carlo algorithm similar to those utilized in the HCP that could be applied to the k-Unbounded HCP.

Although we propose a transformation to the ATSP, we have not run any trials involving a designated solver. It would be interesting to see a runtime comparison of an algorithm such as the Concorde TSP Solver on this construction compared to the heuristics proposed in the paper. We also suspect that there exists an m-Unbounded Hamiltonian Cycle such that the number of instances a vertex $v$ appears is bounded by a function of $d_v$. If such a relationship exists, the size of the TSP construction would be greatly reduced in sparse graphs.

No trials on both heuristics have found a path length of an m-Unbounded Hamiltonian Cycle greater than $7N$. It remains to prove some linear upper bound for the final path length of the heuristic. Experimental results have also shown that the difference between $m$ and the number of cut-vertices $c$ in a graph increases as the average degree decreases, and there may be a formal relationship between $m-c$ and the average degree.

The heuristics presented in the paper are deterministic in nature, and accuracy may further be improved by incorporating randomisation techniques or considering multiple start vertices at the cost of increased runtime.  

Lastly, further exploration of k-Unbounded Cycles is encouraged in special graphs outside of those discussed in the paper.

\section{Acknowledgements}
I would like to thank Jesse Stern for his guidance throughout the writing of the paper. I would also like to give a special acknowledgement to Michael Haythorpe, both for his providing of the FHCP Challenge set and his willingness to discuss the content of the paper and suggest further improvements.

\printbibliography

\newpage
\begin{appendices}

\section{Preemptive Cycle Check}
\label{appendix:cycleCheck}
\begin{algorithm}[!htbp] 
\DontPrintSemicolon
\SetKwInput{KwInput}{Input} 
\SetKwFunction{fIsCycle}{isCycle}
\SetKwProg{Fn}{void}{:}{}
\SetKwProg{Int}{int}{:}{}

\Int{\fIsCycle{v, isLast}}{
    s $\leftarrow$ v[0], e $\leftarrow$ v.back()\;
    
    \For{$i\gets1$ \KwTo v.size()-2}{
        \uIf{A[v[i]][e] and A[s][v[i+1]]}{
            reverse(v, i+1 to end)\;
            break\;
        } 
    }
    \BlankLine
    e $\leftarrow$ v.back()\;
    \lIf{!A[s][e]}{ return -1}
    \lIf{isLast}{return 1}
    \BlankLine
    \tcp{Find lowest degree neighbour and adjust path accordingly}
    nv=0, cur=0\;
    \For{$i\gets1$ \KwTo v.size()-1}{
        \For{k in adj[v[i]]}{
            \lIf{vis[k]}{continue}
            \uIf{deg[k] $<$ (nv is 0 ? INF : deg[nv])}{
                nv $\leftarrow$ k\;
                cur $\leftarrow$ i\;
            }
        }
    }
    \BlankLine
    \uIf{cur is 0}{
        v.push(v[0])\;
        v.erase(v.begin())\;
    }
    \Else{
        reverse(v, 0 to cur-1)\;
        reverse(v, cur to end)\;
    }
    return nv\;
}
\end{algorithm}

\newpage

\section{Rerouting}
\label{appendix:rerouting}

\begin{algorithm}[!htbp] \label{algReroute}
\small
\DontPrintSemicolon
\SetKwInput{KwInput}{Input} 
\SetNoFillComment
\SetKwFunction{fReroute}{Reroute}
\SetKwProg{Fn}{void}{:}{}
\SetKwProg{Int}{int}{:}{}

\Int{\fReroute{path, isLast}}{
   vis[path[0]] $\leftarrow$ true\;
  Q.pushFront(path.back())\;
  \While{Q is not empty}{
     v $\leftarrow$ Q.front()\;
     Q.popFront()\;
     \BlankLine
     \uIf{!isLast and deg[v] $>$ 0}{return v}
     \lIf{vis[v]}{continue}
     vis[v] $\leftarrow$ true\;
     \BlankLine
     dist $\leftarrow$ d[v].f, cur $\leftarrow$ d[v].s\;
     F $\leftarrow$ cur[0], L $\leftarrow$ cur.back()\;
     \BlankLine
     
     cyc $\leftarrow$ isCycle(cur, isLast)\;
     \uIf{cyc is not -1}{
        path $\leftarrow$ cur\;
        return cyc\;
     }
     \BlankLine
      \For{$i\gets1$ \KwTo cur.size()-1}{
        nv $\leftarrow$ cur[i]\;
        \uIf{dist $<$ d[nv].f}{
            \tcp{Traveling to an already unbounded vertex}
            \uIf{i $\ge$ 1 and i $<$ cur.size()-1 and UB[nv]}{
                \uIf{!isLast and A[f][nv]}{
                    cur.pushFront(nv)\;
                    d[nv] $\leftarrow$ (dist, cur), q.pushFront(nv)\;
                    cur.popFront();
                    continue\;
                }
                \ElseIf{A[nv][L]}{
                    cur.pushBack(nv)\;
                    d[nv] $\leftarrow$ (dist, cur), q.pushFront(nv)\;
                    cur.popBack()\;
                    continue\;
                }
            }
            \tcp{Rotating around last vertex}
            \uIf{i $\ge$ 2 and i $<$ cur.size()-1 A[cur[i-1]][L]}{
                reverse(cur, i to end)\;
                d[nv] $\leftarrow$ (dist, cur), q.pushFront(nv)\;
                reverse(cur, i to end)\;
                continue\;
            } 
            \tcp{Rotating around start vertex}
            \uIf{!isLast and i $\ge$ 1 and i $<$ cur.size()-2 and A[cur[i+1]][F]}{
                reverse(cur, start to i)\;
                d[nv] $\leftarrow$ (dist, cur), q.pushFront(nv)\;
                reverse(cur, start to i)\;
                continue\;
            }
        }
        \tcp{Traveling to bounded vertex and making it unbounded}
        \uIf{dist+1 $<$ d[nv].f and i $\ge$ 1 and i $<$ cur.size()-1}{
            \uIf{!isLast and A[F][nv]}{
                reverse(cur), cur.pushBack(nv)\;
                d[nv] $\leftarrow$ (dist+1, cur), q.pushBack(nv)\;
                cur.popBack(), reverse(cur)\;
            }
            \ElseIf{A[nv][L]}{
                cur.pushBack(nv)\;
                d[nv] $\leftarrow$ (dist+1, cur), q.pushBack(nv)\;
                cur.popBack()\;
            }
        }
      }
  }
}
\end{algorithm}
isLast is true when reroute is used to convert from a k-Unbounded Hamiltonian Path to Cycle
\newpage
\section{Converting k-Unbounded Hamiltonian Path to Cycle}
\label{appendix:pathToCycle}

\begin{algorithm}[!htbp] \label{algReroute2}
\footnotesize
\DontPrintSemicolon
\SetKwInput{KwInput}{Input} 
\SetNoFillComment
\SetKwFunction{fpathToCycle}{pathToCycle}
\SetKwProg{Fn}{void}{:}{}

\Fn{\fpathToCycle{path}}{
  F $\leftarrow$ min(path[0], path.back()), L $\leftarrow$ max(path[0], path.back())\;
  D[F][L] $\leftarrow$ (0, path)\;
  Q.pushFront(min(path[0], path.back()), max(path[0], path.back()))\;
  \While{Q is not empty}{
     p $\leftarrow$ Q.front()\;
     Q.popFront()\;
     \BlankLine
     \lIf{vis[p.f][p.s]}{continue}
     vis[p.f][p.s] $\leftarrow$ 1\;
     dist $\leftarrow$ D[p.f][p.s].f, cur $\leftarrow$ D[p.f][p.s].s\;
     F $\leftarrow$ cur[0], L  $\leftarrow$ cur.back()\;
     \BlankLine
     cyc $\leftarrow$ isCycle(cur, true)\;
     \uIf{cyc is not -1}{
        path $\leftarrow$ cur\;
        return\;
     }
     \tcc{Traveling to an already unbounded vertex}
     \For{$i\gets1$ \KwTo cur.size()-2}{
        nv $\leftarrow$ cur[i]\;
        \lIf{!UB[nv]}{continue}
        \uIf{A[F][nv] and dist $<$ D[min(L, nv)][max(L, nv)].f}{
            cur.pushFront(nv)\;
            tP $\leftarrow$ (min(L, nv), max(L, nv))\;
            D[tP.f][tP.s] $\leftarrow$ (dist, cur), Q.pushFront(tP.f, tP.s)\;
            cur.popFront()\;
        }
        \ElseIf{A[nv][L] and dist $<$ D[min(F, nv)][max(F, nv)].f}{
            cur.pushBack(nv)\;
            tP $\leftarrow$ (min(F, nv), max(F, nv))\;
            D[tP.f][tP.s] $\leftarrow$ (dist, cur), Q.pushFront(tP.f, tP.s)\;
            cur.popBack()\;
        }
     }
     \tcc{Rotating around last vertex}
     \For{$i\gets2$ \KwTo cur.size()-2}{
        nv $\leftarrow$ cur[i], P $\leftarrow$ cur[i-1]\;
        \uIf{A[P][L] and dist $<$ D[min(F, nv)][max(F, nv)].f}{
            reverse(cur, i to end)\;
            tP $\leftarrow$ (min(F, nv), max(F, nv))\;
            D[tP.f][tP.s] $\leftarrow$ (dist, cur), Q.pushFront(tP.f, tP.s)\;
            reverse(cur, i to end)\;
        }
     }
     \tcc{Rotating around start vertex}
     \For{$i\gets1$ \KwTo cur.size()-3}{
        nv $\leftarrow$ cur[i], W $\leftarrow$ cur[i+1]\;
        \uIf{A[W][F] and dist $<$ D[min(L, nv)][max(L, nv)].f}{
            reverse(cur, start to i)\;
            tP $\leftarrow$ (min(L, nv), max(L, nv))\;
            D[tP.f][tP.s] $\leftarrow$ (dist, cur), Q.pushFront(tP.f, tP.s)\;
            reverse(cur, start to i)\;
        }
     }
     \tcc{Traveling to bounded vertex and making it unbounded}
     \For{$i\gets1$ \KwTo cur.size()-2}{
        nv $\leftarrow$ cur[i]\;
        \lIf{UB[nv]}{continue}
        \uIf{A[F][nv] and dist+1 $<$ D[min(L, nv)][max(L, nv)].f}{
            cur.pushFront(nv)\;
            tP $\leftarrow$ (min(L, nv), max(L, nv))\;
            D[tP.f][tP.s] $\leftarrow$ (dist+1, cur), Q.pushBack(tP.f, tP.s)\;
            cur.popFront()\;
        }
        \ElseIf{A[nv][L] and dist+1 $<$ D[min(F, nv)][max(F, nv)].f}{
            cur.pushBack(nv)\;
            tP $\leftarrow$ (min(F, nv), max(F, nv))\;
            D[tP.f][tP.s] $\leftarrow$ (dist+1, cur), Q.pushFront(tP.f, tP.s)\;
            cur.popBack()\;
        }
     }
  }
}

\end{algorithm}

To save memory, only $D[i][j]$ where $i\le j$ are considered, hence the use of the min and max operators.

\end{appendices}

\end{document}